\pdfoutput=1
\documentclass[acmsmall]{acmart}

\AtBeginDocument{%
  }

\setcopyright{none}
\copyrightyear{2025}
\acmYear{2025}
\acmDOI{XXXXXXX.XXXXXXX}

\acmJournal{PACMPL}
\acmVolume{1}
\acmNumber{1}
\acmArticle{1}
\acmMonth{1}

\usepackage[noend]{algpseudocode}
\usepackage{algorithm}
\usepackage{wrapfig}
\usepackage{url}
\usepackage{subcaption}
\usepackage{comment}
\usepackage{multirow}
\usepackage{enumitem}
\usepackage{mdframed}
\usepackage{hyperref}
\usepackage{syntax}
\usepackage{soul}
\usepackage{xspace}
\usepackage{bbding}
\usepackage{quiver}
\usepackage{graphics}
\usepackage{graphicx}
\usepackage{subcaption}
\usepackage[rightcaption]{sidecap}
\usepackage{pgfplots}
\usepackage{tikz} 
\usetikzlibrary{arrows.meta}
\usepgflibrary{shapes.multipart} 
\usepgflibrary[shapes.multipart] 
\usetikzlibrary{shapes.multipart} 
\usetikzlibrary[shapes.multipart] 
\usetikzlibrary{shapes,positioning,calc}
\colorlet{lightgray}{gray!20}

\newcommand{\yang}[1]{{\color{orange}Yang: #1}}

\newcommand{\highlight}[1]{{\color{red}#1}}

\newcommand{\trans}{database transformer\xspace}

\newcommand{\irule}[2]%
   {\mkern-2mu\displaystyle\frac{#1}{\vphantom{,}#2}\mkern-2mu}
\newcommand{\irulelabel}[3]
{
\mkern-2mu
\begin{array}{ll}
\displaystyle\frac{#1}{\vphantom{,}#2} & \!\!\!\!\!\! #3
\end{array}
\mkern-2mu
}

\newcommand{\bfpara}[1]{\vspace{5pt} \noindent {\bf \emph{#1}}}
\newcommand{\textbox}[1]{\vspace{5pt} \fbox{\parbox{.95\linewidth}{#1}}}

\newcommand{\tool}{\textsc{Graphiti}\xspace}
\newcommand{\mediator}{\textsc{Mediator}\xspace}

\newcommand{\verieql}{\textsc{VeriEQL}\xspace}

\newcommand{\dynamite}{\textsc{Dynamite}\xspace}
\newcommand{\qex}{\textsc{Qex}\xspace}
\newcommand{\cosette}{\textsc{Cosette}\xspace}
\newcommand{\opentranspiler}{\textsc{OpenCypherTranspiler}\xspace}
\newcommand{\kuzu}{\textsc{Kuzu}\xspace}

\def\ojoin{\setbox0=\hbox{$\bowtie$}%
  \rule[.17ex]{.25em}{.6pt}\llap{\rule[.88ex]{.25em}{.6pt}}}
\def\leftouterjoin{\mathbin{\ojoin\mkern-5.8mu\bowtie}}
\def\rightouterjoin{\mathbin{\bowtie\mkern-5.8mu\ojoin}}
\def\fullouterjoin{\mathbin{\ojoin\mkern-5.8mu\bowtie\mkern-5.8mu\ojoin}}

\newcommand{\doubleplus}{+\!+~}

\newcommand{\transpattern}{\overset{\textsf{pattern}}{\longrightarrow}}
\newcommand{\transpred}{\overset{\textsf{pred}}{\longrightarrow}}
\newcommand{\transexpr}{\overset{\textsf{expr}}{\longrightarrow}}
\newcommand{\transquery}{\overset{\textsf{query}}{\longrightarrow}}
\newcommand{\transclause}{\overset{\textsf{clause}}{\longrightarrow}}

\newcommand{\fresh}{\textsf{fresh}\xspace}

\newcommand{\ttMatch}{\texttt{MATCH}\xspace}

\newcommand{\boldSelect}{\texttt{\textbf{SELECT}}\xspace}
\newcommand{\boldFrom}{\texttt{\textbf{FROM}}\xspace}
\newcommand{\boldWhere}{\texttt{\textbf{WHERE}}\xspace}
\newcommand{\boldAs}{\texttt{\textbf{AS}}\xspace}
\newcommand{\boldSum}{\texttt{\textbf{Sum}}\xspace}
\newcommand{\boldCount}{\texttt{\textbf{Count}}\xspace}

\newcommand{\boldJoin}{\texttt{\textbf{JOIN}}\xspace}
\newcommand{\boldDistinct}{\texttt{\textbf{DISTINCT}}\xspace}
\newcommand{\boldIsNull}{\texttt{\textbf{IS NULL}}\xspace}
\newcommand{\boldIsNotNull}{\texttt{\textbf{IS NOT NULL}}\xspace}

\newcommand{\boldLeftJoin}{\texttt{\textbf{LEFT JOIN}}\xspace}
\newcommand{\boldRightJoin}{\texttt{\textbf{RIGHT JOIN}}\xspace}
\newcommand{\boldOn}{\texttt{\textbf{ON}}\xspace}
\newcommand{\boldIn}{\texttt{\textbf{IN}}\xspace}
\newcommand{\boldGroupBy}{\texttt{\textbf{GROUP BY}}\xspace}
\newcommand{\boldOrderBy}{\texttt{\textbf{ORDER BY}}\xspace}
\newcommand{\boldDesc}{\texttt{\textbf{DESC}}\xspace}

\newcommand{\boldWith}{\texttt{\textbf{WITH}}\xspace}
\newcommand{\boldAnd}{\texttt{\textbf{AND}}\xspace}
\newcommand{\boldMatch}{\texttt{\textbf{MATCH}}\xspace}
\newcommand{\boldOptMatch}{\texttt{\textbf{OPTIONAL MATCH}}\xspace}
\newcommand{\boldReturn}{\texttt{\textbf{RETURN}}\xspace}
\newcommand{\boldExists}{\texttt{\textbf{EXISTS}}\xspace}

\newcommand{\sfMap}{\mathsf{map}\xspace}
\newcommand{\sfFoldl}{\mathsf{foldl}\xspace}
\newcommand{\sfZip}{\mathsf{zip}\xspace}
\newcommand{\sfHead}{\mathsf{head}\xspace}

\newcommand{\sfLast}{\mathsf{last}\xspace}

\newcommand{\sfFilter}{\mathsf{filter}\xspace}

\newcommand{\sfGraph}{\mathsf{Graph}\xspace}
\newcommand{\sfTable}{\mathsf{Table}\xspace}
\newcommand{\sfDedup}{\mathsf{dedup}\xspace}
\newcommand{\sfHasAgg}{\mathsf{hasAgg}\xspace}
\newcommand{\sfIte}{\mathsf{ite}\xspace}

\newcommand{\sfReturn}{\mathsf{Return}\xspace}
\newcommand{\sfOrderBy}{\mathsf{OrderBy}\xspace}
\newcommand{\sfUnion}{\mathsf{Union}\xspace}
\newcommand{\sfUnionAll}{\mathsf{UnionAll}\xspace}
\newcommand{\sfMatch}{\mathsf{Match}\xspace}
\newcommand{\sfOptMatch}{\mathsf{OptMatch}\xspace}
\newcommand{\sfOptionalMatch}{\mathsf{Optional Match}\xspace}
\newcommand{\sfWith}{\mathsf{With}\xspace}

\newcommand{\sfAgg}{\mathsf{Agg}\xspace}
\newcommand{\sfCast}{\mathsf{Cast}\xspace}
\newcommand{\sfIsNull}{\mathsf{IsNull}\xspace}
\newcommand{\sfExists}{\mathsf{Exists}\xspace}
\newcommand{\sfCount}{\mathsf{Count}\xspace}
\newcommand{\sfAvg}{\mathsf{Avg}\xspace}
\newcommand{\sfMin}{\mathsf{Min}\xspace}
\newcommand{\sfMax}{\mathsf{Max}\xspace}
\newcommand{\sfSum}{\mathsf{Sum}\xspace}

\newcommand{\sfNull}{\mathsf{Null}\xspace}

\newcommand{\sfMerge}{\mathsf{merge}\xspace}
\newcommand{\ttMerge}{\texttt{merge}\xspace}
\newcommand{\sfNullify}{\mathsf{Nullify}\xspace}

\newcommand{\ttWith}{\texttt{With}\xspace}
\newcommand{\ttGroupBy}{\texttt{GroupBy}\xspace}
\newcommand{\ttOrderBy}{\texttt{OrderBy}\xspace}
\newcommand{\ttCast}{\texttt{Cast}\xspace}
\newcommand{\ttIsNull}{\texttt{IsNull}\xspace}
\newcommand{\ttAgg}{\texttt{Agg}\xspace}
\newcommand{\ttCount}{\texttt{Count}\xspace}
\newcommand{\ttAvg}{\texttt{Avg}\xspace}
\newcommand{\ttMin}{\texttt{Min}\xspace}
\newcommand{\ttMax}{\texttt{Max}\xspace}
\newcommand{\ttSum}{\texttt{Sum}\xspace}
\newcommand{\ttHaving}{\texttt{Having}\xspace}

\newcommand{\sffoldl}{\mathsf{foldl}\xspace}

\newcommand{\sfzip}{\mathsf{zip}\xspace}
\newcommand{\sfmap}{\mathsf{map}\xspace}

\newcommand{\sfwhere}{\mathsf{where}\xspace}

\newcommand{\sfNodes}{\mathsf{Nodes}\xspace}

\newcommand{\sfite}{\mathsf{ite}\xspace}

\newcommand{\convert}{\mathcal{C}\xspace}

\newcommand{\pattern}{{PP\xspace}}
\newcommand{\nodepattern}{{{NP}}\xspace}
\newcommand{\edgepattern}{{{EP}}\xspace}

\newcommand{\stackoverflow}{\texttt{StackOverflow}\xspace}
\newcommand{\tutorial}{\texttt{Tutorial}\xspace}
\newcommand{\academic}{\texttt{Academic}\xspace}
\newcommand{\verieqlset}{\texttt{VeriEQL}\xspace}
\newcommand{\mediatorset}{\texttt{Mediator}\xspace}
\newcommand{\gpt}{\texttt{GPT-Translate}\xspace}

\definecolor{deepred}{RGB}{228,26,28}
\definecolor{deepblue}{RGB}{55,126,184}
\definecolor{deepgreen}{RGB}{77,175,74}
\definecolor{deeppurple}{RGB}{152,78,163}
\definecolor{deeporange}{RGB}{255,127,0}

\newcommand{\set}[1]{\{ #1 \}}
\newcommand{\denot}[1]{\llbracket #1 \rrbracket}
\newcommand{\transformer}{\Phi}

\newcommand{\ijoin}{\bowtie}
\newcommand{\ljoin}{\leftouterjoin}

\newcommand{\fjoin}{\fullouterjoin}
\newcommand{\joinop}{\otimes}
\newcommand{\arithop}{\oplus}
\newcommand{\logicop}{\odot}
\newcommand{\pred}{\phi}
\newcommand{\proj}{\Pi}
\newcommand{\filter}{\sigma}
\newcommand{\rename}{\rho}
\newcommand{\Xset}{\mathcal{X}}

\newcommand{\stype}{t_\mathsf{src}}
\newcommand{\ttype}{t_\mathsf{tgt}}
\newcommand{\nodeType}{t_\mathsf{node}}
\newcommand{\edgeType}{t_\mathsf{edge}}
\newcommand{\lbl}{\mathsf{label}}

\newcommand{\schema}{\Psi}

\newcommand{\keys}{\mathsf{keys}}
\newcommand{\gschema}{\Psi_G}
\newcommand{\rschema}{\Psi_R}
\newcommand{\graph}{G}
\newcommand{\conform}{\triangleright}
\newcommand{\dstType}{\mathsf{dstType}}
\newcommand{\srcType}{\mathsf{srcType}}
\newcommand{\tbleq}{\equiv}
\newcommand{\tables}{S}
\newcommand{\integrity}{\xi}
\newcommand{\smapping}{\Lambda}
\newcommand{\sdt}{\transformer_\mathsf{sdt}}
\newcommand{\rdt}{\transformer_\mathsf{rdt}}
\newcommand{\sdtname}{SDT\xspace}
\newcommand{\rdtname}{RDT\xspace}
\newcommand{\reldb}{R}

\newcommand{\fk}{\mathsf{fk}}

\newcommand{\PK}{\mathsf{PK}}
\newcommand{\FK}{\mathsf{FK}}

\newcommand{\attributes}{\mathsf{Attrs}}

\newcommand{\tuples}{\mathcal{T}}

\title{\tool: Bridging Graph and Relational Database Queries}

\author{Yang He}
\orcid{0009-0007-7755-3112}
\affiliation{%
  \institution{Simon Fraser University}
  \city{Burnaby}
  \country{Canada}
}
\email{yha244@sfu.ca}

\author{Ruijie Fang}
\orcid{0000-0001-5348-5468}
\affiliation{%
  \institution{University of Texas at Austin}
  \city{Austin, TX}
  \country{USA}
}
\email{ruijief@cs.utexas.edu}

\author{Isil Dillig}
\orcid{0000-0001-8006-1230}
\affiliation{%
  \institution{University of Texas at Austin}
  \city{Austin, TX}
  \country{USA}
}
\email{isil@cs.utexas.edu}

\author{Yuepeng Wang}
\orcid{0000-0003-3370-2431}
\affiliation{%
  \institution{Simon Fraser University}
  \city{Burnaby}
  \country{Canada}
}
\email{yuepeng@sfu.ca}

\renewcommand{\shortauthors}{Yang He, Ruijie Fang, Işıl Dillig, and Yuepeng Wang}

\begin{abstract}

This paper presents an automated reasoning technique for checking equivalence between graph database queries written in Cypher and relational  queries in SQL. To formalize a suitable notion of equivalence in this setting, we introduce the concept of \emph{database transformers}, which transform database instances between graph and relational models. We then propose a novel verification methodology that checks equivalence modulo a given transformer by reducing the original problem to verifying equivalence between a pair of SQL queries. This reduction is achieved by embedding a subset of Cypher into SQL through syntax-directed translation, allowing us to leverage existing research on automated reasoning for SQL while obviating the need for reasoning simultaneously over two different data models. We have implemented our approach in a tool called \tool and used it to check equivalence between graph and relational queries. Our experiments demonstrate that \tool is useful both for verification and refutation and that it can uncover subtle bugs, including those found in Cypher tutorials and academic papers.

\end{abstract}



\begin{document}

\maketitle

\section{Introduction} \label{sec:intro}


Over the past decades, \emph{graph} databases have garnered significant attention from both industry and academia, offering more flexible data models with different trade-offs compared to \emph{relational} databases. As a result, developers are increasingly interested in migrating relational database applications to graph databases~\cite{comparison-10}, and systems like Apache Age~\cite{age-web24} aim to  incorporate graph components into relational databases to help with this transition.

Nevertheless, transitioning between relational and graph databases often requires developers to convert queries from one model to the other, and, as evidenced by numerous posts on online forums~\cite{post1-web24,post2-web24,post3-web24}, translating between relational  and graph queries can be quite challenging  due to misunderstandings of joins versus relationships, aggregation semantics, and other complexities.
In fact, we have identified multiple incorrect translations in existing literature where queries claimed to be equivalent were, in reality, non-equivalent~\cite{Lin-arxiv16, neo4jdocs-issue, stackoverflow-issue}.
This underscores a growing need for rigorous reasoning about the equivalence of graph and relational queries.

Although significant progress has been made in verifying the equivalence of relational database queries~\cite{cosette, chu2017hottsql, chu2018axiomatic, mediator-popl18}, to the best of our knowledge, no existing work addresses the equivalence verification problem between relational and graph queries. A key challenge in this area arises from the distinct data models of relational and graph databases. Relational databases organize data in tables, with queries operating on rows and columns through well-defined operations like joins and aggregations. In contrast, graph databases represent data as nodes and edges, with queries typically expressing relationships and traversals through graph structures. This fundamental difference in data representation complicates both the definition of equivalence between graph and relational queries and the verification process itself.

This paper takes a first step towards developing automated  
reasoning techniques that can be used to check equivalence between relational queries written in SQL and graph database queries implemented in Cypher~\cite{cypher-sigmod18}, the most popular graph database query language. Our formal reasoning technique is built on a novel method that embeds a subset of the Cypher language into SQL, building on the insight that paths in a graph database instance correspond  to joins of rows in a relational database. This observation not only allows us to translate Cypher queries to SQL in a syntax-directed way but also facilitates checking equivalence between Cypher and SQL queries.  At a high level, our approach hinges on three crucial components: {\bf (1)} a formal foundation for defining equivalence between graph and relational database instances (over arbitrary schemas); {\bf (2)} a correct-by-construction technique for translating graph database queries  to relational queries (over a specific schema); and {\bf (3)} a novel verification methodology that leverages (2) to establish equivalence between Cypher and SQL queries that operate over any arbitrary schema. We next explain each of our contributions in more detail.

\bfpara{Formal foundation for graph and relational database equivalence.} As mentioned earlier, a key challenge in reasoning about equivalence between graph and relational queries is the lack of a straightforward mapping between the data models.  To address this, we introduce the concept of \emph{\trans}, adapted from prior work on schema mappings~\cite{clio-vldb00,clio-cmfa09,dynamite-vldb20}, which allows transforming a database instance from one data model (graphs) to an equivalent instance in another model (relational databases). This transformation forms the foundation for defining equivalence between graph and relational database queries.

\bfpara{Correct-by-construction transpilation.} Building on this notion of database transformer, we introduce the concept of a \emph{standard \trans (\sdtname)} as the default correctness specification. In simple terms, the SDT provides a set of transformation rules to map graph elements (nodes and edges) into relational tables, maintaining the structure and semantics of the graph within the relational model. For example, nodes in a graph schema are transformed into tables in the relational schema, where attributes of the node become columns, and edges are represented as relationships between these tables with foreign keys. The resulting relational schema is referred to as the \emph{induced relational schema}. Our method then defines syntax-directed transpilation rules to convert any Cypher query into a SQL query over this induced schema.
The core insight is that Cypher path queries, which perform pattern matching over subgraphs, can be mapped to relational joins. However, the actual translation is tricky due to the fact that Cypher supports flexible, multi-step pattern matching over  graph structures, which demands careful handling to ensure that pattern matching in Cypher—whether simple or complex—is accurately represented as joins in SQL, preserving the semantics of the original graph query.

\bfpara{Equivalence checking methodology for arbitrary schema.} While the transpilation method described above can generate an equivalent SQL query, the equivalence is \emph{modulo} the \sdtname. However, in practice, the target relational database often uses a different schema (rather than the \emph{induced relational schema}), so the syntax-directed transpilation method alone is insufficient. To address this gap, we propose a verification methodology that checks equivalence between graph and relational queries modulo any \trans.

The key insight of our approach is that, instead of directly reasoning about equivalence between Cypher and SQL queries—which would require complex SMT encodings that combine graph and relational structures—we reduce the problem to checking equivalence between SQL queries over different schemas. This reduction allows us to leverage existing techniques and tools for SQL equivalence checking~\cite{mediator-popl18,verieql-oopsla24}, avoiding the need for handling the intricate challenge of reasoning over two fundamentally different data models at the same time.

\begin{wrapfigure}{r}{0.5\textwidth}
\vspace{-10pt}
\centering
\hspace{-11pt}
\includegraphics[scale=0.24]{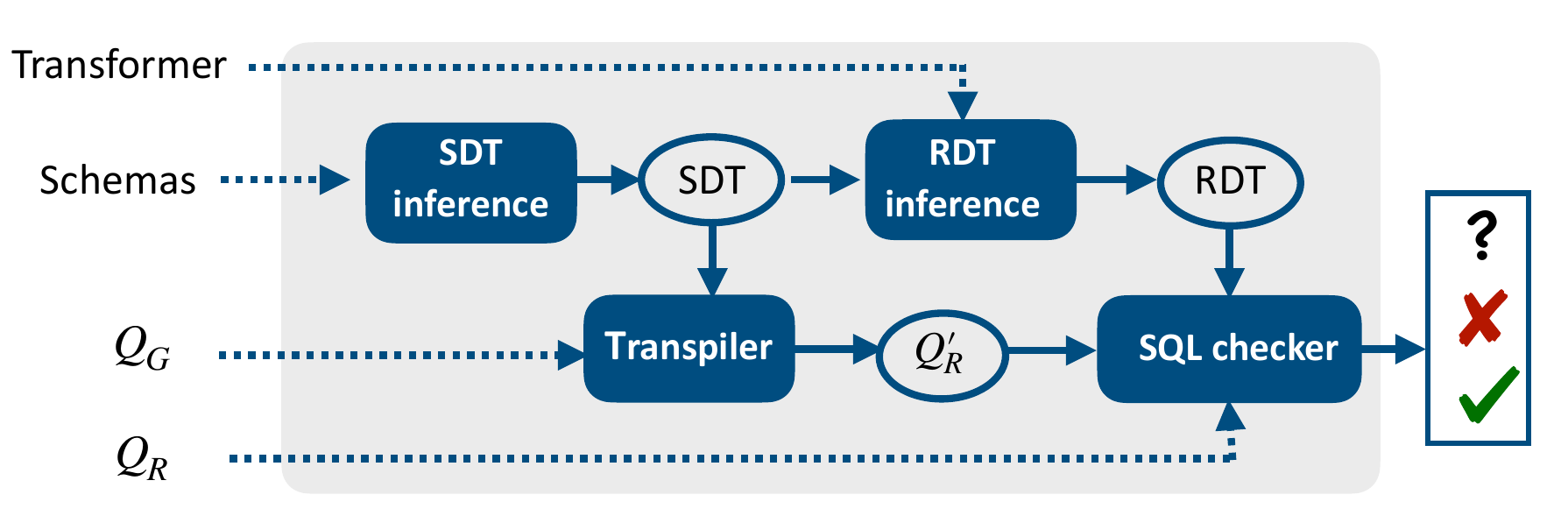}
\vspace{-20pt}
\caption{Overview of approach.}
\label{fig:eq-approach}
\vspace{-10pt}
\end{wrapfigure}


Figure~\ref{fig:eq-approach} illustrates the proposed approach for checking equivalence between Cypher and SQL queries. Our method takes four inputs: (1) a Cypher query \( Q_G \), (2) an SQL query \( Q_R \), (3) schemas for the graph and relational databases, and (4) a correctness specification \( \transformer \) in the form of a \trans. To establish equivalence between \( Q_G \) and \( Q_R \) modulo \( \transformer \), the method first derives the induced relational schema and the \sdtname. It then applies a correct-by-construction transpilation technique to generate a SQL query \( Q_R' \) that is provably equivalent to \( Q_G \) modulo the SDT (but not modulo \( \transformer \)). In the final step, the method computes a semantic "diff" between the induced and target relational schemas, constructing a \emph{residual \trans} to align the two relational instances. Finally, an off-the-shelf SQL equivalence checker is then used to check  equivalence between \( Q_R' \) and \( Q_R \).


Overall, our proposed methodology has two key advantages. First, when the user just wants to translate a Cypher query to SQL but does not care about the underlying relational schema (or if the desired  schema is the same as the induced relational schema), our method can be used to perform correct-by-construction transpilation. Second, given any arbitrary correctness specification (in the form of a \trans), our method  can leverage a combination of the proposed transpilation approach and existing automated reasoning tools for SQL to reason about equivalence between any pair of Cypher and SQL queries. 

We have implemented the proposed approach in a new tool called \tool for reasoning about equivalence between Cypher and SQL queries and conducted an extensive experimental evaluation of \tool on 410 benchmarks. 
These include 45 benchmarks sourced from public platforms such as StackOverflow, tutorials, and academic papers; 160 benchmarks translated from SQL by students with Cypher experience; and 205 benchmarks translated using ChatGPT. In the first evaluation, we combine \tool with the \verieql~\cite{verieql-oopsla24} bounded model checker for SQL, performing equivalence verification on all 410 query pairs. This reveals equivalence violations in 34 benchmarks, including 3 from the wild, 4 from manual translations, and 27 from GPT-generated translations.
In the second evaluation, we pair \tool with the deductive SQL verifier \mediator~\cite{mediator-popl18} for full-fledged verification of an aggregation-free subset. This enables unbounded equivalence verification between Cypher and SQL queries, where both tables and graphs can have arbitrary sizes. Here, about 80\% of supported queries are verified as equivalent in a push-button manner.
Finally, in the third experiment, we show that \tool can generate SQL queries that are competitive with manually-written ones in terms of execution efficiency.


\bfpara{Contributions.}
To summarize, this paper makes the following key contributions:
\vspace{-10pt}
\begin{itemize}[leftmargin=*]
\item We propose the first technique for reasoning about equivalence between graph and relational queries, based on a formal definition of equivalence modulo \trans. 
\item We introduce the concept of \emph{standard \trans}, which acts as the default correctness specification for equivalence between Cypher and SQL queries.
\item We develop a sound and complete transpilation technique that translates a subset of Cypher queries into equivalent SQL queries, guaranteeing that both queries produce the same results under \emph{standard \trans}.
\item We show how to leverage the proposed transpilation approach to reduce the equivalence checking problem (modulo \emph{any} \trans) into the problem of checking equivalence between a pair of SQL queries,  allowing us to leverage existing work on automated reasoning for SQL.  
\item {We implement these ideas in a tool called \tool and conduct an empirical evaluation on 410 benchmarks, showing that our  approach can be used for both verification and falsification.}
\end{itemize}

\section{Motivating Example} \label{sec:overview}

In this section, {we motivate the problem addressed in this work} through an example from prior work~\cite{Lin-arxiv16} that studies SQL analytics for graphs.  
{The  queries  in this section pertain to a real-world biomedical research database~\cite{SemMedDB-web24}.

\begin{figure}[!t]
\centering

\begin{subfigure}{0.45\linewidth}
\centering
\scalebox{0.6}{
\begin{tikzpicture}[main/.style = {draw=black, circle, thin, minimum size=0.8cm}] 

\node[main, draw=deepred] (1) at (0, 0){};
\node[main, draw=deepblue, dashed] (2) at (2.5, 0){};
\node[main, draw=deepgreen, double,double distance=1pt] (3) at (5, 0){};

\node (CONCEPT) at (0, 0.7){\textbf{CONCEPT}};
\node (CONCEPT) at (2.5, 0.7){\textbf{PA}};
\node (CONCEPT) at (5, 0.7){\textbf{SENTENCE}};

\draw[->, line width=1.2pt] (1) -- (2) node [midway, above, sloped] (4) {\small \textbf{:CS}};
\draw[->, line width=1.2pt] (2) -- (3) node [midway, above, sloped] (5) {\small \textbf{:SP}};

\end{tikzpicture}
}
\caption{Graph schema. 
CONCEPT has two property keys: CID and Name. PA has two property keys: PID and CSID. SENTENCE has two property keys: SID and PMID.}
\label{fig:example-graph-schema}
\end{subfigure}
\hspace{10pt}
\begin{subfigure}{0.5\linewidth}
\centering
\scalebox{0.6}{
\begin{tikzpicture}[relation/.style={rectangle split, rectangle split parts=#1, rectangle split part align=base, draw, anchor=center, align=center, text height=3mm, text centered}]\hspace*{-0.3cm}

\node (Concept) {\textbf{Concept}};
\node [left=-3.cm of Concept] (Cs) {\textbf{Cs}};
\node [left=-3.cm of Cs] (Pa) {\textbf{Pa}};
\node [left=-3.cm of Pa] (Sp) {\textbf{Sp}};
\node [left=-4.5cm of Sp] (Sentence) {\textbf{Sentence}};

\node [relation=2, rectangle split horizontal, rectangle split part fill={lightgray!50}, anchor=north west, below=0.7cm of Concept.west, anchor=west] (ConceptRow)
{\nodepart{one} \textbf{CID}
\nodepart{two} NAME};

\node [relation=2, rectangle split horizontal, rectangle split part fill={lightgray!50}, anchor=north west, below=0.7cm of Cs.west, anchor=west] (CsRow)
{\nodepart{one} \textbf{CSID}
\nodepart{two} CID};

\node [relation=2, rectangle split horizontal, rectangle split part fill={lightgray!50}, anchor=north west, below=0.7cm of Pa.west, anchor=west] (PaRow)
{\nodepart{one} \textbf{PID}
\nodepart{two} \underline{CSID}};

\node [relation=3, rectangle split horizontal, rectangle split part fill={lightgray!50}, anchor=north west, below=0.7cm of Sp.west, anchor=west] (SpRow)
{\nodepart{one} \textbf{SPID}
\nodepart{two} \underline{SID}
\nodepart{three} \underline{PID}};

\node [relation=2, rectangle split horizontal, rectangle split part fill={lightgray!50}, anchor=north west, below=0.7cm of Sentence.west, anchor=west] (SentenceRow)
{\nodepart{one} \textbf{SID}
\nodepart{two} PMID};

\draw[-latex] (CsRow.two north) -- ++(0,0) -| ($(CsRow.two north) + (0,0)$) |- ($(ConceptRow.one north) + (0,0.2)$) -| ($(ConceptRow.one north) + (0,0)$);

\draw[-latex] (PaRow.two south) -- ++(0,0) -| ($(PaRow.two south) + (0,0)$) |- ($(CsRow.one south) + (0,-0.2)$) -| ($(CsRow.one south) + (0,0)$);

\draw[-latex] (SpRow.three north) -- ++(0,0) -| ($(SpRow.three north) + (0,0)$) |- ($(PaRow.one north) + (0,0.2)$) -| ($(PaRow.one north) + (0,0)$);

\draw[-latex] (SpRow.two south) -- ++(0,0) -| ($(SpRow.two south) + (0,0)$) |- ($(SentenceRow.one south) + (0,-0.2)$) -| ($(SentenceRow.one south) + (0,0)$);
\end{tikzpicture}
}

\caption{Relational schema. Primary keys are in bold. Foreign keys are underlined and connected to references with arrows.}
\label{fig:example-relational-schema}
\hspace{10pt}
\end{subfigure}

\vspace{-10pt}
\caption{A pair of relational and graph schemas.}
\label{fig:example-schemas}
\vspace{-10pt}
\end{figure}
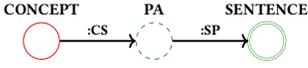
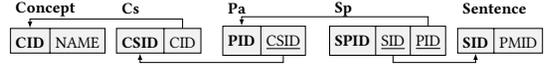

\bfpara{Incorrect translation.}
Figure~\ref{fig:example-graph-schema} shows a graph schema that contains three types of nodes called \texttt{CONCEPT}, \texttt{PA}\footnote{Here, \texttt{PA} stands for \emph{predication argument}}, and \texttt{SENTENCE}. The \texttt{CS} relationship (represented by an edge) links concepts to predication arguments, and the \texttt{SP} relationship (also represented as an edge) links predication arguments to sentences. On the other hand, Figure~\ref{fig:example-relational-schema} shows the relational representation of the graph model. As shown in Figure~\ref{fig:example-relational-schema}, the relational schema  contains five tables called \texttt{Concept}, \texttt{Cs}, \texttt{Pa}, \texttt{Sp}, \texttt{Sentence} that correspond to both the nodes and edges of the graph schema in Figure~\ref{fig:example-graph-schema}. For some intuition about the correspondence between the two databases,  Figure~\ref{fig:example-instances} shows sample instances of both database schemas that contain the same entries ``Atropine'' and ``Aspirin''.

\begin{figure}[!t]
\centering

\begin{minipage}{0.45\linewidth}
\centering
\begin{subfigure}{\linewidth}
\centering
\scalebox{0.7}{
\begin{tikzpicture}[main/.style = {draw, circle}] 

\node[circle, draw=deepred] (1) at (0, 2){\scalebox{.6}{Atropine}};
\node[circle, draw=deepred] (2) at (0, 0){\scalebox{.7}{Aspirin}};
\node[circle, draw=deepblue, dashed] (3) at (2.5, 2){\scalebox{.95}{PA$_0$}};
\node[circle, draw=deepblue, dashed] (4) at (2.5, 0){\scalebox{.95}{PA$_1$}};
\node[circle, draw=deepgreen, double,double distance=1pt] (5) at (4.5, 2){\scalebox{1.}{S$_0$}};
\node[circle, draw=deepgreen, double,double distance=1pt] (6) at (4.5, 0){\scalebox{1.}{S$_1$}};

\draw[->] (1) -- (3) node [midway, above, sloped] (cs1) {:CS};
\draw[->] (1) -- (4) node [midway, above, sloped] (cs4) {:CS};

\draw[->] (3) -- (5) node [midway, above, sloped] (cs5) {:SP};
\draw[->] (4) -- (5) node [midway, above, sloped] (cs6) {:SP};

\end{tikzpicture} 
}
\vspace{-5pt}
\caption{Graph database instance. Aspirin and Atropine denote two Concept nodes, PA$_0$ and PA$_1$ are two PA nodes, and S$_0$ and S$_1$ are two Sentence nodes.}
\label{fig:example-graph-instance}
\end{subfigure}
\end{minipage}
~
\begin{minipage}{0.5\linewidth}
\centering
\begin{subfigure}{\textwidth}
\centering
\scriptsize
\begin{subfigure}[b]{0.35\textwidth}
\centering
\begin{tabular}{|c|c|}
\hline
CID & NAME \\
\hline
1 & Atropine  \\
\hline
2 & Aspirin  \\
\hline
\end{tabular}
\vspace{-5pt}
\caption*{\texttt{Concept}}
\end{subfigure}
~
\begin{subfigure}[b]{0.4\textwidth}
\centering
\begin{tabular}{|c|c|}
\hline
CID & CSID \\
\hline
1 & 0 \\
\hline
1 & 1 \\
\hline
\end{tabular}
\vspace{-5pt}
\caption*{\texttt{Cs}}
\end{subfigure}
~
\begin{subfigure}[b]{0.3\textwidth}
\centering
\begin{tabular}{|c|c|}
\hline
PID & CSID \\
\hline
0 & 0 \\
\hline
1 & 1 \\
\hline
\end{tabular}
\vspace{-5pt}
\caption*{\texttt{Pa}}
\end{subfigure}

\vspace{5pt}

\begin{subfigure}[b]{0.4\textwidth}
\centering
\begin{tabular}{|c|c|c|}
\hline
SPID & SID & PID \\
\hline
0 & 0 & 0 \\
\hline
1 & 0 & 1 \\
\hline
\end{tabular}
\vspace{-5pt}
\caption*{\texttt{Sp}}
\end{subfigure}
~
\begin{subfigure}[b]{0.35\textwidth}
\centering
\begin{tabular}{|c|c|}
\hline
SID & PMID \\
\hline
0 & 0  \\
\hline
1 & 0  \\
\hline
\end{tabular}
\vspace{-5pt}
\caption*{\texttt{Sentence}}
\end{subfigure}

\vspace{-5pt}
\caption{Relational database instance.}
\label{fig:example-relational-instance}
\end{subfigure}
\end{minipage}

\vspace{-5pt}
\caption{Example graph and relational database instances.}
\label{fig:example-instances} 
\vspace{-5pt}
\end{figure}

Next, consider the SQL and Cypher queries shown in Figures~\ref{fig:example-sql-query} and~\ref{fig:example-cypher-query} respectively.
The SQL query aims to find all concepts that link to a concept \texttt{c1} with \texttt{CID} = 1 and their corresponding frequencies of connected paths to \texttt{c1} through the join of \texttt{Cs}, \texttt{Pa}, and \texttt{Sp} tables.
Similarly, the Cypher query first finds all sentences linked to $c_1$ through the \texttt{CS - PA - SP} path and then counts the frequencies of paths from those sentences to all concepts. 
According to~\citet{Lin-arxiv16}, these queries are intended to be equivalent; however, they are actually \emph{not} equivalent due to the subtle differences in Cypher and SQL semantics. 
At a high level, both queries explore how certain concepts (starting with \texttt{CID} = 1) are linked to other concepts via their occurrence in shared sentences. They do this by traversing relationships (either explicitly in the graph or via joins in the relational database) and aggregating the results to identify how frequently these connections occur. However, the two queries actually differ in how they compute the frequencies and end up providing different results on two database instances that are meant to contain the same data.

\begin{figure}[!t]
\centering
\small

\begin{subfigure}[t]{0.48\linewidth}
\centering
\scriptsize 
\begin{tabular}{l}
\boldSelect \texttt{c2.CID, \boldCount(*)} \boldFrom \texttt{Cs \boldAs c2, Pa \boldAs p2, Sp \boldAs s2} \\
\boldWhere \texttt{s2.PID = p2.PID \boldAnd p2.CSID = c2.CSID \boldAnd s2.SID \boldIn (} \\
\hspace{2.em} \texttt{\boldSelect s1.SID} 
\texttt{\boldFrom Cs \boldAs c1, Pa \boldAs p1, Sp \boldAs s1} \\
\hspace{2.em} \texttt{\boldWhere s1.PID = p1.PID \boldAnd p1.CSID = c1.CSID \boldAnd c1.CID = 1}
\texttt{)} \\
\boldGroupBy \texttt{CID} \\
\end{tabular}
\vspace{-5pt}
\caption{SQL query}
\label{fig:example-sql-query}
\end{subfigure}
\hfill
\begin{subfigure}[t]{0.48\linewidth}
\centering
\scriptsize
\begin{tabular}{|c|c|}
\hline
c2.CID & Count(*) \\
\hline
1 & 2 \\
\hline
\end{tabular}
\caption{The result of SQL query.}
\label{fig:example-sql-result}
\end{subfigure}

\vspace{10pt}

\begin{subfigure}[t]{0.48\linewidth}
\centering
\scriptsize 
\begin{tabular}{l}
\boldMatch \texttt{($c_1$:CONCEPT \{CID:1\})-[$r_1$:CS]->($p_1$:PA)-[$r_2$:SP]->($s$:SENTENCE)} \\
\boldWith \texttt{$s$} \\
\boldMatch \texttt{($s$:SENTENCE)<-[$r_3$:SP]-($p_2$:PA)<-[$r_4$:CS]-($c_2$:CONCEPT)} \\
\boldReturn \texttt{$c_2$.CID, \boldCount(*)} \\
\end{tabular}
\vspace{-5pt}
\caption{Cypher query}
\label{fig:example-cypher-query}
\end{subfigure}
\hfill
\begin{subfigure}[t]{0.48\linewidth}
\centering
\scriptsize
\begin{tabular}{|c|c|}
\hline
c2.CID & Count(*) \\
\hline
1 & 4 \\
\hline
\end{tabular}
\caption{The result of Cypher query.}
\label{fig:example-cypher-result}
\end{subfigure}

\vspace{-5pt}
\caption{A pair of SQL and Cypher queries with their execution results.}
\label{fig:example-queries-and-results}
\vspace{-10pt}
\end{figure}

In particular, when run on the database instances from Figure~\ref{fig:example-instances}, the SQL query produces the table shown in Figure~\ref{fig:example-sql-result} whereas the Cypher query produces the one in Figure~\ref{fig:example-cypher-result}.  These results agree on the \texttt{CID}'s of the related concepts; however they differ on the \emph{frequencies} of the relationships, as is evident from the entries in the \texttt{Count} column of these tables.\footnote{An equivalent Cypher query of the SQL query in Figure \ref{fig:example-sql-query} is shown in Appendix~\ref{sec:correct-cypher}.} As this example illustrates, queries that \emph{appear} to be ostensibly equivalent can have subtle differences in their semantics, motivating the need for automated reasoning tools that can be used to expose semantic differences between graph and relational queries.  In the remainder of this section, we elucidate some important aspects and design choices behind our proposed approach.

\begin{figure}[!t]
\footnotesize

\resizebox{.98\linewidth}{!}{
\hspace{-13pt}
\begin{minipage}{\linewidth}

\[
\begin{array}{rclrcl}
\texttt{CONCEPT}(\mathsf{cid}, \mathsf{name}) & \hspace{-10pt} \to \hspace{-10pt} & \texttt{Concept}(\mathsf{cid}, \mathsf{name}) & 
\texttt{CONCEPT}(\mathsf{cid}, \_), \texttt{CS}(\mathsf{cid, csid, cid, pid}), \texttt{PA}(\mathsf{pid}, \mathsf{csid}) & \hspace{-10pt} \to \hspace{-10pt} & \texttt{Cs}(\mathsf{cid, csid}) \\
\texttt{PA}(\mathsf{pid, csid}) & \hspace{-10pt} \to \hspace{-10pt} & \texttt{Pa}(\mathsf{pid, csid)} & 
\texttt{PA}(\mathsf{pid}, \_), \texttt{SP}(\mathsf{spid, sid, pid, pid, sid}), \texttt{SENTENCE}(\mathsf{sid}, \_) & \hspace{-10pt} \to \hspace{-10pt} & \texttt{Sp}(\mathsf{spid, sid, pid}) \\
\texttt{SENTENCE}(\mathsf{sid, pmid}) & \hspace{-10pt} \to \hspace{-10pt} & \texttt{Sentence}(\mathsf{sid, pmid}) & & \hspace{-10pt} \hspace{-10pt} & \\
\end{array}
\]

\end{minipage}
}
\vspace{-5pt}
\caption{Database transformer for our example. All variables are implicitly universally quantified.}
\label{fig:example-transformer}
\vspace{-15pt}
\end{figure}

\bfpara{Need for database transformers.}
In order to conclude that the SQL and Cypher queries are semantically different, we need to reason about how they behave when executed on the \emph{same data}. However, because graph and relational databases have such different data models, the input database instances are never identical. Thus, in order to reason about query equivalence, we first need a mechanism for defining \emph{data equivalence}. In our framework, this is done through the concept of \emph{\trans}, which takes as input a graph database instance $D$ over a certain schema $\schema$ and produces a relational database instance $D'$  over a relational schema $\schema'$. {Our concept of \trans is adapted from prior work on \emph{schema mappings}~\cite{clio-vldb00,clio-cmfa09}, generalized to model the correspondence between graph and relational databases. For our running example, the correspondence between the two database instances is given by the transformer shown  in Figure~\ref{fig:example-transformer}. 
Intuitively, each rule describes how each table in the relational database can be generated based on nodes and edges in the graph.
For example, the rule $\texttt{CONCEPT}(\mathsf{cid}, \_), \texttt{CS}(\mathsf{cid, csid, cid, pid}), \texttt{PA}(\mathsf{pid}, \mathsf{csid}) \to \texttt{Cs}(\mathsf{cid, csid})$ specifies if there is an edge \texttt{CS} connecting two nodes \texttt{CONCEPT} and \texttt{PA} in the graph, and the first two properties of \texttt{CS} are $\mathsf{cid}$ and $\mathsf{csid}$, then there is a row $(\mathsf{cid}, \mathsf{csid})$ in the \texttt{Cs} table.
Here, the last two attributes of \texttt{CS} serve as foreign keys referencing to the source and target nodes of the \texttt{CS} edge.

\begin{figure}[!t]
\centering
\footnotesize

\begin{tikzpicture}[relation/.style={rectangle split, rectangle split parts=#1, rectangle split part align=base, draw, anchor=center, align=center, text height=3mm, text centered}]\hspace*{-0.3cm}

\node (Concept) {\textbf{Concept'}};
\node [left=-2.5cm of Concept] (CS) {\textbf{Cs'}};
\node [left=-4.0cm of CS] (Pa) {\textbf{Pa'}};
\node [left=-2.5cm of Pa] (Sp) {\textbf{Sp'}};
\node [left=-5.5cm of Sp] (Sentence) {\textbf{Sentence'}};

\node [relation=2, rectangle split horizontal, rectangle split part fill={lightgray!50}, anchor=north west, below=0.6cm of Concept.west, anchor=west] (ConceptRow)
{\nodepart{one} \textbf{CID}
\nodepart{two} NAME};

\node [relation=4, rectangle split horizontal, rectangle split part fill={lightgray!50}, anchor=north west, below=0.6cm of CS.west, anchor=west] (CSRow)
{\nodepart{one} \textbf{CSID}
\nodepart{two} CID
\nodepart{three} \underline{SRC}
\nodepart{four} \underline{TGT}};

\node [relation=2, rectangle split horizontal, rectangle split part fill={lightgray!50}, anchor=north west, below=0.6cm of Pa.west, anchor=west] (PaRow)
{\nodepart{one} \textbf{PID}
\nodepart{two} CSID};

\node [relation=5, rectangle split horizontal, rectangle split part fill={lightgray!50}, anchor=north west, below=0.6cm of Sp.west, anchor=west] (SpRow)
{\nodepart{one} \textbf{SPID}
\nodepart{two} SID
\nodepart{three} PID
\nodepart{four} \underline{SRC}
\nodepart{five} \underline{TGT}};

\node [relation=2, rectangle split horizontal, rectangle split part fill={lightgray!50}, anchor=north west, below=0.6cm of Sentence.west, anchor=west] (SentenceRow)
{\nodepart{one} \textbf{SID}
\nodepart{two} PMID};

\draw[-latex] (CSRow.three south) -- ++(0,0) -| ($(CSRow.three south) + (0,0)$) |- ($(ConceptRow.one south) + (0,-0.2)$) -| ($(ConceptRow.one south) + (0,0)$);

\draw[-latex] (CSRow.four south) -- ++(0,0) -| ($(CSRow.four south) + (0,0)$) |- ($(PaRow.one south) + (-0.1,-0.2)$) -| ($(PaRow.one south) + (-0.1,0)$);

\draw[-latex] (SpRow.four south) -- ++(0,0) -| ($(SpRow.four south) + (0,0)$) |- ($(PaRow.one south) + (0.1,-0.2)$) -| ($(PaRow.one south) + (0.1,0)$);

\draw[-latex] (SpRow.five south) -- ++(0,0) -| ($(SpRow.five south) + (0,0)$) |- ($(SentenceRow.one south) + (0,-0.2)$) -| ($(SentenceRow.one south) + (0,0)$);

\end{tikzpicture}

\vspace{-5pt}
\caption{Induced relational schema. Primary keys are in bold. Foreign keys are underlined and connected to references with arrows.}
\label{fig:example-induced-schema}
\end{figure}
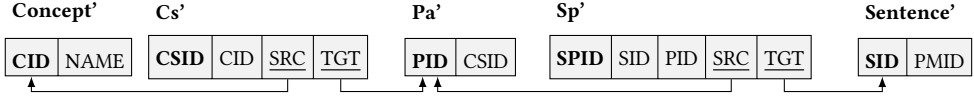

\bfpara{Induced relational schema and default transformer.} As mentioned in Section~\ref{sec:intro}, our approach to reasoning about equivalence between SQL and Cypher queries relies on first transpiling the given Cypher query to a SQL query over the \emph{induced relational schema},  which corresponds to a natural relational representation of the graph database.
In particular, Figure~\ref{fig:example-induced-schema} shows the induced relational schema for  Figure~\ref{fig:example-graph-schema} and is obtained by (a) translating both node and edge types in the graph schema into tables, and (b) translating incidence and adjacency information between node and edge types into \emph{functional dependencies}. For instance, nodes of type {\tt CONCEPT} are mapped to the {\tt Concept} table, edges of type {\tt CS} are mapped to the {\tt Cs} table, and so on. To handle functional dependencies, we need to introduce foreign  keys in the corresponding tables. 
Compared to the relational schema in Figure~\ref{fig:example-relational-schema}, the induced relational schema preserves the attributes while using additional attributes (i.e. \texttt{SRC} and \texttt{TGT}) as foreign keys to represent functional dependencies.
For example, for edges of type {\tt CS}, the source node is of type {\tt CONCEPT}; thus, the induced relational schema has the \texttt{SRC} attribute as a foreign key to \texttt{CID} of the {\tt Concept} table. Similarly, the \texttt{TGT} attribute is considered a foreign key to the {\tt PA} table.
Given the original graph database schema, our method  constructs a so-called \emph{standard database transformer} $\sdt$ that can be used to convert any instance of the given graph schema to a relational database over its induced schema.


\begin{wrapfigure}{R}{0.5\textwidth}
\centering
\scriptsize 
\resizebox{1\linewidth}{!}{
\begin{tabular}{l}
\boldWith T1 \boldAs ( \boldSelect c1.CID \boldAs c1_CID, \ldots, s.SID \boldAs s_SID \\
\hspace{2em} \boldFrom \underline{Concept \boldAs c1 \boldJoin CS \boldAs r1} \\
\hspace{5em} \underline{\boldJoin PA \boldAs p1 \boldJoin SP \boldAs r2 \boldJoin Sentence \boldAs s } \\ 
\hspace{2em} \boldOn c1.CID = 1 \boldAnd c1.CID = r1.SRC \boldAnd \ldots~ \boldAnd r2.TGT = s.SID ), \\
\hspace{2em} T2 \boldAs ( \boldSelect s_SID \boldFrom T1 ), \\
\hspace{2em} T3 \boldAs ( \boldSelect s.SID \boldAs s_SID, \ldots, c2.CID \boldAs c2_CID \\
\hspace{2em} \boldFrom Sentence \boldAs s \boldJoin SP \boldAs r3 \\
\hspace{5em} \boldJoin PA \boldAs p2 \boldJoin CS \boldAs r4 \boldJoin Concept \boldAs c2 \\ 
\hspace{2em} \boldOn s.SID = r3.TGT \boldAnd \ldots~ \boldAnd r4.SRC = c2.CID ), \\
\hspace{2em} T4 \boldAs ( \boldSelect * \boldFrom T2 \boldJoin T3 \boldOn T2.s_SID = T3.s_SID ) \\
\boldSelect T4.c2_CID, \boldCount(*) \boldFrom T4 \boldGroupBy T4.c2_CID \\
\end{tabular}
}
\vspace{-10pt}
\caption{Transpilation result for the Cypher query.}
\label{fig:example-transpiled}
\vspace{-10pt}
\end{wrapfigure}

\bfpara{Syntax-directed transpilation.} 
Intuitively, the standard \trans establishes a one-to-one correspondence between  elements of the graph database and the corresponding entries in the induced relational database. Hence, we can use syntax-directed translation to directly transpile the Cypher query to a SQL query over the induced schema.  Here,  the transformer  $\sdt$ fixes the mapping between ``atomic elements'' of both databases; thus, atomic queries over nodes and edges in Cypher can be translated to atomic queries over SQL tables. This forms the base case for an inductive transpilation scheme, leveraging the key insight that Cypher path queries correspond to relational joins. For example, a Cypher path query like 
$\boldMatch \ (u) \texttt{-[:CS]->} (v)$ translates to a SQL join query between the \texttt{Concept}, \texttt{Cs}, and \texttt{Pa} tables, since the \trans specifies that the \texttt{CS} edge type maps to the \texttt{Cs} table in the relational schema.
Such a transpilation scheme is conceptually simple, but as our final transpilation result in Figure~\ref{fig:example-transpiled} shows, one must take significant care to address the different types of path patterns in Cypher syntax, in addition to translating compositions of path queries with other Cypher operators such as aggregation.
Specifically, the first pattern matching in Figure~\ref{fig:example-cypher-query} is translated to \texttt{T1} while the \texttt{WITH} clause propagates the intermediate results to \texttt{T2}. Similarly, the second pattern matching is translated to \texttt{T3}. Due to the shared node \texttt{s:SENTENCE} in these two path patterns, we join \texttt{T2} and \texttt{T3} as \texttt{T4}. The final \texttt{RETURN} clause is translated to \texttt{GroupBy} because of the aggregation expression.

\begin{figure}[!t]
\footnotesize
\vspace{-10pt}
\resizebox{.9\linewidth}{!}{
\hspace{-25pt}
\begin{minipage}{\linewidth}

\[
\begin{array}{rclrcl}
\texttt{Concept'}(\mathsf{cid}, \mathsf{name}) & \hspace{-10pt} \to \hspace{-10pt} & \texttt{Concept}(\mathsf{cid}, \mathsf{name}) & 
\texttt{Concept'}(\mathsf{cid}, \_), \texttt{Cs'}(\mathsf{cid, csid, cid, pid}), \texttt{Pa'}(\mathsf{pid}, \mathsf{csid}) & \hspace{-10pt} \to \hspace{-10pt} & \texttt{Cs}(\mathsf{cid, csid}) \\
\texttt{Pa'}(\mathsf{pid, csid}) & \hspace{-10pt} \to \hspace{-10pt} & \texttt{Pa}(\mathsf{pid, csid)} & 
\texttt{Pa'}(\mathsf{pid}, \_), \texttt{Sp'}(\mathsf{spid, sid, pid, pid, sid}), \texttt{Sentence'}(\mathsf{sid}, \_) & \hspace{-10pt} \to \hspace{-10pt} & \texttt{Sp}(\mathsf{spid, sid, pid}) \\
\texttt{Sentence'}(\mathsf{sid, pmid}) & \hspace{-10pt} \to \hspace{-10pt} & \texttt{Sentence}(\mathsf{sid, pmid}) & & \hspace{-10pt} \hspace{-10pt} & \\
\end{array}
\]

\end{minipage}
}

\vspace{-5pt}
\caption{Residual database transformer. All variables are universally quantified.}
\label{fig:example-rdt}
\vspace{-15pt}
\end{figure}

\bfpara{Checking equivalence.}
While our approach allows correct-by-construction transpilation from Cypher to SQL, there are several reasons why this is not sufficient. First, our approach does not guarantee that the transpilation result is the most efficient, so even though one could use existing SQL query optimizers to further optimize the query, the user may want to write their hand-optimized SQL query. Second, the user may want to use a different relational schema rather than our default version. Third, while our transpilation rules allow translating Cypher to SQL, they do not address the reverse direction. Motivated by these shortcomings, our method leverages the proposed transpilation algorithm to perform verification between any given pair of Cypher and SQL queries by utilizing existing tools for SQL. In particular, the key idea  is to infer a \emph{residual database transformer} that specifies the relationship between the relational database over the induced schema and the target relational database (as specified by the user-provided database transformer). For our running example,   Figure~\ref{fig:example-rdt} shows the residual transformer that can be used to convert instances of the induced relational schema from Figure~\ref{fig:example-induced-schema} to instances of the desired schema. Given such a residual schema, we  can use an existing SQL equivalence checker, such as \verieql~\cite{verieql-oopsla24}, to refute equivalence between these queries and obtain the counterexample shown in Figure~\ref{fig:example-instances}.


\section{Preliminaries} \label{sec:prelim}


\subsection{Background on Graph Databases}

A graph database instance is a \emph{property graph}, which contains nodes and edges carrying data. Typically, nodes model entities, and edges model relationships. Each node or edge in the graph stores data represented as pairs of property keys and values. Additionally, each node or edge is assigned a \emph{label}, which describes the kind of entity or relationship it models. A well-formed property graph should conform to a graph schema, which is formalized below.

\begin{definition}[{\bf Node/edge type}]
A node type $\nodeType$ is a tuple $(l, K_1, \ldots, K_n)$ where $l$ is the label of the node (e.g., $\mathsf{Actor}$) and $K_1,\ldots, K_n$ are the property keys for that node type (e.g., $\mathtt{name}$, $\mathtt{dob}$, etc.). An edge type $\edgeType$ is also a tuple $(l, \stype, \ttype, K_1,\ldots , K_m)$ where $l$ is a label (e.g., $\mathsf{ACTS\_IN}$), $\stype$ and $\ttype$ are the types of the source and target nodes respectively, and $K_1,\ldots, K_m$ are the property keys (e.g., $\mathtt{role}$) for that edge type.

\end{definition}

For each node type $\nodeType = (l, K_1, \ldots, K_n)$, we define $\lbl(\nodeType)$ to give the label $l$ of $\nodeType$, and we assume that $K_1$ is the \emph{default property key} for $\nodeType$, which is a key with a globally unique value, similar to a primary key in a relational database.  We define $\keys(\nodeType) = \{K_1,\ldots,K_n\}$ to yield the set of property keys associated with $\nodeType$. For an edge type $\edgeType = (l', \stype,\ttype, K'_1, \ldots, K'_m)$ we similarly define $\lbl(\edgeType) = l'$ and $\keys(\edgeType) = \{K'_1, \ldots, K'_m\}$, with key $K'_1$ being the default property key. Additionally, we define $\dstType(\edgeType) = \stype$ and $\srcType(\edgeType) = \ttype$. 



\begin{definition}[{\bf Graph database schema}]
A graph database schema $\gschema$ is a pair $(T_N, T_E)$ where $T_N$ is a set of node types and $T_E$ is a set of edge types. 
\end{definition}

For each graph database schema $\gschema = (T_N,T_E)$, we assume that the label of each node or edge type uniquely defines it: $\forall t_1, t_2 \in T_N \cup T_E. \lbl(t_1) \neq \lbl(t_2)$. Thus, we can use types and labels interchangeably. {Additionally, we assume that all property keys are unique inside a given schema $\Psi_G$, i.e., there are no name clashes between arbitrary pairs of property keys between different types.}

\begin{definition}[{\bf Graph database}]
An instance of a graph database schema $\Psi_G = (T_N, T_E)$ is a tuple $G = (N, E, P, T)$ where $N$ is a set of nodes,  $E \subseteq N \times N$ is a set of edges, $P: (N \cup E) \times \textbf{Keys} \to \textbf{Values}$, and $T : N\cup E \to T_N \cup T_E$ gives the type of a node $n \in N$ or an edge $e \in E$. 



\end{definition}

We use the notation $P(n, k)$ to give the value of a property key $k$ in node $n$  and analogously define $P(e, k)$  for an edge $e$. 
We also use the notation $G \conform \Psi_G$ to denote that G is an instance of schema $\Psi_G$, and we refer to any subgraph of $G$ as a \emph{property graph}.

\begin{figure}
\centering
\small 
\[
\begin{array}{l l c l}
\text{Query} & Q & ::= &  R  ~|~ \sfOrderBy(R, k, b)  ~|~ \sfUnion(Q, Q) ~|~ \sfUnionAll(Q, Q) \\ 
\text{Return Query} & R & ::= & \sfReturn(C, \overline{E}, \overline{k}) \\
\text{Clause} & C & ::= & \sfMatch(\pattern, \pred) ~|~ \sfMatch(C, \pattern, \pred) ~|~ \sfOptMatch(C, \pattern, \pred)  ~|~ \sfWith(C, \overline{X}, \overline{X}) \\ 
\text{Path Patt.} & \pattern & ::= & \nodepattern ~|~ \nodepattern, \edgepattern, \pattern \\ 
\text{Node Patt.} & \nodepattern & ::= & (X, l) \quad\quad \text{Edge Patt.} \quad \edgepattern ::= (X, l, d) \\ 
\text{Expression} & E & ::= & k ~|~ v ~|~ \sfCast(\phi) ~|~ {\sfAgg(E)} ~|~ E \arithop E \\ 
\text{Predicate} & \pred & ::= & \top ~|~ \bot ~|~ E \odot E ~|~ \sfIsNull(E) ~|~ E \in \overline{v} 
                 ~|~ \sfExists(\pattern) ~|~ \pred \land \pred ~|~ \pred \lor \pred ~|~ \neg \pred \\
\end{array}
\]
\[
\begin{array}{c}
X \in \textbf{Node/Edge Names} \quad
l \in \textbf{Labels} \quad
k \in \textbf{Property Keys} \quad
v \in \textbf{Values} \quad
b \in \textbf{Bools} \\
\sfAgg \in \set{\sfCount, \sfAvg, \sfSum, \sfMin, \sfMax} \quad 
d \in \{\rightarrow, \leftarrow, \leftrightarrow \} 
\end{array}
\]
\vspace{-10pt} 
\caption{Featherweight Cypher syntax where $\arithop, \logicop$ correspond to arithmetic and logical operators respectively.}
\label{fig:cypher-syntax}
\vspace{-10pt} 
\end{figure}

\subsection{Query Language for Graph Databases}
To formalize our method, we  focus on a subset of the popular Cypher database query language~\cite{cypher-sigmod18}. This subset,  which we refer to as ``Featherweight Cypher'', is presented in Figure~\ref{fig:cypher-syntax}. 
\footnote{
The denotational semantics is formally described in Appendix~\ref{sec:cypher-semantics}.
}
A  query $Q$ is either a union of return queries  $R$ or an order-by statement following one or more such return queries. 
Each return query $R$ takes as input a clause $C$, a list of expressions $\overline{E}$, and a list of property key names $\overline{k}$. Intuitively, the return query shapes a list of graphs into a table. Each clause in the return statement is a $\sfMatch$, representing a pattern match over an input property graph, and the patterns are specified using the path pattern $\pattern$.
We do not include Cypher features such as unbounded-length path queries and graph reachability primitives (e.g.,  \texttt{shortestPath}), which are not expressible in the core SQL fragment considered in this paper.

\begin{example}\label{example:cypher1}
Consider the following Cypher query
\[
\texttt{\boldMatch (n:EMP)-[:WORK_AT]->(m:DEPT) \boldReturn m.dname \boldAs name, \textbf{Count}(n) \boldAs num}
\]
that returns a table containing department names and the number of employees. We can represent it using our featherweight Cypher syntax as follows
\[
\sfReturn(\sfMatch([(\texttt{n}, \texttt{EMP}), (\texttt{e}, \texttt{WORK\_AT}, \rightarrow), (\texttt{m}, \texttt{DEPT})], \top), [\texttt{m.dname}, \sfCount(\texttt{n.id})], [\texttt{name}, \texttt{num}])
\]
Here, the match clause retrieves all paths of length one from \texttt{EMP} nodes to \texttt{DEPT} nodes connected by an edge of type \texttt{WORK_AT}. Then, the return clause reshapes the set of matched paths into a table with two columns: \texttt{name} and \texttt{num}. The \texttt{name} column is populated with values corresponding to the property key \texttt{dname} of the \texttt{DEPT} node, and the \texttt{num} column is populated by the count of \texttt{n.id}, {where \texttt{m} and \texttt{n} refer to the source and target nodes of the matched edge, respectively.} 

\end{example}

\subsection{Relational Databases}


\begin{definition}[{\bf Relational  schema}]
A \emph{relational database schema} is a pair $\rschema := (\tables, \integrity)$ where $\tables : \mathcal{R} \to [\mathcal{A}]$ is a mapping from a set of relation names $\mathcal{R}$ to a list of attributes, and $\integrity$ is an \emph{integrity constraint}. {We assume that all attribute names in a schema are unique.}
\end{definition}




We represent an integrity constraint as a conjunction of  three types of \emph{atomic} constraints:

\begin{enumerate}[leftmargin=*]
    \item {\bf Primary key constraints:} A primary key constraint $\PK(R) = a$ specifies that attribute $a$ is the \emph{primary key} for relation $R$ --- i.e.,  there  cannot be multiples tuples of $R$ that agree on  $a$.
    \item {\bf Foreign key constraints:} A foreign key constraint $\FK(R.a) = R'.a'$ specifies that the attribute $a$ in relation $R$ is a \emph{foreign key} corresponding to attribute $a'$ in relation $R'$. That is, the values stored in attribute $a$ of relation $R$ must be a subset of the values stored in attribute $a'$ of $R'$. 
    \item {\bf Not-null constraints:} A not-null constraint $\mathsf{NotNull}(R, a)$ specifies that the value stored at attribute $a$ of relation $R$ must not be $\mathsf{Null}$.
\end{enumerate}

\begin{definition}[{\bf Relational database instance}]
    A relational database instance $\reldb$ is a collection of  tuples $\{r_1,\ldots,r_m\}$, where each $r_i \in \reldb$ is of the form $(a_1 : v_1, \ldots, a_n : v_n)$. Here, $a_1,\ldots,a_n$ are attributes and $v_1,\ldots,v_n$ are values. We let $\attributes(r_i)$ return the list of attributes $a_1,\ldots,a_n$ in sequence. We use the notation $r_i.a$ to denote the value stored in attribute $a$ of tuple $r_i$.
\end{definition}


As with graph databases, we use the notation  $\reldb \conform \rschema $, to denote that $\reldb$ is instance of $\rschema$.

\bfpara{SQL query language}. In this paper, we consider relational database queries written in SQL. Figure~\ref{fig:sql-syntax} shows the subset of SQL that we use in our formalization.{
At a high level, this language extends  core relational algebra (e.g., projection $\proj$, selection $\filter$, renaming $\rename$, joins $\joinop$, set and bag unions $\cap, \uplus$) to incorporate \ttGroupBy, \ttOrderBy, and \ttWith clauses. It can express a representative fragment of SQL queries that are commonly used in practice.
The semantics of these SQL operators are standard and formally defined by prior work such as \citet{verieql-oopsla24}.
}

\begin{figure}
\vspace{-0.2in}
\centering
\small
\[
\begin{array}{l l c l}
\text{Query} & Q & ::= & R ~|~ \proj_L(Q) ~|~ \filter_{\phi}(Q) ~|~ \rho_{R}(Q) ~|~ Q \cup Q ~|~ Q \uplus Q ~|~ Q \joinop Q  \\
& & ~|~ & \ttGroupBy(Q, \overline{E}, L, \phi) ~|~ \ttWith(\overline{Q}, \overline{R}, Q) ~|~ \ttOrderBy(Q, a, b) \\ 
\text{Attribute List} & L & ::= & E ~|~ \rho_a(E) ~|~ L, L \\ 
\text{Attribute Expr} & E & ::= & a ~|~ v ~|~ \ttCast(\phi) ~|~ \ttAgg(E) ~|~ E \arithop E \\ 
\text{Predicate} & \phi & ::= & b ~|~ E \logicop E ~|~ \ttIsNull(E) ~|~ E \in \overline{v} ~|~ \overline{E} \in Q ~|~ \phi \land \phi ~|~ \phi \lor \phi ~|~ \neg \phi \\
\text{Join Op} & \joinop & ::= & \times ~|~ \bowtie_\phi ~|~ \fjoin_\phi ~|~ \leftouterjoin_\phi ~|~ \rightouterjoin_\phi \\ 
\end{array}
\] 
\[
\begin{array}{c}
R \in \textbf{Relation Names} \quad a \in \textbf{Attr Names} \quad  v \in \textbf{Values} \quad b \in \textbf{Bools}  \quad
\ttAgg \in \set{\ttCount, \ttAvg, \ttSum, \ttMin, \ttMax}
\end{array}
\]
\vspace{-20pt}
\caption{Featherweight SQL syntax; $\arithop$ and $\logicop$ represent arithmetic and logical operators respectively}
\label{fig:sql-syntax}
\vspace{-10pt}
\end{figure}



\section{Problem Statement} \label{sec:problem}

In this section, we first describe the language for database transformers and then formally define the equivalence checking  problem between graph and relational databases.

\subsection{Language for Database Transformers}\label{sec:transformer}

\begin{figure}
\centering
\small
\[
\begin{array}{l l c l}
\text{Transformer} & \transformer & ::= & P, \ldots, P \to P ~|~ \transformer ~\transformer \\
\text{Predicate}   & P & ::= & E(t, \ldots, t) \\
\text{Term}        & t & ::= & c ~|~ v ~|~ \_ \\
\end{array}
\]
\[
\begin{array}{c}
E \in \textbf{Table Names} \cup \textbf{Node Labels} \cup \textbf{Edge Labels} \quad
c \in \textbf{Constants} \quad
v \in \textbf{Variables} \\
\end{array}
\]
\vspace{-20pt}
\caption{Syntax of the database transformer.}
\label{fig:transformer-syntax}
\vspace{-10pt}
\end{figure}

In this section, we present a small domain-specific language (DSL), shown in Figure~\ref{fig:transformer-syntax},  for expressing database transformers. Following prior work~\cite{dynamite-vldb20}, our DSL generalizes the standard concept of \emph{schema mapping} ~\cite{clio-vldb00,clio-cmfa09} for relational databases to a more flexible form. In particular, a database transformer in this DSL is expressed as a set of first-order formulas of the form $P_1, \ldots, P_n \to P_0$, where each $P_i$ is a predicate that represents a  database element, such as a table in a relational database or a node or edge in a graph database. Each predicate is of the form $E(t_1, \ldots, t_n)$ where $E$ corresponds a table name, node labels, or edge labels, and each $t_j$ is a term  (variable or a constant), with $\_$ denoting a fresh variable that is not used. All  variables are implicitly universally quantified but the quantifiers are omitted in the syntax for brevity.
Intuitively, the formula $P_1, \ldots, P_n \to P_0$ expresses that, if predicates $P_1, \ldots, P_n$ hold over a database instance $D$, then predicate $P_0$ holds over another database instance~$D'$.

\bfpara{Semantics.}
To define the semantics of our transformer DSL, we first represent a database transformer $\transformer$ as a set of universally quantified first-order logic formulas, denoted as $\denot{\transformer}$.  The idea behind the semantics of our DSL is to represent each database instance as a set of ground predicates and then check whether these predicates entail the first-order logic formula $\denot{\transformer}$ under the Herbrand semantics.

To make this discussion more precise, we introduce a function $\convert$ that maps a database instance to a set of ground predicates representing its structure and contents. Formally, $\convert$ is defined as follows:
\[
\convert(D) = \{ E(t_1, \ldots, t_n) \mid E \in D \}
\]
The mapping of elements in $D$ to predicates depends on whether $D$ is a relational or graph database:

\begin{itemize}[leftmargin=*]
    \item \textbf{For relational database instance $D$:}  If $R$ is a table in $D$ with a set of records $\{ (a_1, \ldots, a_n) \}$, then $R$ is converted to the following set of ground predicates:
        \[
        \set{ R(a_1, \ldots, a_n) \mid (a_1, \ldots, a_n) \in R }
        \]
        Here, $R$ represents the table name, and each $a_i$ is a constant representing a value in the record.
    
    \item \textbf{For graph database instance $D$:} If $N(l, a_1, \ldots, a_n)$ denotes a node with label $l$ and values $a_1, \ldots, a_n$ of property keys $K_1, \ldots, K_n$, nodes are converted to the following set of ground facts:
        \[
        \set{ l(a_1, \ldots, a_n) \mid \text{node } N(l, a_1, \ldots, a_n) \in D }
        \]
   Similarly, if $E(l, s, t, a_1, \ldots, a_n)$ denotes an edge with label $l$ that connects nodes $s$ and $t$  and has property  values $a_1, \ldots a_n$, then edges are converted to predicates as follows:
        \[
        \set{ l(a_1, \ldots, a_n, s, t) \mid \text{edge } E(l, s, t, a_1, \ldots, a_n) \in D }
        \]

\end{itemize}

Given a database instance $D$,  $\convert(D)$ yields a set of ground predicates representing the structure and contents of $D$. We can now define the semantics of the transformer $\transformer$ as follows:
\[
\transformer(D) = D' \quad \Leftrightarrow \quad  \convert(D)\cup \convert(D') \models \denot{\transformer}
\]
where the notation $S \models \varphi$ indicates that the set $S$ of ground predicates is a Herbrand model of $\varphi$.

\begin{example}
Consider the graph and relational database instances $\graph, \reldb$ from Figures~\ref{fig:example-graph-instance} and ~\ref{fig:example-relational-instance} respectively. For the transformer $\transformer$ shown in Figure~\ref{fig:example-transformer}, we have $\transformer(\graph) = {\reldb}$.
\end{example}

\subsection{Equivalence Checking Problem}

In this section, we formally define what it means for a pair of graph and relational queries to be equivalent modulo a \trans $\transformer$, expressed in the DSL of Section~\ref{sec:transformer}.  To this end, we first introduce some necessary definitions and notations.

\begin{SCfigure}
\centering
\includegraphics[scale=0.2]{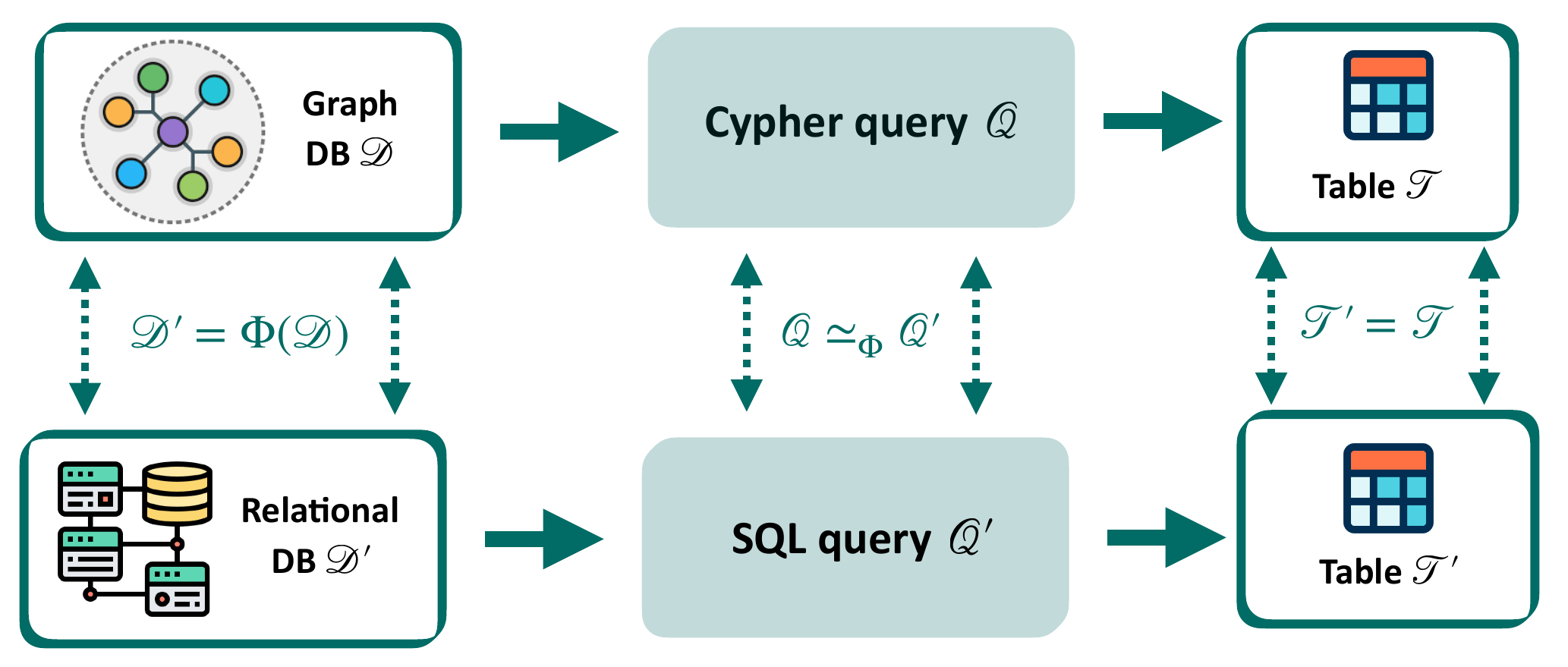}
\caption{Definition of equivalence between Cypher and SQL queries where $\Phi$ is the user-provided schema transformer. A pair of Cypher and SQL queries are equivalent if they produce the same table (modulo renaming) whenever they are executed on a pair of database instances satisfying $\Phi$. \\}
\label{fig:db-equivalence}
\end{SCfigure}

\begin{definition}[{\bf Database query}]
A database query $Q$ over schema $\schema$ takes as input a database instance $D$ such that   $D \conform \schema$ and yields a table. We denote the semantics of $Q$ as  $\denot{Q}_D$.
\end{definition}

The definition above is sufficiently general to describe both SQL and Cypher queries, as both query languages return a table.


\begin{definition}[{\bf Database equivalence}]
Let $D, D'$ be  database instances over schemas $\schema, \schema'$ respectively and let $\transformer$ be a \trans that can be used to convert instances of $\schema$ to instances of $\schema'$. Then,  $D$ is said to be equivalent to $D'$ modulo $\transformer$, denoted $D \sim_\transformer D'$, if $\transformer(D) = D'$.
\end{definition}

According to the above definition, a graph database instance $G$ is equivalent to a relational database instance $R$ if the contents of $R$ can be obtained from $G$ by applying the transformer $\transformer$ to $G$.  Next, to define equivalence between graph and relational queries, we need to define what it means for the query outputs to be the same. Since queries in both Cypher and SQL return tables, we need a notion of equivalence between tables. 

\begin{definition}[{\bf Table equivalence}]
\label{def:tableEquivalence}
Two tables $T$ and $T'$ are said to be equivalent, denoted $T \equiv T'$, if there exists a bijective mapping $\pi$ from columns of $T$ to those of $T'$ such that, for each  tuple $r \in T$ with multiplicity $n$, there exists a unique tuple $r' \in T'$ with multiplicity $n$ where $\forall a \in \attributes(r).~ r.a = r'.\pi(a)$. 
\end{definition}

In other words, our notion of table equivalence disregards the order of attributes as well as their names, which allows for a more robust notion of query equivalence.\footnote{If a query includes an $\ttOrderBy$ clause, list semantics will be applied, where the result is an ordered list of tuples. In this case, the equivalence of two tables $T$ and $T'$ is defined such that there exists a bijective mapping $\pi$ between their attributes, and for each pair of tuples $r \in T$ and $r' \in T'$ at the same index, it holds that $\forall a \in \attributes(r).~ r.a = r'.\pi(a)$.}


\begin{definition}[{\bf Graph-relational query equivalence}]
\label{def:queryEquivalence}
Let $\gschema$ and $\rschema$ be graph and relational  schemas respectively, and $\transformer$ be a transformer from $\gschema$ to $\rschema$. A query $Q$ over $\gschema$ is \emph{equivalent} to $Q'$ over $\rschema$ modulo $\transformer$, denoted $Q \simeq_\transformer Q'$, iff:
\[
\forall G, R. \ \ (G \conform \gschema \land R \conform \rschema \land G \sim_{\transformer} R) \ \Rightarrow  \ \denot{Q}_G \tbleq \denot{Q'}_R
\]
\end{definition}

In other words, $Q, Q'$ are considered equivalent if they produce the same tables (modulo renaming/re-ordering of columns) when executed on a pair of database instances $G, R$ satisfying
$G \sim_\transformer R $. Figure~\ref{fig:db-equivalence} visualizes this definition of equivalence between graph databases and relational databases. 




\begin{definition}[{\bf Equivalence checking problem}]
\label{def:verificationProblem}
Given  graph and relational queries $Q, Q'$ and  a  transformer $\transformer$ from graph schema $\gschema$ to relational schema $\rschema$,  the  equivalence checking problem is to decide whether  $Q \simeq_\transformer Q'$.
\end{definition}

\section{Equivalence Checking Algorithm}\label{sec:algorithm}

\begin{figure}[!t]
\begin{algorithm}[H]
\caption{Methodology for checking equivalence between Cypher and SQL queries}
\label{algo:verify}
\begin{algorithmic}[1]
\Procedure{\textsc{CheckEquivalence}}{$\gschema, Q_G, \rschema, Q_R, \transformer$}
\vspace{2pt}
\Statex \textbf{Input:} Graph and relational schemas $\gschema, \rschema$, Cypher query $Q_G$, SQL query $Q_R$, transformer  $\transformer$
\Statex \textbf{Output:} $\top$ for equivalence or $\bot$ indicating failure
\vspace{2pt}


\State $( \sdt, \rschema') \gets \textsc{InferSDT}(\gschema)$
\State $Q_{R'} \gets \textsc{Transpile}(Q_G, \sdt, \rschema')$
\State \Return $\textsc{ReduceToSQL}(\rschema, Q_R, \schema_{R}', Q_{R}', \transformer, \sdt)$

\EndProcedure
\end{algorithmic}
\end{algorithm}
\vspace{-20pt}
\end{figure}

In this section, we present our algorithm, summarized in Algorithm~\ref{algo:verify}, for checking equivalence between Cypher and SQL queries.  As shown in Algorithm~\ref{algo:verify}, our algorithm consists of three steps: 

\begin{enumerate}[leftmargin=*]
\item {\bf Schema and transformer inference:}
Given the graph database schema $\gschema$, our algorithm first invokes the {\sc InferSDT} function to infer both the induced relational schema $\rschema' = (S, \integrity)$ as well as the standard \trans $\sdt$. 
\item {\bf Syntax-directed transpilation:}
Next, our algorithm uses the inferred \trans and integrity constraints to transpile the Cypher query into an SQL query $Q_{R}'$ that is guaranteed to be equivalent to $Q_G$ modulo the standard $\sdt$. 
\item {\bf Checking SQL equivalence:}
{
Finally, the algorithm computes a residual \trans $\rdt$ 
that can be used to convert instances of $\rschema'$ into instances of $\rschema$ and checks equivalence between SQL queries $Q_R$ and $Q_{R}'$ modulo the residual \trans $\rdt$ relating a pair of database schemas.
}
\end{enumerate} 



\bfpara{Discussion.} An alternative approach to solving the equivalence checking problem would be to directly reason about equivalence between the graph query $Q_G$ and the relational query $Q_R$. While such an approach would be more direct, we adopt the methodology illustrated in Figure~\ref{fig:eq-approach} for several reasons. First, in order to directly verify equivalence between graph and relational database queries, we need suitable SMT encodings of \emph{both} graphs and relations, which makes the resulting constraint solving problem harder compared to the alternative. Second, the reduction to relational equivalence checking allows us to leverage a variety of existing tools that have been developed for SQL, including testing tools~\cite{xdata-vldbj15}, bounded model checkers~\cite{verieql-oopsla24}, and deductive verifiers~\cite{mediator-popl18}.

\subsection{Induced Relational Schema and Standard Transformer Inference}

\begin{figure}[!t]
\begin{mdframed}
\small
\[
\hspace{-10pt}
\begin{array}{c}

{
\irulelabel
{\begin{array}{c}
 \nodeType = (l, K_1, \ldots, K_n) \quad 
 \integrity = (\mathsf{PK}(R_l) = K_1) \quad
 \transformer = \set{ l(K_1, \ldots, K_n) \to R_l(K_1, \ldots, K_n) }
\end{array}}
{ \nodeType \hookrightarrow (\set{R_l \mapsto  (K_1, \ldots, K_n)}, \integrity, \transformer)  }
{\textrm{(Node)}}
}

\\ \ \\ 

{
\irulelabel
{\begin{array}{c}
 \edgeType = (l,  \stype, \ttype, K_1, \ldots, K_m) \quad
 \lbl(\stype) = s \quad
 \lbl(\ttype) = t \\
 \integrity = \PK(R_l) = R_l.K_1 \land \FK(R_l.\fk_s) = \PK(R_s) \land \FK(R_l.\fk_t) = \PK(R_t) \\
 \transformer = \set{ l(K_1, \ldots, K_m, \fk_s, \fk_t) \to R_l(K_1, \ldots, K_m, \fk_s, \fk_t) }
\end{array}}
{ \edgeType \hookrightarrow (\set{R_l \mapsto (K_1, \ldots, K_m, 
\fk_s, \fk_t)}, \integrity, \transformer)  }
{\textrm{(Edge)}}
}

\\ \ \\ 

\irulelabel
{\begin{array}{c}
T_1 \hookrightarrow (S_1, \integrity_1, \transformer_1) \quad  T_1 \hookrightarrow (S_2, \integrity_2, \transformer_2) \\
\end{array}}
{T_1 \uplus T_2 \hookrightarrow (S_1 \uplus S_2, \integrity_1 \land \integrity_2, \transformer_1 \cup \transformer_2)}
{\textrm{(Set)}}

\irulelabel
{\begin{array}{c}
(T_N \uplus T_E) \hookrightarrow (S, \integrity, \transformer)  \\
\end{array}}
{ (T_N, T_E) \hookrightarrow (S, \integrity, \transformer)}
{\textrm{(Schema)}}
\end{array}
\]
\end{mdframed}
\vspace{-10pt}
\caption{Rules for the {\sc InferSDT} procedure. $R'_l$ denotes the table name of $R_l$ in the induced relational schema.}
\vspace{0pt}
\label{fig:infer-sdt}

\end{figure}

We first discuss the {\sc InferSDT} procedure for inferring the induced relational schema as well as the standard \trans. This procedure is formalized in Figure~\ref{fig:infer-sdt} using inference rules of the form:
$
\gschema \hookrightarrow (S, \integrity, \sdt)
$
where $\gschema$ corresponds to elements of the graph schema (nodes, edges, and subgraphs),  $\rschema = (S, \integrity)$ is the induced relational schema for $\gschema$, and $\sdt$ is the standard \trans that can be used to convert instances of $\gschema$ into instances of $\rschema$.

\bfpara{Induced relational schema.}
Intuitively, the induced relational schema is the ``closest'' relational representation of the graph database schema $\gschema = (T_N, T_E)$ that represents each node and edge type as a relational table.
As shown in the \textrm{Node} rule, for each node type $\nodeType = (l, K_1, \ldots, K_n) \in T_N$ with label $l$ and default property key $K_1$, we introduce a table $R_l$ with attributes $K_1, \ldots, K_n$ in the induced relational schema. We also use the default property key $K_1$ as the primary key of the corresponding table and generate an integrity constraint $\PK(R_l) = K_1$. 
Similarly, as shown in the \textrm{Edge} rule, for each edge type $\edgeType = (l, \stype, \ttype, K_1, \ldots, K_m) \in T_E$ with default property key $K_1$, we introduce a table $R_l$ with attributes $K_1, ..., K_m, \fk_s, \fk_t$. Here, $K_1$ is also the primary key of table $R_l$, and $\fk_s, \fk_t$ are foreign keys which reference to primary keys of the tables corresponding to source and target nodes. Thus, the integrity constraint is $\PK(R_l) = R_l.K_1 \land \FK(R_l.\fk_s) = \PK(R_s) \land \FK(R_l.\fk_t) = \PK(R_t)$, where $R_s, R_t$ are the tables corresponding to the source and target nodes.

\begin{example}\label{ex:irs}

\begin{figure}[!t]
\centering
\small

\begin{subfigure}{0.56\linewidth}
\centering
\begin{tikzpicture}[main/.style = {draw=black, circle, ultra thick, minimum size=1.cm}] 

\node[main] (1) at (0, 0){\footnotesize EMP};
\node[main] (2) at (3., 0){\footnotesize DEPT};

\draw[->, line width=1.2pt] (1) -- (2) node [midway, above, sloped] (PURCHASED) {\footnotesize WORK\_AT};

\end{tikzpicture} 
\vspace{-5pt}
\caption{Graph schema. Node labels are inside the nodes. Edge labels are above the arrows. \textrm{EMP} has property keys: \textbf{id}, \textrm{name}. \textrm{DEPT} has property keys: \textbf{dnum}, \textrm{dname}. \textrm{WORK\_AT} has a property key: \textbf{wid}. The default property keys are in bold.}
\label{fig:example-irs-graph}
\end{subfigure}
~\hspace{10pt}
\begin{subfigure}{0.43\linewidth}
\centering
\footnotesize
\begin{tikzpicture}[relation/.style={rectangle split, rectangle split parts=#1, rectangle split part align=base, draw, anchor=center, align=center, text height=3mm, text centered}]

\node (emp) {\textbf{emp}};
\node [left=-3.cm of emp] (work) {\textbf{work\_at}};
\node [left=-3.cm of work] (dept) {\textbf{dept}};

\node [relation=2, rectangle split horizontal, rectangle split part fill={lightgray!50}, anchor=north west, below=0.6cm of emp.west, anchor=west] (empRow)
{\nodepart{one} \textbf{id}
\nodepart{two} name};

\node [relation=3, rectangle split horizontal, rectangle split part fill={lightgray!50}, anchor=north west, below=0.6cm of work.west, anchor=west] (workRow)
{\nodepart{one} \textbf{wid}
\nodepart{two} \underline{SRC}
\nodepart{three} \underline{TGT}};

\node [relation=2, rectangle split horizontal, rectangle split part fill={lightgray!50}, anchor=north west, below=0.6cm of dept.west, anchor=west] (deptRow)
{\nodepart{one} \textbf{dnum}
\nodepart{two} dname};

\draw[-latex] (workRow.two south) -- ++(0,0) -| ($(workRow.two south) + (0,0)$) |- ($(empRow.one south) + (0,-0.2)$) -| ($(empRow.one south) + (0,0)$);

\draw[-latex] (workRow.three south) -- ++(0,0) -| ($(workRow.three south) + (0,0)$) |- ($(deptRow.one south) + (0,-0.2)$) -| ($(deptRow.one south) + (0,0)$);

\end{tikzpicture}
\caption{Induced relational schema. Primary keys are in bold. Foreign keys are underlined and connected to references with arrows.}
\label{fig:example-irs-relational}
\end{subfigure}

\vspace{-5pt}
\caption{Example of a graph schema and its induced relational schema.}
\label{fig:example-irs}
\vspace{-10pt}
\end{figure}
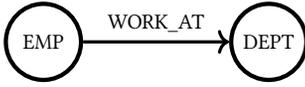
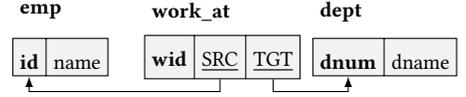

Consider the graph schema $\gschema$ shown in Figure~\ref{fig:example-irs-graph}. Its induced relational schema $\rschema$ is visualized in Figure~\ref{fig:example-irs-relational}.
\end{example}

\bfpara{Standard \trans.}
The standard \trans (\sdtname) is expressed in the same language as a general \trans from Section~\ref{sec:transformer}, which transforms instances of a graph schema $\gschema$ to instances of a relational schema $\rschema$. 
Intuitively, the standard database transformer for a  graph schema $\gschema$ converts each node and edge type to a separate table, and all occurrences of that element type in a graph database $\graph$ become  tuples in the corresponding table of the relational database. 
Specifically, as shown in the \textrm{Node} rule, for each node of type $\nodeType = (l, K_1, \ldots, K_n)$, we generate a formula $l(K_1, \ldots, K_n) \to R_l(K_1, \ldots, K_n)$ which transforms a predicate $R_l(v_1, \ldots, v_n)$ representing a graph element  to a predicate $R'_l(v_1, \ldots, v_n)$ representing a tuple in the relational database.
Similarly, we generate a formula $l(K_1, \ldots, K_m, \fk_s, \fk_t) \to R_l(K_1, \ldots, K_m, \fk_s, \fk_t)$ for each edge type in the graph schema.

\begin{example} \label{ex:sdt}
Consider the graph database $\graph$ visualized in Figure~\ref{fig:example-small-sdt-graph}. The \sdtname $\sdt$ transforms $\graph$ to the relational database shown in Figure~\ref{fig:example-small-sdt-relation}.
\begin{figure}[!t]
\centering
\begin{minipage}{\linewidth}
\centering
\begin{subfigure}{0.3\linewidth}
\centering
\normalsize
\begin{tikzpicture}[main/.style = {draw, circle}] 
\node[circle, draw=black] (1) at (0, 1){A};
\node[circle, draw=black] (2) at (0, 0){B};
\node[circle, draw=black, dashed] (3) at (2.3, 1){CS};
\node[circle, draw=black, dashed] (4) at (2.3, 0){EE};

\draw[->] (1) -- (3) node [midway, above, sloped] (cs1) {\scalebox{.7}{:WORK\_AT}};
\draw[->] (2) -- (3) node [midway, below, sloped] (cs2) {\scalebox{.7}{:WORK\_AT}};
\end{tikzpicture} 
\vspace{-5pt}
\caption{A graph database instance.}
\label{fig:example-small-sdt-graph}
\end{subfigure}
~
\centering
\begin{subfigure}{0.6\linewidth}
\centering
\begin{tabular}{|c|c|c|}
\multicolumn{2}{c}{\textbf{emp}} \\
\hline
id & name \\
\hline
1 & A \\
\hline
2 & B \\
\hline
\end{tabular}
~
\begin{tabular}{|c|c|c|c|}
\multicolumn{3}{c}{\textbf{work\_at}} \\
\hline
wid & SRC & TGT \\
\hline
10 & 1 & 1 \\
\hline
11 & 2 & 1 \\
\hline
\end{tabular}
~
\begin{tabular}{|c|c|c|}
\multicolumn{2}{c}{\textbf{dept}} \\
\hline
dnum & dname \\
\hline
1 & CS \\
\hline
2 & EE \\
\hline
\end{tabular}
\caption{The relational tables.}
\label{fig:example-small-sdt-relation}
\end{subfigure}
\end{minipage}
\vspace{-5pt}
\caption{Example graph database and its corresponding relational database over the induced relational schema.}
\vspace{-5pt}
\label{fig:example-small-sdt}
\end{figure}

\end{example}

\subsection{Syntax-Directed Transpilation}
Building upon the standard database transformer (SDT) introduced earlier, we now turn to the core task of translating Cypher queries into corresponding SQL queries over the induced schema. This process presents several challenges, as it involves mapping graph-based operations in Cypher to the relational model in SQL, while ensuring that the original query semantics are preserved. Specifically, Cypher includes features like pattern matching over subgraphs and optional matches, which do not have direct equivalents in SQL. Additionally, Cypher aggregates data over matched subgraphs, whereas SQL aggregates over grouped tuples. Finally, ensuring consistent mappings of graph nodes and edges across the query is critical to maintaining the integrity of references throughout the process.
Our syntax-directed transpilation approach addresses these challenges by converting Cypher queries into SQL queries over the induced relational schema. The key idea is that path patterns in Cypher can be represented by relational joins in SQL. For instance, Cypher's match clauses are translated into SQL inner joins, and optional match clauses map to outer joins. This approach ensures that the pattern matching semantics of Cypher queries are faithfully represented within the relational model, maintaining the integrity of the original graph-based operations.
We now describe a core subset of our syntax-directed transpilation rules, with further details available in Appendix~\ref{sec:transpilation-pred-expr}.

\begin{figure}[!t]
\scriptsize
\begin{mdframed}
\[
\hspace{-10pt}
\begin{array}{c}

\irulelabel
{\begin{array}{c}
\neg \mathsf{hasAgg}(\overline{E}) \quad
\sdt, \rschema \vdash C \transclause \Xset, Q \\
\sdt, \rschema \vdash E_i \transexpr E_i' \quad 1 \leq i \leq |\overline{E}| \\
\end{array}}
{\sdt, \rschema \vdash \sfReturn(C, \overline{E}, \overline{k}) \transquery \Pi_{\rename_{\overline{k}}(\overline{E'})}(Q)}
{\textrm{(Q-Ret)}}

\irulelabel
{\begin{array}{c}
\mathsf{hasAgg}(\overline{E}) \quad
\sdt, \rschema \vdash E_i \transexpr E_i' \quad
1 \leq i \leq |\overline{E}| \\
\sdt, \rschema \vdash C \transclause \Xset, Q \quad
\overline{A} = \mathsf{filter}(\lambda x. \neg \mathsf{IsAgg}(x), \overline{E'}) \\
\end{array}}
{\sdt, \rschema \vdash \sfReturn(C, \overline{E}, \overline{k}) \transquery \ttGroupBy(Q, \overline{A}, \rename_{\overline{k}}(\overline{E'}), \top)}
{\textrm{(Q-Agg)}}

\\ \ \\

\irulelabel
{\begin{array}{c}
\sdt, \rschema \vdash Q \transquery Q' \quad
\sdt, \rschema \vdash k \transexpr a \\
\end{array}}
{\sdt, \rschema \vdash \sfOrderBy(Q, k, b) \transquery \ttOrderBy(Q', a, b)}
{\textrm{(Q-OrderBy)}}

\\ \ \\

\irulelabel
{\begin{array}{c}
\sdt, \rschema \vdash Q_1 \transquery Q_1' \quad
\sdt, \rschema \vdash Q_2 \transquery Q_2' \\
\end{array}}
{\sdt, \rschema \vdash \sfUnion(Q_1, Q_2) \transquery Q_1' \cup Q_2'}
{\textrm{(Q-Union)}}

\irulelabel
{\begin{array}{c}
\sdt, \rschema \vdash Q_1 \transquery Q_1' \quad
\sdt, \rschema \vdash Q_2 \transquery Q_2' \\
\end{array}}
{\sdt, \rschema \vdash \sfUnionAll(Q_1, Q_2) \transquery Q_1' \uplus Q_2'}
{\textrm{(Q-UnionAll)}}

\end{array}
\]
\end{mdframed}
\vspace{-10pt}
\caption{Translation rules for queries.}
\vspace{-10pt}
\label{fig:trans-query}
\end{figure}

\bfpara{Translating queries.}
The translation rules for queries, illustrated in Figure~\ref{fig:trans-query}, use judgments of the form $\sdt, \rschema \vdash Q \transquery Q'$, where a Cypher query $Q$ maps to an SQL query $Q'$ given \sdtname $\sdt$ and induced \trans $\rschema$. Among these rules, handling $\sfReturn$ requires particular attention, as it involves distinguishing between cases with and without aggregate functions.
If there are no aggregation functions in $\overline{E}$, the \textrm{Q-Ret} rule produces a straightforward translation: the Cypher query $\sfReturn(C, \overline{E}, \overline{k})$ becomes a simple SQL projection $\proj_{\rename_{\overline{k}}(\overline{E'})}(Q)$. Here, $Q$ is the translated result of the Cypher clause $C$, and $\overline{E'}$ represents the translated expressions of $\overline{E}$, with all attributes renamed to $\overline{k}$.  In contrast, when aggregation functions appear in $\overline{E}$, the translation shifts to the \textrm{Q-Agg} rule, which requires generating a GroupBy query. This is necessary because SQL uses GroupBy to manage aggregation by partitioning rows based on non-aggregated columns. In the Cypher query $\sfReturn(C, \overline{E}, \overline{k})$, the non-aggregation expressions $\overline{A}$ act as grouping keys, while the aggregated expressions ensure the correct computation of results for each group.

\begin{example} \label{ex:query}
Consider the following Cypher query and the SDT from Example~\ref{ex:sdt}: 
\[
\sfReturn(\sfMatch([(\texttt{n}, \texttt{EMP}), (\texttt{e}, \texttt{WORK\_AT}, \rightarrow), (\texttt{m}, \texttt{DEPT})], \top), [\texttt{m.dname}, \sfCount(\texttt{n.id})], [\texttt{name}, \texttt{num}])
\]
Since there is a $\sfCount$ aggregation in the return query, we apply the \textrm{Q-Agg} rule to translate it to a \ttGroupBy query in SQL.
Specifically, we first apply the \textrm{C-Match1} rule to translate the match clause $\sfMatch([(\texttt{n}, \texttt{EMP}), (\texttt{e}, \texttt{WORK\_AT}, \rightarrow), (\texttt{m}, \texttt{DEPT})], \top)$ into a SQL query $Q$ (explained in Example~\ref{ex:clause}).
Then we translate the returned Cypher expressions \texttt{m.dname} and $\sfCount(\texttt{n.id})$ to their corresponding SQL expressions \texttt{m.dname} and $\ttCount(\texttt{n.id})$. Among these expressions, we find the one that does not contain aggregations, namely \texttt{m.dname}, and use it as the grouping key for \ttGroupBy. Since there is no filtering based on the aggregated results, the \ttHaving clause for \ttGroupBy is not generated.
Therefore, the translated SQL query is 
$\ttGroupBy(Q, [\mathtt{m.dname}], \rename_{[\mathtt{name}, \mathtt{num}]}([\mathtt{m.dname}, \mathtt{Count(n.id)}]), \top)$.
\end{example}

\begin{figure}[!t]
\scriptsize
\begin{mdframed}
\[
\hspace{-10pt}
\begin{array}{c}

\irulelabel
{\begin{array}{c}
\sdt, \rschema \vdash PP \transpattern \Xset, Q \quad
\sdt, \rschema \vdash \pred \transpred \pred' \\
\end{array}}
{\sdt, \rschema \vdash \sfMatch(PP, \pred) \transclause \Xset, \filter_{\pred'}(Q)}
{\textrm{(C-Match1)}}

\irulelabel
{\begin{array}{c}
\sdt, \rschema \vdash C \transclause \Xset, \proj_L(Q) \quad
\Xset' = \Xset \setminus \overline{Y} \cup \overline{Z} \\
\end{array}}
{\sdt, \rschema \vdash \sfWith(C, \overline{Y}, \overline{Z}) \transclause \Xset', \proj_{\rename_{L[\overline{Z}/\overline{Y}]}(L)}(Q)}
{\textrm{(C-With)}}

\\ \ \\

\irulelabel
{\begin{array}{c}
\sdt, \rschema \vdash C \transclause \Xset_1, Q_1 \quad
\sdt, \rschema \vdash PP \transpattern \Xset_2, Q_2 \quad
\sdt, \rschema \vdash \pred \transpred \pred' \quad
\fresh ~ T_1, T_2 \\
\pred'' = \pred' \land \land_{(X:l) \in \Xset_1 \cap \Xset_2} T_1.\PK(R_l) = T_2.\PK(R_l) \text{ where } l(\ldots) \to R_l(\ldots) \in \sdt \\
\end{array}}
{\sdt, \rschema \vdash \sfMatch(C, PP, \pred) \transclause \Xset_1 \cup \Xset_2, \rename_{T_1}(Q_1) \ijoin_{\pred''} \rename_{T_2}(Q_2)}
{\textrm{(C-Match2)}}

\\ \ \\

\irulelabel
{\begin{array}{c}
\sdt, \rschema \vdash C \transclause \Xset_1, Q_1 \quad
\sdt, \rschema \vdash PP \transpattern \Xset_2, Q_2 \quad
\sdt, \rschema \vdash \pred \transpred \pred' \quad
\fresh ~ T_1, T_2 \\
\pred'' = \pred' \land \land_{(X:l) \in \Xset_1 \cap \Xset_2} T_1.\PK(R_l) = T_2.\PK(R_l) \text{ where } l(\ldots,) \to R_l(\ldots) \in \sdt \\
\end{array}}
{\sdt, \rschema \vdash \sfOptMatch(C, PP, \pred) \transclause \Xset_1 \cup \Xset_2, \rename_{T_1}(Q_1) \ljoin_{\pred''} \rename_{T_2}(Q_2)}
{\textrm{(C-OptMatch)}}

\end{array}
\]
\end{mdframed}
\vspace{-10pt}
\caption{Translation rules for clauses.}
\vspace{-10pt}
\label{fig:trans-clause}
\end{figure}

\bfpara{Translating Clauses.}
Unlike translating entire queries, translating individual clauses requires tracking additional information about the node and edge variables used within the clauses. This is necessary to ensure that multiple occurrences of the same variable across different clauses are translated to refer to the same tuple in SQL. As shown in Figure~\ref{fig:trans-clause}, our translation judgments are of the form $\sdt, \rschema \vdash C \transclause \Xset, Q'$, meaning a Cypher clause $C$ is translated into a SQL query $Q'$, and $\Xset$ is the set of all used node and edge variables.

Our translation rules for the $\sfMatch$ and $\sfOptMatch$ clauses are based on the observation that graph pattern matching in Cypher can be emulated using sequences of join operations in SQL. Specifically, the join operations for $\sfMatch$ are inner joins, while those for $\sfOptionalMatch$ are left outer joins. Intuitively, a $\sfMatch$ clause returns no results if there is no matching pattern, mirroring the behavior of inner joins where unmatched tuples are discarded. In contrast, an $\sfOptionalMatch$ clause returns \textsf{null} for missing matches, similar to how outer joins include unmatched rows with \textsf{null} values.

For example, consider the clause $\sfMatch(C, PP, \pred)$. The translation rule \textrm{C-Match2} first translates the preceding clause $C$ into a subquery $Q_1$ and the path pattern $PP$ into another subquery $Q_2$. It also collects the sets of node and edge variables used in $C$ and $PP$, denoted as $\Xset_1$ and $\Xset_2$, respectively. For each common variable $X$ with label $l$, the translation generates a join predicate $T_1.K_1 = T_2.K_1$.
This ensures that occurrences of the same variable in different parts of the clause are correctly matched by joining on their primary keys, effectively referring to the same tuple in the SQL translation.

The \textrm{C-With} rule handles the $\sfWith$ clause by translating $\sfWith(C, \overline{Y}, \overline{Z})$ into a renaming operation in SQL. It generates a query $\proj_{\rename_{L[\overline{Z}/\overline{Y}]}(L)}(Q)$, which projects and renames columns, replacing the old names $\overline{Y}$ with the new names $\overline{Z}$.

\begin{example} \label{ex:clause}
Consider the Cypher clause $\sfMatch([(\texttt{n}, \texttt{EMP}), (\texttt{e}, \texttt{WORK\_AT}, \rightarrow), (\texttt{m}, \texttt{DEPT})], \top)$ and  the  $\sdt$ from Example~\ref{ex:sdt}. 
Based on \textrm{C-Match1}, we use the \textrm{PT-Path} rule to collect all node and edge variables in the pattern, namely $\Xset = \set{(\texttt{n}: \texttt{EMP}), (\texttt{e}: \texttt{WORK\_AT}), (\texttt{m}: \texttt{DEPT})}$, and translate it to a SQL query $\rename_{\mathtt{n}}(\mathtt{emp}) \ijoin_{\mathtt{n.id} = \mathtt{e.SRC}} \rename_{\mathtt{e}}(\mathtt{work\_at}) \ijoin_{\mathtt{e.TGT} = \mathtt{m.dnum}} \rename_{\mathtt{m}}(\mathtt{dept})$.
Thus, the Cypher clause is translated to $\filter_{\top}(\rename_{\mathtt{n}}(\mathtt{emp}) \ijoin_{\mathtt{n.id} = \mathtt{e.SRC}} \rename_{\mathtt{e}}(\mathtt{work\_at}) \ijoin_{\mathtt{e.TGT} = \mathtt{m.dnum}} \rename_{\mathtt{m}}(\mathtt{dept}))$.
\end{example}

\begin{example} \label{ex:clause-opt}
Consider the Cypher clause $\sfOptMatch(C, PP, \pred)$ where $C$ is the $\sfMatch$ clause in Example~\ref{ex:clause}, $PP = [(\texttt{m}, \texttt{DEPT})]$ and the  $\sdt$ from Example~\ref{ex:sdt}. Based on \textrm{C-Match1} and \textrm{C-OptMatch}, we know $\Xset_1 = \set{(\texttt{n}: \texttt{EMP}), (\texttt{e}: \texttt{WORK\_AT}), (\texttt{m}: \texttt{DEPT})}$ and $\Xset_2 = \set{(\texttt{m}: \texttt{DEPT})}$. $C$ and $PP$ are translated to $Q_1 = \filter_{\top}(\rename_{\mathtt{n}}(\mathtt{emp}) \ijoin_{\mathtt{n.id} = \mathtt{e.SRC}} \rename_{\mathtt{e}}(\mathtt{work\_at}) \ijoin_{\mathtt{e.TGT} = \mathtt{m.dnum}} \rename_{\mathtt{m}}(\mathtt{dept}))$ and $Q_2 = \filter_{\pred}(\rename_{\mathtt{m}}(\mathtt{dept}))$, respectively. Since there is a shared node $(\texttt{m}, \texttt{DEPT})$ between $C$ and $PP$, and the primary key of $\texttt{dept}$ is $\texttt{dnum}$, the Cypher clause is translated to $\rename_{T_1}(Q_1) \ljoin_{T_1.\mathtt{dnum} = T_2.\mathtt{dnum}} \rename_{T_2}(Q_2)$.
\end{example}

\begin{figure}[!t]
\scriptsize
\begin{mdframed}
\[
\hspace{-20pt}
\begin{array}{c}

\irulelabel
{\begin{array}{c}
l(K_1, \ldots, K_n) \to R_l(K_1, \ldots, K_n) \in \sdt
\end{array}}
{\sdt, \rschema \vdash (X, l) \transpattern \set{(X:l)}, \rename_{X}(R_l)}
{\textrm{(PT-Node)}}

\\ \ \\

\irulelabel
{\begin{array}{c}
\sdt, \rschema \vdash PP \transpattern \Xset, Q' \quad
(X_3, l_3) = \text{head}(PP) \quad
l_1(\ldots) \to R_{l_1}(\ldots) \in \sdt \\
l_2(\ldots, \fk_s, \fk_t) \to R_{l_2}(\ldots, \fk_s, \fk_t) \in \sdt \quad
l_3(\ldots) \to R_{l_3}(\ldots) \in \sdt \\
\pred = (R_{l_2}.\fk_s = \PK(R_{l_1})) \quad
\pred' = (R_{l_2}.\fk_t = \PK(R_{l_3})) \quad
\integrity(\rschema) \Rightarrow \pred \land \pred' \\
\end{array}}
{\sdt, \rschema \vdash (X_1, l_1), (X_2, l_2, d_2), PP \transpattern \set{(X_1:l_1), (X_2:l_2)} \cup \Xset, \rename_{X_1}(R_{l_1}) \ijoin_{\pred} \rename_{X_2}(R_{l_2}) \ijoin_{\pred'} Q'}
{\textrm{(PT-Path)}}

\end{array}
\]
\end{mdframed}
\vspace{-10pt}
\caption{Translation rules for path patterns.}
\vspace{-10pt}
\label{fig:trans-pattern}
\end{figure}

\bfpara{Translating patterns.}
Figure~\ref{fig:trans-pattern} illustrates the rules for translating patterns, using judgments of the form $\sdt, \rschema \vdash PP \transpattern \Xset, Q'$. This notation indicates that a Cypher pattern $PP$ translates to an SQL query $Q'$, with all node and edge variables in the pattern represented by $\Xset$.

The translation of patterns follows an inductive structure, with two key rules. The base case handles a single node pattern using the \textrm{PT-Node} rule, which maps the node variable $X$ to its corresponding table $R'_l$ in the relational schema derived from \sdtname $\sdt$, renaming it to $X$. The inductive case, represented by the \textrm{PT-Path} rule, addresses more complex patterns where a new node $(X_2, l_2)$ expands an existing sub-pattern $PP$. The rule identifies the connecting node $(X_3, l_3)$ in $PP$ and determines the appropriate joins. It locates the tables $R'_{l_1}, R'_{l_2}, R'_{l_3}$ corresponding to nodes $X_1$, $X_2$, and $X_3$, and constructs join predicates to connect the foreign keys $\fk_s$ and $\fk_t$ in the edge table $R'_{l_2}$ to the primary keys in the source and target node tables, respectively. This approach aligns with the observation that Cypher’s pattern matching naturally corresponds to a series of SQL join operations.


\begin{example}
Given the standard \trans $\transformer$ in Example~\ref{ex:sdt}, let us focus on the path pattern $[(\texttt{n}, \texttt{EMP}), (\texttt{e}, \texttt{WORK\_AT}, \rightarrow), (\texttt{m}, \texttt{DEPT})]$.
According to the \textrm{PT-Path} rule, we first need to apply the \textrm{PT-Node} rule to get the variables $\set{(\texttt{m} : \texttt{DEPT})}$ and query $\rename_{\texttt{m}}(\texttt{DEPT})$ from the node pattern $(\texttt{m}, \texttt{DEPT})$. Based on these results, we can use the \textrm{PT-Path} rule to collect all variables $\Xset = \set{(\texttt{n} : \texttt{EMP}), (\texttt{e} : \texttt{WORK\_AT}), (\texttt{m} : \texttt{DEPT})}$ and obtain the SQL query $\rename_{\mathtt{n}}(\mathtt{emp}) \ijoin_{\mathtt{n.id} = \mathtt{e.SRC}} \rename_{\mathtt{e}}(\mathtt{work\_at}) \ijoin_{\mathtt{e.TGT} = \mathtt{m.dnum}} \rename_{\mathtt{m}}(\mathtt{dept})$.
\end{example}

\begin{theorem}[Soundness of translation]\label{lem:soundness}
Let $\gschema$ be a graph schema and $Q$ be a Cypher query over $\gschema$. Let $\rschema$ be the induced relational schema of $\gschema$, and $\sdt$ be the standard \trans from $\gschema$ to $\rschema$. If $\sdt, \rschema \vdash Q \transquery Q'$, then $Q'$ is equivalent to $Q$ modulo $\sdt$, i.e., $Q \simeq_{\sdt} Q'$.
\end{theorem}

\begin{theorem}[Completeness of translation] \label{thm:completeness}
Let $\gschema$ be a graph schema and $\rschema$ be the induced relational schema of $\gschema$. Given any Cypher query $Q$ over $\gschema$ accepted by the grammar shown in Figure~\ref{fig:cypher-syntax}, there exists a SQL query $Q'$ over $\rschema$ such that $\sdt, \rschema \vdash Q \transquery Q'$.
\end{theorem}

\subsection{Reduction to SQL Equivalence Checking}

\begin{figure}[!t]
\begin{algorithm}[H]
\caption{{Algorithm for inferring the residual of database transformers and \sdtname's}}
\label{algo:infer-rdt}
\begin{algorithmic}[1]
\Procedure{\textsc{ReduceToSQL}}{$\rschema, Q_R, \schema_{R}', Q_{R}', \transformer, \sdt$}
\vspace{2pt}
\Statex \textbf{Input:} Database transformer $\transformer$, standard \trans $\sdt$
\Statex \textbf{Output:} $\top$ for equivalence or $\bot$ indicating failure
\vspace{2pt}
\State $\sigma \gets \{ P_1 \mapsto P_0  \ | \  P_1(\ldots) \rightarrow P_0(\ldots) \in \sdt \} $
\State $\rdt \gets \transformer[\sigma] $
\State \Return $\mathsf{CheckSQL}(\rschema, Q_R, \schema_{R}', Q_{R}', \rdt) $

\EndProcedure
\end{algorithmic}
\end{algorithm}
\vspace{-20pt}
\end{figure}

The final step of our algorithm utilizes the transpiled query to reduce our original problem to checking equivalence between a pair of SQL queries. As shown in Algorithm~\ref{algo:infer-rdt}, the {\sc ReduceToSQL} procedure first infers the residual database transformer $\rdt$ through a simple syntactic substitution: Since every clause of the \sdtname is of the form $P_1(\ldots) \rightarrow P_0(\ldots)$, we can obtain the residual transformer simply by  substituting every occurrence of $P_1$ in $\transformer$ by $P_0$. Finally, since the residual transformer $\rdt$ specifies how to convert instances of the induced schema to those of the desired schema, we can use an existing tool for checking SQL equivalence by utilizing $\rdt$.  As stated by the following theorems, the original Cypher query is equivalent to the given SQL query if and only if $Q_R$ and $Q_R'$ are equivalent modulo $\rdt$. 

\begin{theorem}[Soundness]\label{thm:full-soundness}
Let $\textsf{CheckSQL}(\rschema, Q, \rschema', Q', \rdt)$ be a sound procedure for equivalence checking of SQL queries $Q, Q'$ over relational schemas $\rschema, \rschema'$ connected by \rdtname $\rdt$.
Given a Cypher query $Q_G$ over graph schema $\gschema$, a SQL query $Q_R$ over relational schema $\rschema$, and their \trans $\transformer$, if $\textsc{CheckEquivalence}(\gschema, Q_G, \rschema, Q_R, \transformer)$ returns $\top$, it holds that $Q_G \simeq_{\transformer} Q_R$.
\end{theorem}

\begin{theorem}[Completeness]\label{thm:full-completeness}
Let $\textsf{CheckSQL}(\rschema, Q, \rschema', Q', \rdt)$ be a complete procedure for equivalence checking of SQL queries $Q, Q'$ over schemas $\rschema, \rschema'$ connected by \rdtname $\rdt$.
Given a Cypher query $Q_G$ over graph schema $\gschema$, a SQL query $Q_R$ over relational schema $\rschema$, and their \trans $\transformer$, if $Q_G \simeq_{\transformer} Q_R$, then $\textsc{CheckEquivalence}(\gschema, Q_G, \rschema, Q_R, \transformer)$ returns $\top$.
\end{theorem}

\section{Evaluation} \label{sec:eval}

In this section, we describe three experiments to evaluate \tool. Because  \tool's verification methodology reduces the Cypher-SQL equivalence checking problem to pure SQL,  our results depend on what backend is used for SQL equivalence checking. Thus, in our first experiment, we evaluate \tool using the \verieql~\cite{verieql-oopsla24} bounded model checker as its backend, and in our second experiment, we use a deductive verifier called \mediator~\cite{mediator-popl18} as the backend. Finally, we also evaluate the quality of \tool's transpilation results. 

\begin{table}[!t]
\vspace{-5pt}
\small
\centering
\caption{Statistics of Cypher and SQL queries in the benchmarks. Sizes are the number of AST nodes.}
\vspace{-10pt}
\begin{tabular}{c|c|cccc|cccc|cccc}
\toprule
\multirow{2}{*}{Dataset} &
\multirow{2}{*}{\#} &
\multicolumn{4}{c|}{SQL Size} &
\multicolumn{4}{c|}{Cypher Size} &
\multicolumn{4}{c}{{Transformer Size}} \\
\cline{3-14}
                         &  & Min   & Max   & Avg   & Med  & Min   & Max   & Avg   & Med & Min   & Max   & Avg   & Med\\
\midrule
\stackoverflow & 12 & 15    & 74     & 32.5      & 28     & 20      & 149      & 54.9      & 41    & 1 & 6 & 3.3 & 4 \\
\tutorial      & 26 & 5    & 76      & 25.8       & 22     & 12      & 77      & 31.2      & 28    & 1 & 17 & 8.3 & 5 \\
\academic      & 7  & 27    & 59      & 46.9      & 54     & 45      & 121      & 75.0      & 66    & 5 & 7 & 6.7 & 7 \\
\verieqlset      & 60  & 15    & 143      & 42.2      & 39     & 29      & 174      & 65.8      & 61    & 1 & 10 & 6.7 & 10 \\
{\mediatorset} & 100 & 9 & 63 & 18.6 & 13 & 20 & 114 & 33.8 & 28 & 1 & 11 & 5.1 & 4 \\
{\gpt} & 205 & 5 & 143 & 28.2 & 25 & 12 & 171 & 46.0 & 38 & 1 & 17 & 5.9 & 5 \\
\hline
{Total} & 410 & 5 & 143 & 28.2 & 25 & 12 & 174 & 45.7 & 38 & 1 & 17 & 5.9 & 5 \\
\bottomrule
\end{tabular}
\vspace{-10pt}
\label{tab:statistics}
\end{table}

\vspace{-3pt}
\bfpara{Benchmarks.}
We evaluate \tool on 410 pairs of SQL and Cypher queries (see Table~\ref{tab:statistics}) from the following sources:
\begin{itemize}[leftmargin=*]
\item \stackoverflow: 
We identified 12 StackOverflow posts where users inquire about translating a SQL query to Cypher or vice versa. All of these posts  contain a description of the schemas and SQL/Cypher queries that are intended to be semantically equivalent.

\item \tutorial.
We identified 26 tutorial examples, including from the official Neo4j guide~\cite{neo4j-tutorial}, that explain how a SQL query can be implemented using the Cypher query language. 
\item \academic.
We identified 7 examples from academic papers~\cite{Lin-arxiv16,boudaoud2022towards} that contain relational queries and their corresponding version in Cypher.
\item \verieqlset.
We collected 60 benchmarks from the \verieql paper~\cite{verieql-oopsla24} by randomly sampling 20 queries from each of its three datasets and asking various people with at least 3 months of Cypher experience  to  write an equivalent Cypher query.

\item \mediatorset. We collected 100 benchmarks from the \mediator evaluation set~\cite{mediator-popl18}. Each \mediator benchmark, consisting of an SQL query pair $(Q_1, Q_2)$ over schemas $\schema_1 $and $\schema_2$, was translated into pairs $(G_1, Q_2)$ and $(G_2, Q_1)$ where $G_1, G_2$ are Cypher queries over graph schemas $\schema'_1$ and $\schema'_2$, and $\schema_1$ and $\schema_2$ are the induced relational schemas for $\schema'_1$ and $\schema'_2$.

\item \gpt.
Given that large language models like GPT are increasingly used by people for coding and transpilation tasks, we included GPT-generated Cypher queries to assess \tool's ability to detect errors in automated translations. 
Specifically, we used GPT to transpile SQL queries from the previous five categories, yielding an additional 205 benchmarks.

\end{itemize}

\bfpara{Database transformers.}
Since the induced schema of graph databases may differ from the schema of relational databases, \tool requires a database transformer to describe the relationship between the graph and relational schemas. To evaluate \tool across all pairs of SQL and Cypher queries, one of the authors constructed a database transformer for each query pair based on their schema descriptions. We observe that writing these database transformers is not difficult. As shown in Table~\ref{tab:statistics}, each database transformer consists of an average of 5.9 rules, which takes approximately one minute to write.

\bfpara{Machine configuration.}
All of the experiments are conducted on a laptop with an Intel Core i7-8750H processor and 32GB physical memory running the Debian 12 operating system.

\subsection{Evaluation of \tool with BMC backend} \label{sec:eval-bmc}

In this section, we present the results of the evaluation in which we use the \verieql bounded model checker as \tool's SQL equivalence checking backend. 
\verieql is a bounded model checker that requires a hyperparameter specifying the size bound of symbolic tables. However, since it is difficult to estimate a suitable bound a priori, we set a 10-minute time limit and gradually increase the bound until either a counterexample is found or the time-limit is reached.  For each refuted benchmark, \tool uses the relational counterexamples produced by \verieql to construct a counterexample over the graph schema. 

The results of this evaluation are presented in Table~\ref{tab:eval-results}.  Here, the column labeled ``\# Non-Equiv'' shows the number of benchmarks proven to be \emph{not} equivalent, and the last column shows the average time to find a counterexample. The column labeled ``Avg Checked Bound'' shows the  average size of symbolic tables (measured in terms of the number of rows) when the 10 minute time limit is reached. As shown in Table~\ref{tab:eval-results}, \tool refutes equivalence for 34 out of the 410 benchmarks, taking an average of 23.4 seconds to find a counterexample. For the remaining 376 benchmarks, \tool performs bounded verification, demonstrating that there is no counterexample for database instances with symbolic tables of average size 19.6.  

\begin{table}[!t]
\small
\centering
\caption{Results of bounded equivalence checking.}
\vspace{-5pt}
\begin{tabular}{c|c|c|c|c}
\toprule
Dataset & \# & \# Non-Equiv & Avg Checked Bound & Avg Refutation Time (s) \\
\midrule
\stackoverflow & 12 & 1 & 9.2 & 0.6 \\
\tutorial & 26 & 1 & 7.7 & 56.2 \\
\academic & 7 & 1 & 2.5 & 5.4 \\
\verieqlset & 60 & 4 & 7.2 & 8.5 \\
\mediatorset & 100 & 0 & 33.2 & N/A \\
\gpt & 205 & 27 & 18.7 & 25.9 \\
\hline
Total & 410 & 34 & 19.6 & 23.4 \\
\bottomrule
\end{tabular}
\vspace{-10pt}
\label{tab:eval-results}
\end{table}

\bfpara{Uncovered bugs.}
We have manually inspected all 34 bugs uncovered by \tool and confirmed that all  counterexamples produced by the tool correspond to true positives.   As expected, GPT-generated queries have a higher probability of being incorrect compared to the human-written queries.  In particular, \tool finds a counterexample to equivalence for 13\% of the queries transpiled by GPT.
This experiment shows that GPT may introduce semantic bugs when converting SQL to Cypher, and \tool can effectively identify these bugs. As developers increasingly rely on large language models for assistance with coding-related tasks, we believe this demonstrates \tool’s practical value for developers relying on LLM-generated queries.

Perhaps more surprisingly, \tool also finds incorrect translations among benchmarks in the \texttt{StackOverflow, Tutorial, Academic}, and \texttt{VeriEQL} categories, all of which involve queries transpiled by experts. Most surprisingly, \tool  uncovers  a bug in an example from a Neo4j tutorial~\cite{neo4j-tutorial} that is intended to help developers familiar with SQL to learn Cypher.  This tutorial contains several pairs of SQL and Cypher queries that are intended to be equivalent, but, for one of these examples, \tool finds a counterexample on which the query results are actually different. We refer the interested reader to Appendix~\ref{sec:qual-analysis} for a case study explaining some of the uncovered bugs, including the example from the Neo4j tutorial.

\bfpara{False negatives.} Since \verieql is a bounded model checker, benchmarks that are not refuted by \tool within the 10-minute time limit \emph{may} still contain bugs.  To assess how frequently this occurs, we sampled {50} pairs of queries that were not refuted by \tool and manually inspected whether the translation was correct. 
For {48} of the {50} manually-inspected benchmarks, we found that the translation is indeed correct, but for 2 benchmarks (both from the \texttt{GPT-Translate} category), the translation is incorrect but \tool fails to find a counterexample within the 10 minute time limit.  Thus, while \tool with the \verieql backend does not provide theoretical soundness guarantees, we find that it is useful for finding bugs and has a low chance of missing incorrect translations ({4\%} according to our manual inspection results).

\textbox{
\textbf{Key finding:} 
Using a bounded model checker backend, \tool identifies {27} bugs among 205 SQL queries transpiled to Cypher using GPT. More surprisingly, among the 205 manually-written query pairs that are meant to be equivalent, \tool also identifies {7} inconsistencies, including 3 benchmarks from the wild and 4 benchmarks from manual translations. 
}

\subsection{Evaluation of \tool with Deductive Verifier}

In this section, we present the results of a second experiment wherein we evaluate \tool with \mediator as its backend. As mentioned earlier, \mediator is an SMT-based deductive verifier for reasoning about SQL applications over different schemas. Unlike \verieql, \mediator can perform full-fledged verification; however, it supports a  limited subset of SQL without aggregations or outer joins. Additionally, unlike \verieql, \mediator cannot disprove equivalence by generating counterexamples. Hence, when using \tool with \mediator as its backend, \tool can either prove equivalence or it returns ``Unknown''.
For performing this experiment, we also use  a time limit  of 10 minutes per benchmark.

\begin{table}[!t]
\small
\centering
\caption{Results of full equivalence verification.}
\label{tab:eval-full-verification}
\vspace{-10pt}
\begin{tabular}{c|c|c|c|c|c}
\toprule
Dataset & \# & \# Supported & \# Verified & \# Unknown & Avg Time (s) \\       
\midrule
\stackoverflow & 12 & 1 & 1 & 0 & 1.0 \\
\tutorial & 26 & 1& 1 & 0 &  0.2 \\
\academic & 7 & 0 & 0 & 0 &  N/A \\
\verieqlset & 60 & 0 & 0 & 0 &  N/A \\
\mediatorset & 100 & 100 & 77 & 23 & 20.5 \\
\gpt & 205 & 94 & 73 & 21 &  23.5 \\
\hline
Total & 410 & 196 & 152 & 44 &  16.8 \\
\bottomrule
\end{tabular}
\vspace{-15pt}
\end{table}


The results of this experiment are summarized in Table~\ref{tab:eval-full-verification}. 
As shown in the ``\# Supported'' column, about half of the benchmarks (196 out of 410) fall inside the fragment of SQL supported by  \mediator, so we conduct our evaluation on this subset. 
Overall, \tool can verify {77.6\%} of these 196 benchmarks, with an average running time of 16.8 seconds. Since the syntax-directed transpilation performed by \tool takes negligible time, most of the verification time is dominated by SMT queries for discharging the generated verification conditions.


\bfpara{Qualitative analysis.}
We manually inspected 44 benchmarks that cannot be verified by \tool with the \mediator backend. Among these 44 benchmarks, two of them are in fact \emph{refuted} by \tool with \verieql, so they should \emph{not} be verified. 
Among the 42 remaining benchmarks, \mediator fails to complete verification within the 10 minute time limit for 14 of these, and, for the final 28 benchmarks, it terminates but returns ``Unknown''. 
To gain insight about failure cases, note that \mediator needs to infer an \emph{inductive bisimulation invariant} between the two queries~\cite{mediator-popl18}. However, for queries involving long join chains, the corresponding bisimulation invariant can be complex. This can either lead to expensive SMT queries, thereby causing time-outs, or the required invariant might fall outside the inference capabilities of \mediator. 

\vspace{5pt}
\textbox{
\textbf{Key finding:}  Among the 196  SQL queries supported by \mediator, \tool can prove equivalence between roughly 80\% of (Cypher, SQL) query pairs using the \mediator backend. 
}

\subsection{Evaluation of Transpilation}



While the primary goal of this work is to enable checking equivalence between graph and relational queries, an additional benefit of our method is that it can  transpile graph database queries to relational queries over the induced schema.
To assess the practical effectiveness, we conduct an experiment to evaluate how well our method transpiles graph queries into \emph{efficient} SQL queries.

\begin{table}[!t]
\small
\centering
\caption{Execution time of transpiled and manually-written SQL queries.}
\label{tab:eval-transpilation}
\vspace{-10pt}
\resizebox{\columnwidth}{!}{
\begin{tabular}{c|c|cc|c|c|c|c}
\toprule
\multirow{2}{*}{Dataset} & \multirow{2}{*}{\#} & \multicolumn{2}{c|}{Avg Exec Time (s)} & \% Transpiled & \% Trans Slower & \% Trans Slower & \% Trans Slower \\ \cline{3-4}
& & Transpiled & Manual & Faster & (1x, 1.1x] & (1.1x, 1.2x] & (1.2x, +$\infty$) \\
\midrule
\stackoverflow & 12 & 1.9 & 1.8 & 41.7\% & 8.3\% & 50.0\% & 0.0\% \\
\tutorial & 26 & 2.3 & 1.7 & 19.2\% & 11.5\% & 46.2\% & 23.1\% \\
\academic & 7 & 2.4 & 3.2 & 71.4\% & 0.0\% & 14.3\% & 14.3\% \\
\hline
Total & 45 & 2.2 & 2.0 & 33.3\% & 8.9\% & 42.2\% & 15.6\% \\
\bottomrule
\end{tabular}
}
\vspace{-20pt}
\end{table}

\bfpara{Efficiency of transpilation.} First, we evaluate how long \tool takes to transpile each Cypher query to a SQL query over the induced schema.  \tool can successfully transpile all 410 queries, and  the average, median, and maximum transpilation times are 6.3,  3.0, and 180.2 milliseconds respectively. Hence, we can conclude that transpilation is very fast in practice.

\bfpara{Quality of transpiled queries.} Next, we also set out to evaluate the quality of the transpiled queries by comparing the execution time of manually written SQL queries against \tool's transpilation result. However, performing this evaluation is challenging for two reasons: First, we only have access to the ``ground truth'' Cypher and SQL queries for some benchmark categories.  Second,  we do not have database instances that these queries are meant to be executed on.  To deal with the first challenge, we conduct this evaluation only on those benchmarks from the \stackoverflow, \tutorial, and \academic categories for which we are given the original Cypher query and its SQL equivalent (or vice versa).
To deal with the second challenge, we generate mock database instances and assess query efficiency on them. For each benchmark with SQL query \( Q_R \) over schema \( \rschema \) and Cypher query \( Q_G \), we use \tool to transpile \( Q_G \) into SQL query \( Q'_R \) over the induced schema \( \rschema' \). We then create databases \( \reldb \) and \( \reldb' \) over \( \rschema \) and \( \rschema' \), respectively, ensuring \( \rdt(\reldb') = \reldb \).
To account for execution time variability, we start with 10,000 tuples in each table of \( \reldb \) and iteratively increase the table size by 10x, up to 1 million, choosing the largest size where manually written SQL queries run within 10 seconds. This approach results in 1 million tuples for 36 benchmarks and between 10,000 and 1 million for the remaining 9 benchmarks.
Finally, we measure the execution times of \( Q'_R \) on \( \reldb' \) and \( Q_R \) on \( \reldb \). Table~\ref{tab:eval-transpilation} summarizes the results: for 33.3\% of the benchmarks, the transpiled queries are faster than the manually written queries. For the remaining benchmarks, 8.9\% exhibit a slowdown of no more than $1.1$x, 42.2\% exhibit a slowdown between $1.1$x and $1.2$x, and 15.6\% exceed $1.2$x.
These results indicate that using \tool to perform automated transpilation could be useful in real-world scenarios that necessitate graph-to-relational query conversion, such as for legacy systems, resource-constrained environments, or data integration. 

\textbox{
\textbf{Key finding:} \tool can transpile Cypher queries to SQL queries in milliseconds. The execution time of transpiled SQL queries is faster than manually-written queries on 33.3\% of benchmarks and within 1.2x slowdown on 51.1\% of benchmarks.
}

\section{Related Work} \label{sec:related}

In this section, we discuss prior work that is mostly related to our techniques for equivalence verification between Cypher and SQL queries.

\bfpara{Automated reasoning for SQL.}
Despite the undecidability of checking equivalence between SQL queries \cite{trachtenbrot1950}, there has been much prior work on automated reasoning for relational queries. We can categorize existing work into three classes. The first class targets a decidable subset of SQL and proposes decision procedures for that subset. Examples of work in this category include \cite{aho-sagiv-ullman-1979,chandra-merlin-1977,green2011}. Approaches in the second category propose sound but incomplete algorithms for an undecidable subset of SQL; examples of work in this space include \cite{zhou19,zhou22,chu2017hottsql,chu2018axiomatic}. The third category performs bounded verification  to find bugs in SQL queries; examples include \cosette~\cite{cosette}, \qex~\cite{veans-qex-2010}, and \verieql~\cite{verieql-oopsla24}.  There is also prior work on verifying relational database applications that involve both queries and updates~\cite{mediator-popl18}.    Our proposed approach reduces the verification problem between graph and relational queries to checking equivalence between  a pair of relational queries; as such, it can leverage any future advances in this area. 

\bfpara{Migration between database instances.} 
There is prior work on migrating data between different schemas, including~\cite{jin2017foofah,yaghmazadeh2016synthesizing,feng2017component,martins2019trinity,wang2017synthesizing,yaghmazadeh2018automated,dynamite-vldb20}. The most related to this paper is \dynamite~\cite{dynamite-vldb20}, which automates data migration between graph and relational databases. However, \dynamite is only useful for migrating the \emph{contents} of the database and cannot be used for transpiling queries. There are also prior papers that address the query transpilation problem in the context of SQL~\cite{migrator-pldi19,prism-vldb08,prism-vldbj13}.
While our notion of \trans is inspired by prior work~\cite{dynamite-vldb20,clio-vldb00}, to the best of our knowledge, this paper is the first to formalize the transpilation procedure from Cypher to SQL queries and leverage it for formal equivalence checking.


\bfpara{Data representation refactoring.} 
There is a related line of work on \emph{data representation refactoring}, which aims to refactor programs or specifications from one data representation to another~\cite{pailoor2024semantic,chen2022synthesis,pailoor2021synthesizing,cheung2013optimizing,ge2012reconciling,katara-oopsla22,refactor-jscico13,atropos-pldi21,fiat-popl15}. For instance, Solidare~\cite{pailoor2024semantic} refactors smart contracts between different ADTs. QBS~\cite{chen2022synthesis} converts Java programs operating over collections to SQL queries.
Since graph and relational schemas can be viewed as different data representations, the query transpilation problem in this work can be viewed as a form of data representation refactoring.
However, none of the existing techniques addresses the query transpilation problem between graph and relational data. Additionally, this work can be viewed as presenting a novel methodology for verifying data representation refactoring: Rather than directly going from representation $R$ to representation $R'$, our idea is to introduce an auxiliary representation $R''$ that simplifies the problem by enabling syntax-directed translation.
To the best of our knowledge, such a methodology based on a layer of indirection has not previously been explored in this context.

\bfpara{Graph database query languages.}
There has been a long line of prior work on graph databases and semantic foundations of graph query languages~\cite{cypher-sigmod18,graphql-www18,pgql-grades16,Deutsch22,age-web24,gremlin-dbpl15,sparql-tapp11}. 
The unifying insight behind many of such works is that a graph database schema may be viewed as a graphical representation of the Entity-Relationship Diagram (ER Diagram) of a relational database schema.
In part due to this unifying insight, these graph query languages are both semantically and syntactically similar to Cypher. Because of this similarity, we believe our proposed methodology can be adapted fairly easily to graph database query languages other than Cypher.

\bfpara{Testing database queries.}
Another related line of related work focuses on testing database queries. Work in this space includes differential testing~\cite{mckeeman1998differential,lin2022gdsmith,zheng2022finding} and metamorphic testing~\cite{chen2020metamorphic,rigger2020finding,jiang2024detecting,mang2024testing} to detect bugs in database management systems, mutation-based testing to grade programming assignments involving SQL queries~\cite{xdata-vldbj15}, and provenance-based techniques for explaining wrong SQL queries~\cite{ratest-sigmod19}. In contrast, \tool transpiles graph database queries into equivalent SQL queries and uses existing automated reasoning tools for SQL equivalence checking. Thus, \tool is complementary and can benefit from advances in testing SQL queries. 

\bfpara{Transpiling Cypher queries.}
There are a few existing tools that can translate Cypher queries to queries in SQL-like languages~\cite{oct-web25, kuzu-web25}. 
The most relevant to this paper is \opentranspiler~\cite{oct-web25}, which first transforms a Cypher query into a logical plan and renders it as a relational query. However, it supports only a limited subset of Cypher and lacks soundness guarantees for translated queries.\footnote{The detailed evaluation of \opentranspiler can be found in Appendix~\ref{sec:eval-oct}.} In contrast, \tool ensures soundness during transpilation and supports a broader subset of Cypher queries. Another tool, \kuzu~\cite{kuzu-web25}, can execute graph queries on relational databases with a Cypher interface. It compiles Cypher queries into an intermediate representation similar to a relational database's logical plan. However, \kuzu does not transpile Cypher into SQL but instead directly executes the intermediate representation on the database.

\section{Limitation} \label{sec:limit}
The current version of \tool is focused on a specific subset of SQL and Cypher, which does not yet cover all modern features, such as variable-length pattern matching in Cypher. However, considering the lack of prior research on reasoning about equivalence between graph and relational database queries, we believe our selected fragments offer a strong foundation for advancing this area of study. While some SQL and Cypher queries lie outside the scope of our current subset, evaluations on a diverse set of benchmarks, including real-world queries, demonstrate that this subset is expressive enough for practical use cases. Future work can further extend the transpilation rules and backend equivalence verifiers to support additional features.

\section{Conclusion and Future Work} \label{sec:concl}

In this paper, we proposed automated reasoning techniques between graph and relational database queries. Specifically, we first proposed a formal definition of equivalence between graph and relational queries and used it as the basis of a correct-by-construction transpilation strategy for converting Cypher queries to SQL queries over a so-called \emph{induced relational schema}. We then showed how our translation approach can be used to check equivalence between graph and SQL queries over \emph{arbitrary} schema by leveraging existing automated reasoning techniques for SQL.
We have also evaluated our implementation, \tool, on equivalence checking tasks involving real-world Cypher and SQL queries and showed that \tool can be useful for (a) uncovering subtle bugs in Cypher queries that are meant to be equivalent to a reference SQL implementation, and (b) verifying full equivalence between Cypher and SQL queries.

Looking ahead, we plan to explore the development of a graphical interface for specifying database transformers between graph and relational databases, inspired by prior work on schema mapping visualization~\cite{visual-chi05}. This interface would aim to further reduce the manual effort required by users to verify equivalence between graph and relational queries, providing a more intuitive and user-friendly approach. We see this as a promising avenue for future research.

\bibliography{main}

\newpage
\appendix
\section{Semantics of Cypher Queries} \label{sec:cypher-semantics}

\begin{figure}[!h]
\scriptsize
\[
\begin{array}{l c l}
\multicolumn{3}{l}{\boxed{\denot{Q} :: \sfGraph \to \sfTable}} \\[5pt] 
\denot{\mathsf{OrderBy}(R, k, b)}_G & = & \sffoldl(\lambda xs. \lambda \_. xs \doubleplus [\mathsf{MinTuple}(E, b, \denot{R}_{G} - xs)], [], \denot{R}_G) ~\sfwhere \\
&& \hspace{3em} \mathsf{MinTuple}(E, b, xs) = \sffoldl(\lambda x. \lambda y. \sfite(\mathsf{Cmp}(E, b, x, y), x, y), \sfHead(xs), xs), \\ 
&& \hspace{3em} \mathsf{Cmp}(E, b, x_1, x_2) = (\mathsf{lookup}(x_1, k) < \mathsf{lookup}(x_2, k)) = b \\
\denot{\mathsf{Union}(Q_1, Q_2)}_G & = & \sfDedup(\denot{\mathsf{UnionAll}(Q_1, Q_2)}_G) \\ 
\denot{\mathsf{UnionAll}(Q_1, Q_2)}_G & = & \denot{Q_1}_G \doubleplus \denot{Q_2}_G \\ 
\\

\multicolumn{3}{l}{\boxed{\denot{R} :: \sfGraph \to \sfTable}} \\[5pt] 
\denot{\sfReturn(C, \overline{E}, \overline{k})}_G & = & \sfite(\neg \sfHasAgg(\overline{E}), \sfmap(\lambda g. \sfmap(\lambda (E, k). (k, \denot{E}_{G, [g]}), \sfzip(\overline{E}, \overline{k})), \denot{C}_{G}), \\
&& \hspace{3em} \sfmap(\lambda gs. \sfmap(\lambda (E, k), (k, \denot{E}_{G, gs}), \sfzip(\overline{E}, \overline{k})), \mathsf{Groups}) ) ~\sfwhere~ \\
&& \hspace{6em} \mathsf{Groups} = \sfmap(\lambda V.\sfFilter(\lambda g. \denot{\overline{A}}_{G, [g]} = V, \denot{C}_{G}), \sfDedup(\overline{V})), \\
&& \hspace{6em} \overline{A} = \sfFilter(\lambda E. \neg \mathsf{IsAgg}(E), \overline{E}), \\
&& \hspace{6em} \overline{V} = \sfmap(\lambda g. \denot{\overline{A}}_{G, [g]}, \denot{C}_{G})\\
\\

\multicolumn{3}{l}{\boxed{\denot{C} :: \sfGraph \to [\sfGraph]}}   \\[5pt]

\denot{\sfMatch(\pattern, \pred)}_G & = & \sfFilter(\lambda g. \denot{\phi}_{G, [g]} = \top, \denot{\pattern}_{G})\\
\denot{\sfMatch(C, \pattern, \pred)}_G & = & \sfFilter(\lambda g. \denot{\pred}_{G, [g]} = \top, \sfmap(\lambda g_1. \sfmap(\lambda g_2. \sfMerge(g_1, g_2),  \denot{\pattern}_{G}), \denot{C}_{G}))\\ 

\denot{\sfOptMatch(C, PP, \pred)}_G & = & \sffoldl(\lambda gs. \lambda g. gs \doubleplus \sfite(|v_1(g)| = 0, v_2(g), v_1(g)), [], \denot{C}_G) ~\sfwhere~ \\ 
&& \hspace{3em} v_1(g) = \sfFilter(\lambda x. \denot{\pred}_{G, [x]} = \top, \\
&& \hspace{6em} \sfmap(\lambda g'. \sfMerge(g, g'), \sfFilter(\lambda g''. g \cap g'' \neq \emptyset, \denot{\pattern}_{G}))), \\
&& \hspace{3em} v_2(g) = [\sfMerge(g, \sfNullify(\sfHead(\denot{\pattern}_{G})))] \\
\denot{\mathsf{With}(C, \overline{X}, \overline{Y})}_G & = & \sfmap(\lambda (N, E, P, T). (N, E, P[\overline{X} \mapsto \overline{Y}], T[\overline{X} \mapsto \overline{Y}]), \denot{C}_{G}) \\
\\ 

\multicolumn{3}{l}{\boxed{\denot{\pattern} :: \sfGraph \to [\sfGraph]}} \\[5pt]

\denot{\nodepattern}_{G} & = & \sfmap(\lambda n. (\set{n}, \emptyset, \set{(X, k) \mapsto P(n, k) ~|~ k \in \keys(T(n)) }, \set{X \mapsto (l, \keys(T(n)))}), \\
&& \hspace{3em} [n \in N ~|~ \lbl(T(n)) = l]) ~\sfwhere~ (N, E, P, T) = G, (X, l) = \nodepattern \\
\denot{\nodepattern, \edgepattern, \pattern'}_{G} & = & \sffoldl(\lambda gs. \lambda g'. gs \doubleplus \sfmap(\lambda g. \sfMerge(g, g'), \\
&& \hspace{3em} \sfFilter(\lambda g''. g' \cap g'' \neq \emptyset, \mathsf{Subgraphs}(G, [\nodepattern, \edgepattern, \nodepattern']))), [], \denot{\pattern'}_{G}) \\
&& \hspace{6em} \sfwhere~ \nodepattern' = \sfHead(\pattern') \\
\\

\multicolumn{3}{l}{\boxed{\denot{E} :: \sfGraph \to [\sfGraph] \to \mathsf{Value}}} \\[5pt]

\denot{k}_{G, gs} & = & \mathsf{lookup}(\sfHead(gs), k) \\
\denot{v}_{G, gs} & = & v \\
\denot{\sfCast(\pred)}_{G, gs} & = & \sfite(\denot{\pred}_{G, gs} = \sfNull, \sfNull, \sfite(\denot{\pred}_{G, gs} = \top, 1, 0)) \\
\denot{\sfCount(E)}_{G, gs} & = & \sfite(\bigwedge_{g \in gs} \denot{E}_{G, [g]} = \sfNull, \sfNull, \sffoldl(+, 0, \sfmap( \lambda g. \sfite(\denot{E}_{G, [g]} = \sfNull, 0, 1), gs))) \\
\denot{\sfSum(E)}_{G, gs} & = & \sfite(\bigwedge_{g \in gs} \denot{E}_{G, [g]} = \sfNull, \sfNull, \\
&& \hspace{3em} \sffoldl(+, 0, \sfmap(\lambda g. \sfite(\denot{E}_{G, [g]} = \sfNull, 0, \denot{E}_{G, [g]}), gs))) \\
\denot{\sfAvg(E)}_{G, gs} & = & \denot{\sfSum(E)}_{G, gs}/\denot{\sfCount(E)}_{G, gs} \\
\denot{\sfMin(E)}_{G, gs} & = & \sfite(\bigwedge_{g \in gs} \denot{E}_{G, [g]} = \sfNull, \sfNull, \\
&& \hspace{3em} \sffoldl(\mathbf{min}, + \infty, \sfmap(\lambda g. \sfite(\denot{E}_{G, [g]} = \sfNull, + \infty, \denot{E}_{G, [g]}), gs))) \\
\denot{\sfMax(E)}_{G, gs} & = & \sfite(\bigwedge_{g \in gs} \denot{E}_{G, [g]} = \sfNull, \sfNull, \\
&& \hspace{3em} \sffoldl(\mathbf{max}, - \infty, \sfmap(\lambda g. \sfite(\denot{E}_{G, [g]} = \sfNull, - \infty, \denot{E}_{G, [g]}), gs))) \\
\denot{E_1 \arithop E_2}_{G, gs} & = & \denot{E_1}_{G, gs} \arithop \denot{E_1}_{G, gs} \\
\\

\multicolumn{3}{l}{\boxed{\denot{\pred} :: \sfGraph \to [\sfGraph] \to \mathsf{Bool} \cup \set{\sfNull}}} \\[5pt]

\denot{\top}_{G, gs} &=& \top \\
\denot{\bot}_{G, gs} &=& \bot \\
\denot{E_1 \logicop E_2}_{G, gs} & = & \denot{E_1}_{G, gs} \logicop \denot{E_1}_{G, gs} \\
\denot{\sfIsNull(E)}_{G, gs} & = & \sfite(\denot{E}_{G, gs} = \sfNull, \top, \bot) \\
\denot{E \in \overline{v}}_{G, gs} & = & \bigvee_{v \in \overline{v}} \denot{E}_{G, gs} = v \\
\denot{\sfExists(\pattern)}_{G, gs} & = & \bigvee_{g \in \denot{\pattern}_{G}} \bigwedge_{K \in \overline{K}} \denot{K}_{G, gs} = \denot{K}_{G, [g]} ~\sfwhere~ \overline{K} = [\mathsf{PK}(\sfHead(\pattern)), \mathsf{PK}(\sfLast(\pattern))] \\
\denot{\pred_1 \land \pred_2}_{G, gs} & = & \denot{\pred_1}_{G, gs} \land \denot{\pred_2}_{G, gs} \\
\denot{\pred_1 \lor \pred_2}_{G, gs} & = & \denot{\pred_1}_{G, gs} \lor \denot{\pred_2}_{G, gs} \\
\denot{\neg \pred}_{G, gs} & = & \neg \denot{\pred}_{G, gs} \\

\end{array}
\]


\caption{Semantics of Cypher queries. Here, $\sfIte(\pred, E_1, E_2)$ denotes the if-then-else expression that evaluates to $E_1$ when $\pred$ is true and $E_2$ otherwise. $\sfDedup(L)$ removes all duplicated elements from list $L$. $\sfHasAgg(\overline{E})$ checks if there is an aggregation expression in $\overline{E}$. $\sfMerge(G_1, G_2)$ merges graphs $G_1$ and $G_2$ into one graph. $\sfNodes(G)$ returns all nodes of graph $G$. $\sfNullify(G)$ yields a graph isomorphic to $G$ with all values set to be $\sfNull$. $\mathsf{Subgraphs}(G, [\nodepattern_1, \edgepattern, \nodepattern_2])$ returns all subgraphs of $G$ that match the simple pattern $\nodepattern_1, \edgepattern, \nodepattern_2$.}
\label{fig:cypher-semantics}
\vspace{-15pt}
\end{figure}

\bfpara{Query semantics.}
Figure~\ref{fig:cypher-semantics} presents the formal semantics of featherweight Cypher using standard functional combinators such as $\sfMap$, $\sfFilter$, $\sfFoldl$, $\sfHead$, and $\sfZip$.
In accordance with the standard Cypher semantics~\cite{cypher-sigmod18},  we view a query $Q$ as a function from a graph database instance to a table; hence, $\denot{Q}$ takes as input a property graph $G$ and outputs a table.  As shown in Figure~\ref{fig:cypher-semantics}, the semantics of $\sfUnion$, $\sfUnionAll$, and $\sfOrderBy$ in Cypher are analogous to their SQL counterparts.
The semantics of a $\sfReturn(C, \overline{E}, \overline{k})$ query is denoted  $\denot{\sfReturn(C, \overline{E}, \overline{k})}_G$: Given a property graph $G$, it first obtains a list of subgraphs by evaluating clause $C$ and then transforms the result into a table using arguments $\overline{E}$ and  $\overline{k}$. Here, $\overline{E}$ specifies a list of expressions to be evaluated against the subgraphs, each of which forms a row of the table. $\overline{k}$ specifies the corresponding list of  column names of the return value.

\bfpara{Clause semantics.}
The semantics of a Cypher clause $C$ is denoted as $\denot{C}$: A clause transforms an input property graph $G$ into a list of subgraphs, where each subgraph corresponds to the  result of pattern matching against the input graph $G$. There are four different types of clauses in Cypher: First, the clause $\sfMatch(\pattern, \phi)$ returns all subgraphs of $G$ that match the pattern $\pattern$ and that satisfy predicate $\phi$. Second, the clause $\sfMatch(C, \pattern, \phi)$ evaluates the clause $C$ and the match clause $\sfMatch(\pattern, \phi)$ on the input graph $G$ and merges their results. 
Third, the optional match clause $\sfOptMatch$ is similar to the previous match clause, except that it will  use a null value for missing parts of the pattern. Finally, the fourth type of clause, namely $\sfWith(C, \overline{X}, \overline{Y})$, first applies the nested clause $C$ to the input graph $G$ and then renames nodes and edges as specified by the old and new names $\overline{X}$ and $\overline{Y}$ respectively. We further illustrate the semantics of $\sfOptMatch$ through an example.

\begin{example}
{Figure~\ref{fig:optmatch-example-graph} shows a graph database instance, where employee A works at the CS department, but employee B has no department. Consider the following Cypher query
\[
\begin{tabular}{l}
$\sfReturn(\sfOptMatch(\sfMatch([(\mathsf{n}, \mathsf{EMP})], \top), ~ [(\mathsf{n}, \mathsf{EMP}), (\mathsf{e}, \mathsf{WORK\_AT}, \rightarrow), (\mathsf{m}, \mathsf{DEPT})], \top),$ \\
\hspace{3em} $[\mathtt{n.name}, \mathtt{m.dname}], [\mathtt{n.name}, \mathtt{m.dname}])$ \\
\end{tabular}
\]
Running this query on the graph database yields the result in Figure~\ref{fig:optmatch-example-output}. The idea is that since there is no outgoing edge from the node representing employee B, its corresponding "department" entry in the result table is $\sfNull$.}

\begin{figure}[t]
\centering
\begin{minipage}{.4\linewidth}
\centering
\begin{subfigure}{\linewidth}
\centering
\normalsize
\begin{tikzpicture}[main/.style = {draw, circle}] 
\node[circle, draw=black] (1) at (0, 1){A};
\node[circle, draw=black] (2) at (0, 0){B};
\node[circle, draw=black, dashed] (3) at (2.3, 1){CS};

\draw[->] (1) -- (3) node [midway, above, sloped] (cs1) {\scalebox{.7}{:WORK\_AT}};
\end{tikzpicture} 
\vspace{-5pt}
\caption{The input property graph.}
\label{fig:optmatch-example-graph}
\end{subfigure}
\end{minipage}
~
\begin{minipage}{.4\linewidth}
\centering
\begin{subfigure}{\linewidth}
\centering
\begin{tabular}{|c|c|c|}
\hline
n.name & m.dname \\
\hline
A & CS \\
\hline
B & Null \\
\hline
\end{tabular}
\caption{The output of a Cypher query.}
\label{fig:optmatch-example-output}
\end{subfigure}
\end{minipage}
\vspace{-10pt}
\caption{An input property graph and the output of the example Cypher query.}
\label{fig:optmatch-example}
\end{figure}

\end{example}

\bfpara{Semantics of path patterns.}
Given a graph database instance $G$ and path pattern $\pattern$, $\denot{\pattern}_G$ returns a list of subgraphs of $G$ that match the given pattern. A path pattern can have two forms. A \emph{node pattern} $NP$ is a pair $(X, l)$ where $l$ is the type of the node and $X$ is a variable for the matched node that can be used in the rest of the query. In this case, the output of pattern matching is a list of single-node subgraphs, with each subgraph containing a node of type $l$. A path pattern $\pattern$ can also be of the form $\nodepattern,\edgepattern,\pattern$ where $\edgepattern$ is an edge pattern $(X, l, d)$ where $l$ is the type of the edge, $d$ is its direction, and $X$ is a variable for the matched edge that can be used in the rest of the query.  For example, the  pattern $\nodepattern_1, \edgepattern, \nodepattern_2$  would match all edges that match the edge pattern $\edgepattern$ and where additionally the source and target nodes match the node patterns $\nodepattern_1$ and $\nodepattern_2$.

\bfpara{Semantics of expressions.}
The semantics of an expression $E$ is denoted as $\denot{E}_{G,gs}$ where $G$ is an input graph database instance and $gs$ is a list of subgraphs (of $G$) denoting the result of performing a path pattern matching on $G$. $\denot{E}_{G,gs}$ evaluates the expression $E$ on the subgraphs $gs$ and yields a value. 
An expression $E$ can be a property key $k$, a value $v$, aggregation expressions, a $\mathsf{Cast}(\phi)$ expression, and arithmetics $E \arithop E$. For example, a property key $k$ is evaluated by looking up the value of this key in the first subgraph of the list of matched subgraphs $gs$.\footnote{Our semantics guarantees that $gs$ is a list of length 1 when passed to $\denot{k}$.}  
\emph{Aggregation expressions} involve the $\mathsf{Count,Sum,Avg,Min,Max}$ operators and their semantics are similar to SQL, where an aggregation operator is used to calculate an output value over the return values of the sub-expression. 
Finally, the $\mathsf{Cast(\phi)}$ expression simply casts  a predicate $\phi$ to 0, 1, or $\sfNull$.

\bfpara{Semantics of predicates.}
Given a predicate $\phi$, property graph $G$, and list of subgraphs $gs$, $\denot{\phi}_{G,gs}$ defines how to evaluate $\phi$ on $G, gs$. Similar to SQL, we take the three-valued logic interpretation of a boolean predicate $\phi$, meaning that the evaluation result can contain $\mathsf{Null}$, and we can perform boolean arithmetics involving $\top, \bot, \mathsf{Null}$ as terms. In particular,  $\bot \land \sfNull = \bot$ and $\top \lor \sfNull = \top$; otherwise, the  evaluation result is $\sfNull$ if any of its arguments are $\sfNull$.

\section{Transpilation of Cypher Predicates and Expressions} \label{sec:transpilation-pred-expr}

\begin{figure}[!t]
\scriptsize
\begin{mdframed}
\[
\begin{array}{c}

\irulelabel
{\begin{array}{c}
\end{array}}
{\sdt, \rschema \vdash k \transexpr k}
{\textrm{(E-Prop)}}

\irulelabel
{\begin{array}{c}
\end{array}}
{\sdt, \rschema \vdash v \transexpr v}
{\textrm{(E-Value)}}

\irulelabel
{\begin{array}{c}
\sdt, \rschema \vdash \pred \transpred \pred'
\end{array}}
{\sdt, \rschema \vdash \sfCast(\pred) \transexpr \ttCast(\pred')}
{\textrm{(E-Pred)}}

\\ \ \\

\irulelabel
{\begin{array}{c}
\sdt, \rschema \vdash E \transexpr E'
\end{array}}
{\sdt, \rschema \vdash \sfAgg(E) \transexpr \ttAgg(E')}
{\textrm{(E-Agg)}}

\irulelabel
{\begin{array}{c}
\sdt, \rschema \vdash E_1 \transexpr E'_1 \quad
\sdt, \rschema \vdash E_2 \transexpr E'_2 \\
\end{array}}
{\sdt, \rschema \vdash E_1 \arithop E_2 \transexpr E'_1 \arithop E'_2}
{\textrm{(E-Arith)}}

\end{array}
\]
\end{mdframed}
\vspace{-10pt}
\caption{Translation rules for expressions.}
\vspace{-10pt}
\label{fig:trans-expr}
\end{figure}

\begin{figure}[!t]
\scriptsize
\begin{mdframed}
\[
\begin{array}{c}

\irulelabel
{\begin{array}{c}
\end{array}}
{\sdt, \rschema \vdash \top \transpred \top}
{\textrm{(P-True)}}

\irulelabel
{\begin{array}{c}
\end{array}}
{\sdt, \rschema \vdash \bot \transpred \bot}
{\textrm{(P-False)}}

\irulelabel
{\begin{array}{c}
\sdt, \rschema \vdash E_1 \transexpr E_1' \quad
\sdt, \rschema \vdash E_2 \transexpr E_2' \\
\end{array}}
{\sdt, \rschema \vdash E_1 \logicop E_2 \transpred E_1' \logicop E_2'}
{\textrm{(P-Logic)}}

\\ \ \\

\irulelabel
{\begin{array}{c}
\sdt, \rschema \vdash E \transexpr E' \\
\end{array}}
{\sdt, \rschema \vdash \sfIsNull(E) \transpred \ttIsNull(E')}
{\textrm{(P-IsNull)}}

\irulelabel
{\begin{array}{c}
\sdt, \rschema \vdash E \transexpr E' \\
\end{array}}
{\sdt, \rschema \vdash E \in \overline{v} \transpred E' \in \overline{v}}
{\textrm{(P-In)}}

\irulelabel
{\begin{array}{c}
\sdt, \rschema \vdash \pred \transpred \pred' \\
\end{array}}
{\sdt, \rschema \vdash \neg \pred \transpred \neg \pred'}
{\textrm{(P-Not)}}

\\ \ \\

\irulelabel
{\begin{array}{c}
\circ \in \set{\land, \lor} \\
\sdt, \rschema \vdash \pred_1 \transpred \pred'_1 \\
\sdt, \rschema \vdash \pred_2 \transpred \pred'_2 \\
\end{array}}
{\sdt, \rschema \vdash \pred_1 \circ \pred_2 \transpred \pred_1' \circ \pred_2'}
{\textrm{(P-AndOr)}}

\irulelabel
{\begin{array}{c}
(X_1, l_1) = \text{head}(PP) \quad
(X_2, l_2) = \text{last}(PP) \\
l_1(\ldots) \to R_{l_1}(\ldots) \in \sdt \quad
l_2(\ldots) \to R_{l_2}(\ldots) \in \sdt \\
\sdt, \rschema \vdash PP \transpattern \Xset, Q \quad
\overline{a} = [\PK(R_{l_1}), \PK(R_{l_2})] \\
\end{array}}
{\sdt, \rschema \vdash \text{Exists}(PP) \transpred \overline{a} \in \Pi_{\overline{a}}(Q)}
{\textrm{(P-Exists)}}

\end{array}
\]
\end{mdframed}
\vspace{-10pt}
\caption{Translation rules for predicates.} 
\vspace{-10pt}
\label{fig:trans-pred}
\end{figure}

As shown in Figure~\ref{fig:trans-expr} and Figure~\ref{fig:trans-pred}, most of the rules for translating expressions and predicates are straightforward. The only rule that is tricky is \textrm{P-Exists} for translating $\sfExists$ predicates. Specifically, given a predicate $\sfExists(PP)$ in Cypher, the \textrm{P-Exists} rule first translates the pattern $PP$ to a SQL query $Q$ and then translates $\sfExists(PP)$ to an \boldIn predicate in SQL, namely $\overline{a} \in \proj_{\overline{a}}(Q)$. Recall from the Cypher semantics that $\sfExists(PP)$ checks if there exists a match of pattern $PP$ in the graph, so the translated SQL query uses an \boldIn predicate to perform the existence checking. Here, the attribute list $\overline{a}$ is obtained from the primary keys of relations corresponding to the first and last nodes in the path pattern. The predicate $\overline{a} \in \proj_{\overline{a}}(Q)$ checks that $\overline{a}$ used in the parent query is in the set of results returned by the subquery $Q$.

\begin{example}
Given the standard database transformer $\sdt$ in Example~\ref{ex:sdt}, let us consider the Cypher predicate $\sfExists([(\texttt{n}, \texttt{EMP}), (\texttt{e}, \texttt{WORK\_AT}, \rightarrow), (\texttt{m}, \texttt{DEPT})])$.
Based on the \textrm{P-Exists} rule, it is translated to a SQL predicate
\[
[\texttt{n.id}, \texttt{m.dnum}] \in \proj_{[\texttt{n.id}, \texttt{m.dnum}]}(\rename_{\mathtt{n}}(\mathtt{emp}) \ijoin_{\mathtt{n.id} = \mathtt{e.SRC}} \rename_{\mathtt{e}}(\mathtt{work\_at}) \ijoin_{\mathtt{e.TGT} = \mathtt{m.dnum}} \rename_{\mathtt{m}}(\mathtt{dept}))
\]
\end{example}

\section{An Equivalent Cypher Query of Motivating Example} \label{sec:correct-cypher}

As discussed in Section~\ref{sec:overview}, the Cypher query in Figure~\ref{fig:example-cypher-query} is proven non-equivalent to the SQL query in Figure~\ref{fig:example-sql-query}. The issue arises because the original query filters \texttt{s:SENTENCE} using the \boldWith clause, but this does not enforce filtering at the level of individual matches in the second \boldMatch clause. As a result, the second \boldMatch explores more paths than intended, leading to double counting. The correct Cypher query shown below moves the filtering into an \boldExists condition, ensuring each \texttt{SENTENCE} node is considered only if it satisfies the first pattern, which prevents the discrepancy.

\begin{center}
\small
\begin{tabular}{l}
\boldMatch \texttt{(s:SENTENCE)<-[r3:SP]-(p2:PA)<-[r4:CS]-[c2:CONCEPT]} \\
\boldWhere \boldExists \{ \boldMatch \texttt{(c1:CONCEPT \{CID: 1\})-[r1:CS]->(p1:PA)-[r2:SP]->(s:SENTENCE)} \} \\
\boldReturn \texttt{c2.CID}, \boldCount(*)
\end{tabular}
\end{center}

\section{Qualitative Analysis of Manually-Written Buggy Queries} \label{sec:qual-analysis}

We manually inspected the buggy queries in Section~\ref{sec:eval-bmc} to gain intuition about common problems.
At a high level, there are several root causes of bugs we identified in the manually-written benchmarks. These include:
\begin{enumerate}[leftmargin=*]
\item Using nested match instead of an existential pattern, as shown in Section~\ref{sec:overview}.

\item Misusing path patterns for \boldOptMatch.
For example, let us look at the following benchmark adapted from the Neo4j tutorial\footnote{\url{https://neo4j.com/docs/getting-started/cypher-intro/cypher-sql}}. We removed the OrderBy clauses and renamed node and edge labels to avoid confusion. Given the SQL query
\begin{center}
\footnotesize
\begin{tabular}{l}
\boldSelect \texttt{P.ProductName,} \boldSum \texttt{(OD.UnitPrice * OD.Quantity)} \boldAs \texttt{Volume} \boldFrom \texttt{Customers} \boldAs \texttt{C} \\
\hspace{2em} \boldLeftJoin \texttt{Orders} \boldAs \texttt{O} \boldOn \texttt{C.CustomerID = O.CustomerID} \\
\hspace{2em} \underline{\boldLeftJoin \texttt{OrderDetails} \boldAs \texttt{OD} \boldOn \texttt{O.OrderID = OD.OrderID}} \\
\hspace{2em} \underline{\boldLeftJoin \texttt{Products} \boldAs \texttt{P} \boldOn \texttt{OD.ProductID = P.ProductID}} \\
\boldWhere \texttt{C.CompanyName = 'Drachenblut Delikatessen'} \boldGroupBy \texttt{P.ProductName} \\
\end{tabular}
\end{center}

The tutorial provides a corresponding Cypher query
\begin{center}
\footnotesize
\begin{tabular}{l}
\boldMatch \texttt{(C:Customer \{CompanyName:'Drachenblut Delikatessen'\})} \\
\boldOptMatch \underline{\texttt{(P:Product)<-[OD:OrderDetails]-(:Order)<-[:Purchased]-(C)}} \\
\boldReturn \texttt{P.ProductName,} \boldSum \texttt{(OD.UnitPrice * OD.Quantity}) \boldAs \texttt{Volume} \\
\end{tabular}
\end{center}

However, this Cypher query is not equivalent to the SQL query. The underlined optional match clause uses a path pattern containing three nodes and two edges, which does not consider the missing \texttt{Product} as queried by a left outer join in SQL. The difference is more straightforward to understand if we inspect the transpiled SQL query

{
\hspace{-12pt}
\footnotesize
\begin{tabular}{l}
\boldWith \texttt{T0} \boldAs (\boldSelect \texttt{C.CustomerID} \boldAs \texttt{C_CustomerID} \\
\hspace{4em} \boldFrom \texttt{Customer} \boldAs \texttt{C} \boldWhere \texttt{C.CompanyName = 'Drachenblut Delikatessen'}), \\
\hspace{2em}  \texttt{T1} \boldAs (\boldSelect \texttt{P.ProductName} \boldAs \texttt{P_ProductName}, \texttt{OD.Quantity} \boldAs \texttt{OD_Quantity}, \\
\hspace{6em}  \texttt{OD.UnitPrice} \boldAs \texttt{OD_UnitPrice}, \texttt{C.CustomerID} \boldAs \texttt{C_CustomerID}, \\
\hspace{4em}  \boldFrom \underline{\texttt{Product} \boldAs \texttt{P} \boldJoin \texttt{OrderDetails} \boldAs \texttt{OD} \boldJoin \texttt{Purchased} \boldAs \texttt{PD} \boldJoin \texttt{Customer}} \\
\hspace{6em}  \boldAs \texttt{C} \boldOn \texttt{P.ProductID = OD.TGT} \boldAnd \texttt{OD.SRC = PD.TGT} \boldAnd \texttt{PD.SRC = C.CustomerID}), \\
\hspace{2em} \texttt{T2} \boldAs (\boldSelect \texttt{T1.P_ProductName} \boldAs \texttt{ProductName}, \texttt{T1.O_Quantity} \boldAs \texttt{Quantity}, \texttt{T1.O_UnitPrice} \boldAs \texttt{UnitPrice} \\
\hspace{4em} \boldFrom \texttt{T0} \boldLeftJoin \texttt{T1} \boldOn \texttt{T0.C_CustomerID = T1.C_CustomerID}) \\
\boldSelect \texttt{ProductName}, \boldSum(\texttt{UnitPrice * Quantity}) \boldAs \texttt{Volume} \boldFrom \texttt{T2} \boldGroupBy \texttt{ProductName}; \\
\end{tabular}
}

Here, the underlined clause uses inner joins instead of left outer joins in the original SQL query.

\item Misusing different nodes or edges with the same label.
For instance, consider the following SQL query collected from the \verieqlset dataset
\begin{center}
\footnotesize
\begin{tabular}{l}
\boldSelect \texttt{t0.EmpNo}, \texttt{t0.DeptNo}, \texttt{t1.DeptNo} \boldAs \texttt{DeptNo0}, \boldFrom ( \\
\hspace{4em} \boldSelect \texttt{EmpNo}, \texttt{EName}, \texttt{DeptNo}, \texttt{DeptNo} + \texttt{EmpNo} \boldAs \texttt{f9} \boldFrom \texttt{EMP} \boldWhere \texttt{EmpNo = 10} \\
\hspace{2em} ) \boldAs \texttt{t0} \boldJoin (\boldSelect \texttt{DeptNo}, \texttt{Name}, \texttt{DeptNo} + 5 \boldAs \texttt{f2} \boldFrom \texttt{DEPT} ) \boldAs \texttt{t1} \\
\boldOn \texttt{t0.EmpNo = t1.DeptNo} \boldAnd \texttt{t0.f9 = t1.f2}
\end{tabular}
\end{center}
A graduate student wrote a corresponding Cypher query as follows:
\begin{center}
\footnotesize
\begin{tabular}{l}
\boldMatch \texttt{(t0:EMP \{EmpNo: 10\})-[:WORK_AT]->(t1:DEPT)}  \\
\boldWhere \texttt{t1.DeptNo + t0.EmpNo = t1.DeptNo + 5} \\
\boldReturn \texttt{t0.EmpNo}, \texttt{t1.DeptNo}, \texttt{t1.DeptNo} \boldAs \texttt{DeptNo0}
\end{tabular}
\end{center}

However, this Cypher query is not equivalent to the SQL query because it fails to introduce a \texttt{DEPT} node different from \texttt{t1}. Consequently, the query matches an incorrect \texttt{t1} and uses an incorrect filtering predicate.
\tool disproves the equivalence between the SQL and Cypher queries and provides a counterexample as shown in Figure~\ref{fig:example-qa3}. The counterexample contains a relational and a graph database holding the same data in different formats. Running the SQL query on the relational database in Figure~\ref{fig:example-qa3-relation} returns a table with one row \texttt{(10,5,10)}, whereas running the Cypher query on the graph database in Figure~\ref{fig:example-qa3-graph} returns empty result.

\begin{figure}[t]
\small
\centering

\begin{subfigure}{0.5\linewidth}
\centering
\begin{tabular}{|c|c|c|}
\multicolumn{3}{c}{\texttt{EMP}} \\
\hline
\texttt{EmpNo} & \texttt{EName} & \texttt{DeptNo} \\
\hline
0 & A & 10 \\
\hline
10 & B & 5 \\
\hline
\end{tabular}
~
\begin{tabular}{|c|c|c|}
\multicolumn{2}{c}{\texttt{DEPT}} \\
\hline
\texttt{DeptNo} & \texttt{DName} \\
\hline
5 & C \\
\hline
10 & D \\
\hline
\end{tabular}
\caption{Relational database.}
\label{fig:example-qa3-relation}
\end{subfigure}
\begin{subfigure}{0.4\linewidth}
\centering
\normalsize
\begin{tikzpicture}[main/.style = {draw, circle}] 
\node[circle, draw=black] (1) at (0, 1){A};
\node[circle, draw=black] (2) at (0, 0){B};
\node[circle, draw=black, dashed] (3) at (2.3, 1){D};
\node[circle, draw=black, dashed] (4) at (2.3, 0){C};

\draw[->] (1) -- (3) node [midway, above, sloped] (cs1) {};
\draw[->] (2) -- (4) node [midway, below, sloped] (cs2) {};
\end{tikzpicture} 
\vspace{-5pt}
\caption{Graph database. Solid nodes have label \texttt{EMP}. Dashed nodes have label \texttt{DEPT}.}
\label{fig:example-qa3-graph}
\end{subfigure}

\vspace{-5pt}
\caption{A counterexample that demonstrates non-equivalence.}
\vspace{-10pt}
\label{fig:example-qa3}
\end{figure}
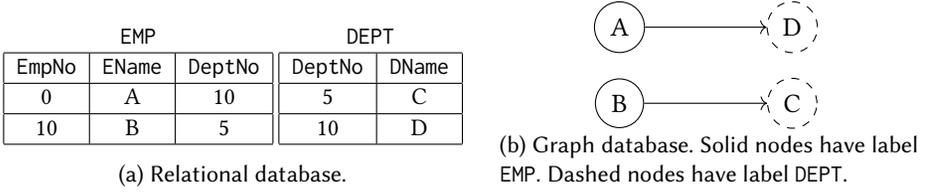

\end{enumerate}
\section{Comparing \tool's Transpiler with \opentranspiler} \label{sec:eval-oct}



In this section, we compare the transpiler of \tool with \opentranspiler~\cite{oct-web25} on all 410 benchmarks.
The results of running \opentranspiler on these benchmarks are presented in Table~\ref{tab:eval-oct-results}. As shown in the table, 284 out of 410 (69\%) queries fall outside the Cypher fragment supported by \opentranspiler. Among the remaining 126 benchmarks, \opentranspiler ill-translate 2 Cypher queries into syntactically invalid SQL queries. Furthermore, a manual inspection reveals that 2 other Cypher queries are transpiled into semantically incorrect SQL queries. These findings demonstrate that \opentranspiler serves as a ``best-effort'' tool without any soundness guarantees about the result of transpilation and cannot preserve semantic equivalence during transpilation.
In contrast, \tool can correctly transpile all 410 Cypher queries used in our evaluation.


\begin{table}[!t]
\small
\centering
\caption{Transpilation Results of \opentranspiler.}
\label{tab:eval-oct-results}
\vspace{-10pt}
\begin{tabular}{c|c|c|c|c|c}
\toprule
Dataset & \# & \# Unsupported & \# SynErr  & \# Incorrect  & \# Correct \\       
\midrule
\stackoverflow & 12 & 8 & 0 & 0 & 4 \\
\tutorial & 26 & 14 & 0 & 0 & 12 \\
\academic & 7 & 6 & 0 & 0 & 1 \\
\verieqlset & 60 & 44 & 1 & 0 & 15 \\
\mediatorset & 100 & 100 & 0 & 0 & 0 \\
\gpt & 205 & 112 & 1 & 2 & 90 \\
\hline
Total & 410 & 284 & 2 & 2 & 122 \\
\bottomrule
\end{tabular}
\vspace{-10pt}
\end{table}

Here are three sample Cypher queries that \opentranspiler cannot transpile correctly.

\begin{enumerate}[leftmargin=*]
\item A Cypher query that cannot be translated by \opentranspiler.

\begin{center}
\normalsize
\begin{tabular}{l}
\boldMatch \texttt{(N:PERSON)} \\
\boldWith \texttt{N.ZIPCODE} \boldAs \texttt{ZIP}, \highlight{\texttt{\textbf{Count}}(*)} \boldAs \texttt{POPULATION} \\
\boldWhere \texttt{POPULATION} > 3 \\
\boldReturn \texttt{ZIP}, \texttt{POPULATION} \\
\end{tabular}
\end{center}
\opentranspiler does not support expressions such as \texttt{\textbf{Count}}(*) or \texttt{\textbf{Avg}}(*).

\item A Cypher query where \opentranspiler generates ill-formed SQL query with undefined variables.

\begin{center}
\normalsize
\begin{tabular}{l}
\boldMatch \texttt{(X:USR)}, \texttt{(U:PIC)}, \texttt{(V:PIC}, \texttt{(W:PIC)} \\
\boldWhere \texttt{X.USRUID = U.PICUID} \boldAnd \texttt{X.USRUID = V.PICUID} \boldAnd \texttt{X.USRUID = W.PICUID} \boldAnd \\
\hspace{2em} \texttt{W.PRICSIZE = V.PRICSIZE} \boldAnd \texttt{U} \boldIsNotNull \boldAnd \texttt{V} \boldIsNull \\
\boldReturn \boldDistinct \texttt{X.USRUID} \boldAs \texttt{XID}, \texttt{X.USRNAME} \boldAs \texttt{XNAME} \\
\end{tabular}
\end{center}
\opentranspiler translates the above Cypher query to the following:
\begin{center}
\normalsize
\begin{tabular}{l}
\boldSelect \boldDistinct \texttt{__X_USRUID} \boldAs \texttt{XID}, \ldots~ \boldFrom ( \\
\hspace{2em} \ldots~ \boldWhere ~\ldots~ \boldAnd ((\highlight{\texttt{U}}) \boldIsNotNull) \boldAnd ~\ldots \\
) \boldAs \texttt{_proj}

\end{tabular}
\end{center}

However, this SQL query uses an undefined table name ``\texttt{U}''.

\item A Cypher query where \opentranspiler produces a semantically incorrect transpilation result.

\begin{center}
\normalsize
\begin{tabular}{l}
\boldMatch \texttt{(e1:EMPLOYEES)} \\
\boldOptMatch \texttt{(e2:EMPLOYEEUNI)-[:IS]->(e1)} \\
\boldReturn \texttt{e2.UNIQUE_ID}, \texttt{e1.NAME}
\end{tabular}
\end{center}
\opentranspiler translates the above Cypher query to the following:
\begin{center}
\normalsize
\begin{tabular}{l}
\boldSelect \texttt{UNIQUE_ID}, \texttt{NAME} \boldFrom ( \boldSelect * \boldFrom \texttt{EMPLOYEEUNI} \boldJoin \texttt{IS} \boldOn \ldots ) \boldAs \texttt{T1} \\
\boldLeftJoin EMPLOYEES \boldAs \texttt{T2} \boldOn \ldots \\
\end{tabular}
\end{center}
The correct SQL query should be as follows:
\begin{center}
\normalsize
\begin{tabular}{l}
\boldSelect \texttt{UNIQUE_ID}, \texttt{NAME} \boldFrom ( \boldSelect * \boldFrom \texttt{EMPLOYEEUNI} \boldJoin \texttt{IS} \boldOn \ldots ) \boldAs \texttt{T1} \\
\highlight{\boldRightJoin} EMPLOYEES \boldAs \texttt{T2} \boldOn \ldots \\
\end{tabular}
\end{center}

\end{enumerate}

\section{Proofs} \label{sec:proofs}

\begin{theorem}[Soundness of translation (\ref{lem:soundness})]~\label{lem:query}
Let $\gschema$ be a graph schema and $Q$ be a Cypher query over $\gschema$. Let $\rschema = (S, \integrity)$ be the induced relational schema of $\gschema$, and $\sdt$ be the standard \trans from $\gschema$ to $\rschema$. If $\sdt, \rschema \vdash Q \transquery Q'$, then $Q'$ is equivalent to $Q$ modulo $\sdt$, i.e., $Q \simeq_{\sdt} Q'$.
\end{theorem}

\begin{proof}
Prove by structural induction on $Q$. 

\begin{enumerate}

\item Base case: $Q = \sfReturn(C, \overline{E}, \overline{k})$.

By the inductive hypothesis, we have $\sdt(\denot{C}_{G}) = \denot{Q'}_{R}$ where $\sdt, \rschema \vdash C \transclause \Xset, Q'$. 

Suppose that $|\overline{E}| = m$ and $|\denot{Q}_{R}| = |\denot{C}_{G}| = n$.
Let us discuss the predicate $\neg \sfHasAgg(\overline{E})$ in two cases.

\begin{enumerate}
\item[(a)]
If $\neg \sfHasAgg(\overline{E}) = \top$, then, for any expression $e \in \overline{E}$, $\mathsf{IsAgg}(e) = \bot$.
By Figure~\ref{fig:cypher-semantics}, we know 
\[
\begin{array}{rcl}
\denot{\sfReturn(C, \overline{E}, \overline{k})}_{G}
&=& \sfmap(\lambda g. \sfmap(\lambda (E, k). (k, \denot{E}_{G, [g]}), \sfzip(\overline{E}, \overline{k})), \denot{C}_{G}) \\
&=& [ \\
&& \hspace{3em} [(k_1, \denot{E_1}_{G, [g_1]}), \ldots, (k_m, \denot{E_m}_{G, [g_1]})], \\
&& \hspace{3em} \ldots, \\
&& \hspace{3em} [(k_1, \denot{E_1}_{G, [g_n]}), \ldots, (k_m, \denot{E_m}_{G, [g_n]})] \\
&& ] \\
\end{array}
\]
Also, by the SQL semantics~\cite{verieql-oopsla24}, we know 
\[
\begin{array}{rcl}
\denot{\proj_{\rho_{\overline{k}}(\overline{E'})}(Q')}_{R}
&=& \sfmap(\lambda t. \sfmap(\lambda (E', k). (k, \denot{E'}_{R, [t]}), \sfzip(\overline{E'}, \overline{k})), \denot{Q'}_{R}) \\
&=& [ \\
&& \hspace{3em} (k_1: \denot{E_1'}_{R, [t_1]}, \ldots, k_m: \denot{E_m'}_{R, [t_1]}), \\
&& \hspace{3em} \ldots, \\
&& \hspace{3em} (k_1: \denot{E_1'}_{R, [t_n]}, \ldots, k_m: \denot{E_m'}_{R, [t_n]}), \\
&& ] \\
\end{array}
\]
There exists a bijective mapping $\pi = \{E_j \mapsto E_j' ~|~ 1 \leq j \leq m \}$ and $\pi^{-1} = \{E_j' \mapsto E_j ~|~ 1 \leq j \leq m \}$. 
Let $g_i \in \denot{C}_{G}$ be a graph derived from $G$ and $t_i \in \denot{Q'}_{R}$ be a tuple derived from $R$ s.t. $\sdt(g_i) = t_i$.
Then, we know $\denot{E_j}_{G, [g_i]} = \denot{E_j'}_{R, [t_i]}$ by Lemma~\ref{lem:expression}.
This case is proved because the following formula holds
\[
\forall G, R. \ G \conform \gschema \land R \conform \rschema \land G \sim_{\sdt} R \ \Rightarrow  \ \sdt(\denot{C}_{G}) = \denot{Q}_{R}
\]

\item[(b)]
If $\neg \sfHasAgg(\overline{E}) = \bot$, then, for some expression $E \in \overline{E}$, $\mathsf{IsAgg}(E) = \top$.

Suppose that $n_v = |\sfDedup(\overline{V})| \leq n = |\overline{V}|$. Then, by the definition of $\mathsf{Groups}$ in Figure~\ref{fig:cypher-semantics}, we know
\[
\begin{array}{rcl}
\denot{\sfReturn(C, \overline{E}, \overline{k})}_{G}
&=& \sfmap(\lambda gs. \sfmap(\lambda (E, k). (k, \denot{E}_{G, gs}), \sfzip(\overline{E}, \overline{k})), \mathsf{Groups})) \\
\end{array}
\]
where $\mathsf{Groups} = [gs_1, \ldots, gs_{n_v}]$ and $gs_i$ is a list of graphs consisting of at least one graph derived from $G$.

Also, by the SQL semantics~\cite{verieql-oopsla24}, we known
\[
\begin{array}{rcl}
\denot{\ttGroupBy(Q', \overline{A'}, \rho_{\overline{k}}(\overline{E'}), \top)}
&=& \sfmap(\lambda xs. \sfmap(\lambda (E, k). (k, \denot{E'}_{R, xs}), \sfzip(\overline{E}, \overline{k})), \\
&& \hspace{3em} \sfFilter(\lambda xs. \denot{\top}_{R, xs} = \top, \mathsf{Ts}))) \\
&=& \sfmap(\lambda xs. \sfmap(\lambda (E, k). (k, \denot{E'}_{R, xs}), \sfzip(\overline{E}, \overline{k})), \mathsf{Ts})) \\
\end{array}
\]
where $\mathsf{Ts}$ denote the grouping list derived by $Q'$.

There exists a bijective mapping $\pi = \{E_j \mapsto E_j' ~|~ 1 \leq j \leq m \}$ and $\pi^{-1} = \{E_j' \mapsto E_j ~|~ 1 \leq j \leq m \}$. Then, $\pi(\overline{E}) = \pi([E_1, \ldots, E_m]) = [E_1', \ldots, E_m'] = \overline{E'}$ and $\pi(\overline{A}) = \sfFilter(\lambda E. \neg \sfHasAgg(E), \pi(\overline{E})) = \sfFilter(\lambda E. \neg \sfHasAgg(E), \overline{E}) = \overline{A'}$.
Since $\sdt(\denot{C}_{G}) = \denot{Q}_{R}$ and $\pi(\overline{A}) = \overline{A'}$, we always find $n_v$ unique results while evaluating $\overline{A}$ on $\denot{C}_{G}$ and $\overline{A'}$ on $\denot{Q}_{R}$ (namely, $|\mathsf{Groups}| = |\mathsf{Ts}| = n_v$). 
Let $\mathsf{Ts} = [ts_1, \ldots, ts_m]$ and $ts_i$ corresponds to $gs_i \in \mathsf{Groups}$. Then $\sdt(\mathsf{Groups}) = \mathsf{Ts}$ and, therefore,  $\sdt(\denot{\sfReturn(C, \overline{E}, \overline{k})}_{G}) = \denot{\ttGroupBy(Q', \overline{A'}, \rho_{\overline{k}}(\overline{E'}), \top)}$.

\end{enumerate}

Thus, Theorem~\ref{lem:query} in this case is proved.

\item Inductive case: $Q = \sfOrderBy(R_1, k, b)$.

Specifically, $R_1$ denotes the $\sfReturn$ statement of Cypher.
By the inductive hypothesis, we have $\denot{R_1}_{G} = \denot{Q'}_{R}$ where $\sdt, \rschema \vdash R_1 \transquery Q'$.
Also, by Lemma~\ref{lem:expression}, we have $\denot{k}_{G, gs} = \denot{k}_{R, \tuples}$ where $\sdt, \rschema \vdash k \transexpr k$ and $k$ is property key.
By Figure~\ref{fig:cypher-semantics}, we know
\[
\begin{array}{rcl}
\denot{\sfOrderBy(R_1, k, b)}_{G} 
&=& \sfFoldl(\lambda xs. \lambda \_. xs \doubleplus [\mathsf{MinTuple}(k, b, \denot{R_1}_{G} - xs)], [], \denot{R_1}_{G}) \\
\end{array}
\]
Also, by the SQL semantics~\cite{verieql-oopsla24}, we known
\[
\begin{array}{rcl}
\denot{\ttOrderBy(Q', k, b)}_{G}
&=& \sfFoldl(\lambda xs. \lambda \_. xs \doubleplus [\mathtt{MinTuple}(k, b, \denot{Q'}_{\tuples} - xs)], [], \denot{Q'}_{\tuples}) \\
\end{array}
\]
Since $\sfOrderBy$ in Cypher and $\ttOrderBy$ in SQL will not change the order of columns, we can reuse the bijective mapping $\pi$ from the inductive hypothesis for this case.
Further, as $\sdt(\denot{R_1}_{G}) = \denot{Q'}_{\tuples}$ and the functions  $\mathsf{MinTuple}$ and $\mathtt{MinTuple}$ are equivalent by their definitions, $\sdt(\denot{\sfOrderBy(R_1, k, b)}_{G}) = \denot{\ttOrderBy(Q', k, b)}_{G}$

\item Inductive case: $Q = \mathsf{Union}(Q_1, Q_2)$.

By the inductive hypothesis, we have $\denot{Q_1}_{G, gs} \simeq_{\sdt} \denot{Q_1'}_{R}$ and $\denot{Q_2}_{G, gs} \simeq_{\sdt} \denot{Q_2'}_{R}$ where $\sdt, \rschema \vdash Q_1 \transquery Q_1'$ and $\sdt, \rschema \vdash Q_2 \transquery Q_2'$.
Similarly, since the order of columns will not be changed in this case, we can inherit the bijective mapping $\pi$ from the inductive hypothesis for this case.
By Figure~\ref{fig:cypher-semantics} and the definitions of $\sfDedup$ and $\texttt{Distinct}$, $\denot{\mathsf{Union}(Q_1, Q_2)}_{G} = \sfDedup(\denot{\mathsf{UnionAll}(Q_1, Q_2)}_{G}) \simeq_{\sdt} \texttt{Distinct}(\denot{Q_1' \uplus Q_2'}_{R}) = \denot{Q_1' \cup Q_2'}_{R}$.

\item Inductive case: $Q = \mathsf{UnionAll}(Q_1, Q_2)$.

By the inductive hypothesis, we have $\denot{Q_1}_{G, gs} \simeq_{\sdt} \denot{Q_1'}_{R}$ and $\denot{Q_2}_{G, gs} \simeq_{\sdt} \denot{Q_2'}_{R}$ where $\sdt, \rschema \vdash Q_1 \transquery Q_1'$ and $\sdt, \rschema \vdash Q_2 \transquery Q_2'$.
Also, since the order of columns will not be changed in this case, we can inherit the bijective mapping $\pi$ from the inductive hypothesis for this case.
By Figure~\ref{fig:cypher-semantics}, $\denot{\mathsf{UnionAll}(Q_1, Q_2)}_{G} = \denot{Q_1}_{G} \doubleplus \denot{Q_2}_{G} \simeq_{\sdt} \denot{Q_1'}_{R} \doubleplus \denot{Q_2'}_{R} = \denot{Q_1' \uplus Q_2'}_{R}$.

\end{enumerate}

\end{proof}

\begin{lemma}~\label{lem:clause}
Let $\gschema$ be a graph schema and $C$ be a Cypher clause. 
Let $\rschema$ be a relational schema and $\sdt$ be the standard \trans from $\gschema$ to $\rschema$. 
If $\sdt, \rschema \vdash C \transclause \Xset, Q$, it holds that~\footnote{We slightly abuse $\sdt$ over lists to apply $\sdt$ to each element in list $\denot{C}_{G}$.}
\[
\forall G, R. \ G \conform \gschema \land R \conform \rschema \land G \sim_{\sdt} R \ \Rightarrow  \ \sdt(\denot{C}_{G}) = \denot{Q}_{R}
\]
\end{lemma}

\begin{proof}
Prove by structural induction on $C$. 

\begin{enumerate}
\item Base case: $C = \sfMatch(\pattern, \pred)$.

By Lemma~\ref{lem:pattern}, we have $\sdt(\denot{\pattern}_{G}) = \denot{Q}_{R}$ where $\sdt, \rschema \vdash \pattern \transpattern \Xset, Q$.
Also, by Lemma~\ref{lem:predicate}, we have $\denot{\pred}_{G, gs} = \denot{\pred'}_{R, \tuples}$ where $\sdt, \rschema \vdash \pred \transpred \pred'$, $gs$ is a list of graph derived from $G$ and $\tuples$ is a list of tuples derived from $R$.
By Figure~\ref{fig:cypher-semantics}, 
\[
\begin{array}{rcl}
\sdt(\denot{\sfMatch(\pattern, \pred)}_{G}) 
&=& \sdt(\sfFilter(\lambda g. \denot{\pred}_{G, [g]} = \top, \denot{\pattern}_{G})) \\
&=& \sfFilter(\lambda t. \denot{\pred'}_{\sdt(G), [t]} = \top, \sdt(\denot{\pattern}_{G})) \\
&=& \sfFilter(\lambda t. \denot{\pred'}_{R, [t]} = \top, \denot{Q}_{R}) \\
&=& \denot{\filter_{\pred'}(Q)}_{R} \\
\end{array}
\]

\item Inductive case: $C = \ttMatch(C_1, \pattern, \pred)$.

By the inductive hypothesis, we have $\sdt(\denot{C_1}_{G}) = \denot{Q_1}_{R}$ where $\sdt, \rschema \vdash C_1 \transclause \Xset_1, Q_1$.
By Lemma~\ref{lem:pattern}, we have $\sdt(\denot{\pattern}_{G}) = \denot{Q_2}_{R}$ where $\sdt, \rschema \vdash \pattern \transpattern \Xset_2, Q_2$.
Also, by Lemma~\ref{lem:predicate}, we have $\denot{\pred}_{G, gs} = \denot{\pred'}_{R, \tuples}$ where $\sdt, \rschema \vdash \pred \transpred \pred'$, $gs$ is a list of graph derived from $G$ and $\tuples$ is a list of tuples derived from $R$. 

By the definition of $\sfMerge$ for graph, we know, for any two graphs $g_1$ and $g_2$, there is no duplicate node, edge, property key in $\sfMerge(g_1, g_2)$.
Similarly, by the definition of $\ttMerge$ for tuple, we know the function $\ttMerge$ has the same functionality as $\sfMerge$.
Therefore, $\sdt(\sfMerge(g_1, g_2)) = \ttMerge(t_1, t_2)$ iff $\sdt(g_1) = t_1 \land \sdt(g_2) = t_2$.

For any Cypher clause $C_1$ and pattern $\pattern$, let us discuss $\Xset_1 \cap \Xset_2$ in two cases.

\begin{enumerate}
\item[(a)] If $\Xset_1 \cap \Xset_2 = \emptyset$, then we know $\pred'' = \pred'$ by Figure~\ref{fig:trans-clause} and
\[
\centering
\begin{array}{rcl}
\sdt(\denot{\ttMatch(C_1, \pattern, \pred)}_{G}) 
&=& \sdt(\sfFilter(\lambda g. \denot{\pred}_{G, [g]} = \top, \sfMap(\lambda g_1. \sfMap(\lambda g_2. \sfMerge(g_1, g_2), \\
&& \hspace{3em} \denot{\pattern}_{G}), \denot{C_1}_{G}))) \\
&=& \sfFilter(\lambda t. \denot{\pred'}_{\sdt(G), [t]} = \top, \sfMap(\lambda t_1. \sfMap(\lambda t_2. \ttMerge(t_1, t_2), \\
&& \hspace{3em} \sdt(\denot{\pattern}_{G})), \sdt(\denot{C_1}_{G}))) \\
&=& \sfFilter(\lambda t. \denot{\pred'}_{R, [t]} = \top, \sfMap(\lambda t_1. \sfMap(\lambda t_2. \ttMerge(t_1, t_2), \\
&& \hspace{3em} \denot{Q_1}_{R}), \denot{Q_2}_{R})) \\
&=& \sfFilter(\lambda t. \denot{\pred'}_{R, [t]} = \top, \denot{Q_1 \times Q_2}_{R}) \\
&=& \sfFilter(\lambda t. \denot{\pred'}_{R, [t]} = \top, \denot{\rho_{T_1}(Q_1) \times \rho_{T_2}(Q_2)}_{R}) \\
&=& \denot{\filter_{\pred'}(\rho_{T_1}(Q_1) \times \rho_{T_2}(Q_2))}_{R} \\
&=& \denot{\rho_{T_1}(Q_1) \ijoin_{\pred'} \rho_{T_2}(Q_2)}_{R} \\
&=& \denot{\rho_{T_1}(Q_1) \ijoin_{\pred''} \rho_{T_2}(Q_2)}_{R} \\
\end{array}
\]
by Figure~\ref{fig:cypher-semantics} where $T_1$ and $T_2$ are fresh names.

\item[(b)] If $\Xset_1 \cap \Xset_2 \neq \emptyset$, then we know $\pred'' = \pred' \land \land_{(X:l) \in \Xset_1 \cap \Xset_2} T_1.\integrity_{pk}(\smapping(l)) = T_2.\integrity_{pk}(\smapping(l))$ for merging two overlapping graphs $T_1$ and $T_2$ by Figure~\ref{fig:trans-clause} and
\[
\begin{array}{rcl}
\sdt(\denot{\ttMatch(C_1, \pattern, \pred)}_{G}) 
&=& \sdt(\sfFilter(\lambda g. \denot{\pred}_{G, [g]} = \top, \sfMap(\lambda g_1. \sfMap(\lambda g_2. \sfMerge(g_1, g_2), \\
&& \hspace{3em} \denot{\pattern}_{G}), \denot{C_1}_{G}))) \\
&=& \sfFilter(\lambda t. \denot{\pred'}_{\sdt(G), [g]} = \top, \sfMap(\lambda t_1. \sfMap(\lambda t_2. \ttMerge(t_1, t_2), \\
&& \hspace{3em} \sdt(\denot{\pattern}_{G})), \sdt(\denot{C_1}_{G}))) \\
&=& \sfFilter(\lambda t. \denot{\pred'}_{R, [g]} = \top, \sfMap(\lambda t_1. \sfMap(\lambda t_2. \ttMerge(t_1, t_2), \\
&& \hspace{3em} \denot{Q_1}_{R}), \denot{Q_2}_{R})) \\
&=& \sfFilter(\lambda t. \denot{\pred'}_{R, [t]} = \top, \\
&& \hspace{3em} \denot{\rho_{T_1}(Q_1) \ijoin_{\land_{(X:l) \in \Xset_1 \cap \Xset_2} T_1.\integrity_{pk}(\smapping(l)) = T_2.\integrity_{pk}(\smapping(l))} \rho_{T_2}(Q_2)}_{R}) \\
&=& \denot{\filter_{\pred'}(\rho_{T_1}(Q_1) \ijoin_{\land_{(X:l) \in \Xset_1 \cap \Xset_2} T_1.\integrity_{pk}(\smapping(l)) = T_2.\integrity_{pk}(\smapping(l))} \rho_{T_2}(Q_2))}_{R} \\
&=& \denot{\rho_{T_1}(Q_1) \ijoin_{\pred' \land \land_{(X:l) \in \Xset_1 \cap \Xset_2} T_1.\integrity_{pk}(\smapping(l)) = T_2.\integrity_{pk}(\smapping(l))} \rho_{T_2}(Q_2)}_{R} \\
&=& \denot{\rho_{T_1}(Q_1) \ijoin_{\pred''} \rho_{T_2}(Q_2)}_{R} \\
\end{array}
\]
by Figure~\ref{fig:cypher-semantics} where $T_1$ and $T_2$ are fresh names.

\end{enumerate}

Thus, Lemma~\ref{lem:clause} in this case is proved.

\item Inductive case: $C = \sfOptMatch(C_1, \pattern, \pred)$.

By the inductive hypothesis, we have $\sdt(\denot{C_1}_{G}) = \denot{Q_1}_{R}$ where $\sdt, \rschema \vdash C_1 \transclause \Xset_1, Q_1$.
By Lemma~\ref{lem:pattern}, we have $\sdt(\denot{\pattern}_{G}) = \denot{Q_2}_{R}$ where $\sdt, \rschema \vdash \pattern \transpattern \Xset_2, Q_2$.
Also, by Lemma~\ref{lem:predicate}, we have $\denot{\pred}_{G, gs} = \denot{\pred'}_{R, \tuples}$ where $\sdt, \rschema \vdash \pred \transpred \pred'$, $gs$ is a list of graph derived from $G$ and $\tuples$ is a list of tuples derived from $R$. 
Further, by the definition of $\sfNullify$, we know $\sfNullify(\sfHead(\denot{\pattern}_{G}))$ return a graph which is isomorphic to the first graph of $\denot{\pattern}_{G}$ but with $\sfNull$ values for all its property keys. Therefore, $\sdt(\sfNullify(\sfHead(\denot{\pattern}_{G})))$ is equal to a unique tuple $T_{\sfNull}$ used in the semantics of SQL~\cite{verieql-oopsla24}.

For any graph $g_i \in \denot{C}_{G}$, let us discuss the predicate $|v_1(g_i)| = 0$ in two cases.

\begin{enumerate}
\item[(a)]
If the predicate $|v_1(g_i)| = 0$ holds, then, for any graph $g_j' \in \denot{\pattern}_{G}$, $\denot{\pred}_{G, [\sfMerge(g_i, g_j')]} = \bot$ does not hold.
By Figure~\ref{fig:cypher-semantics}, 
\[
\begin{array}{rcl}
\sdt(\sfite(|v_1(g_i)| = 0, v_2(g_i), v_1(g_i))) 
&=& \sdt([v_2(g_i)]) \\
&=& [\ttMerge(\sdt(g_i), \sdt(\sfNullify(\sfHead(\denot{\pattern}_{G}))) )] \\
&=& [\ttMerge(t_i, T_{\sfNull})] \\
&=& [t_i] \times [T_{\sfNull}] \\
\end{array}
\]
where $t_i \in \tuples$ is a corresponding tuple of $g_i$ s.t. $\sdt(g_i) = t_i$.

\item[(b)]
If the predicate $|v_1(g_i)| = 0$ does not hold, then, there exists a graph $g_j' \in \denot{\pattern}_{G}$ s.t. $\denot{\pred}_{G, [\sfMerge(g_i, g_j')]} = \top$.
By Figure~\ref{fig:cypher-semantics}, 
\[
\begin{array}{rcl}
\sdt(\sfite(|v_1(g_i)| = 0, v_2(g_i), v_1(g_i))) 
&=& \sdt(v_1(g_i)) \\
&=& \sfFilter(\lambda x. \denot{\pred'}_{\sdt(G), [x]} = \top, \sfmap(\lambda t'. \ttMerge(t_i, t'), \\
&& \hspace{3em} \sfFilter(\lambda t''. t \cap t''  \neq \emptyset, \sdt(\denot{\pattern}_{G}))) \\
&=& \sfFilter(\lambda x. \denot{\pred'}_{R, [x]} = \top, \sfmap(\lambda t'. \ttMerge(t_i, t'), \\
&& \hspace{3em} \sfFilter(\lambda t''. t \cap t'' \neq \emptyset, \denot{Q_2}_{R})) \\
&=& \filter_{\pred'}([t_i] \ijoin_{\pred''} \denot{Q_2}_{R}) \\
&=& [t_i] \ijoin_{\pred' \land \pred''} \denot{Q_2}_{R} \\
\end{array}
\]
where $t_i \in \tuples$ is a corresponding tuple of $g_i$ s.t. $\sdt(g_i) = t_i$, and $\pred'' = \land_{(X:l) \in \Xset_1 \cap \Xset_2} \\ T_1.\integrity_{pk}(\smapping(l)) = T_2.\integrity_{pk}(\smapping(l))$.

\end{enumerate}

Thus, $\sdt(\sfite(|v_1(g_i)| = 0, v_2(g_i), v_1(g_i))) = [t_i] \ljoin_{\pred' \land \land_{(X:l) \in \Xset_1 \cap \Xset_2} T_1.\integrity_{pk}(\smapping(l)) = T_2.\integrity_{pk}(\smapping(l))} \denot{Q_2}_{R}$ and
\[
\begin{array}{rcl}
\sdt(\denot{\sfOptMatch(C_1, \pattern, \pred)}_{G})
&=& \sdt(\sfFoldl(\lambda gs. \lambda g. gs \doubleplus \sfite(|v_1(g_i)| = 0, v_2(g_i), v_1(g_i)), \\
&& \hspace{3em} [], \denot{C_1}_{G})) \\
&=& \sfFoldl(\lambda xs. \lambda x. xs \doubleplus \sdt(\sfite(|v_1(g_i)| = 0, v_2(g_i), v_1(g_i))), \\
&& \hspace{3em} [], \sdt(\denot{C_1}_{G})) \\
&=& \sfFoldl(\lambda xs. \lambda x. xs \doubleplus [t_i] \ljoin_{\pred' \land \pred''} \denot{Q_2}_{R}, [], \denot{Q_1}_{R}) \\
&=& \denot{Q_1 \ljoin_{\pred' \land \pred''} Q_2}_{R} \\
&=& \denot{\rho_{T_1}(Q_1) \ljoin_{\pred' \land \pred''} \rho_{T_2}(Q_2)}_{R} \\
\end{array}
\]

\item Inductive case: $C = \sfWith(C_1, \overline{Y}, \overline{Z})$.

By the inductive hypothesis, we have $\sdt(\denot{C_1}_{G}) = \denot{\proj_{L}(Q)}_{R}$ where $\sdt, \rschema \vdash C_1 \transclause \Xset, \proj_{L}(Q)$ and $L$ is an alias list.
By the definition of $P[\overline{Y} \mapsto \overline{Z}]$ and $T[\overline{Y} \mapsto \overline{Z}]$, we know, for any variable $(Y_i, l_i) \in \Xset$ and $Y_i \in \overline{Y}$, $P(Y_i, k) = P(Z_i, k)$ and $T(Y_i) = T(Z_i)$ where $k$ is a property key from the node or edge of the label $l$, i.e., $k \in \mathsf{Keys}(l)$. Therefore, this process is equal to the renaming operation in SQL, i.e., $\rho_{L[\overline{Z}/\overline{Y}]}$.
By Figure~\ref{fig:cypher-semantics},
\[
\begin{array}{rcl}
\sdt(\denot{\sfWith(C_1, \overline{Y}, \overline{Z})}_{G}) 
&=& \sdt(\sfmap(\lambda (N, E, P, T). (N, E, P[\overline{X} \mapsto \overline{Y}], T[\overline{X} \mapsto \overline{Y}]), \denot{C_1}_{G})) \\
&=& \denot{\proj_{\rho_{L[\overline{Z}/\overline{Y}]}}(Q)}_{R} \\
\end{array}
\]

\end{enumerate}

\end{proof}

\begin{lemma}~\label{lem:pattern}
Let $\gschema$ be a graph schema and $\pattern$ be a Cypher pattern. 
Let $\rschema$ be a relational schema and $\sdt$ be the standard \trans from $\gschema$ to $\rschema$. 
If $\sdt, \rschema \vdash \pattern \transpattern \Xset, Q$, it holds that~\footnote{We slightly abuse $\sdt$ over lists to apply $\sdt$ to each element in list $\denot{\pattern}_{G}$.}
\[
\forall G, R. \ G \conform \gschema \land R \conform \rschema \land G \sim_{\sdt} R \ \Rightarrow  \ \sdt(\denot{\pattern}_{G}) = \denot{Q}_{R}
\]
\end{lemma}

\begin{proof}
Prove by structural induction on $\pattern$. 

\begin{enumerate}
\item Base case: $\pattern = \nodepattern$.

Suppose that $G = (N, E, P, T)$ and $NP = (X, l)$. 
By Figure~\ref{fig:cypher-semantics}, there exists a list of graphs $\denot{\nodepattern}_{G} = \sfMap(\lambda n. (\set{n}, \emptyset, \set{(X, k) \mapsto P(n, k) ~|~ k \in \mathsf{keys}(T(n))}, \set{X \mapsto (l, \mathsf{keys}(T(n)))}, [n \in N ~|~ \lbl(T(n)) = l])$ and a list of tuples $\tuples = \rho_{X}(R(\smapping(l)))$ s.t. $\forall i. \sdt(g_i) = t_i$ where $g_i \in \denot{\nodepattern}_{G}$ is a graph derived from $G$ and $t_i \in \tuples$ is a tuple derived from the relational table $\smapping(l)$.
Therefore,  $\sdt(\denot{\nodepattern}_{G}) = \rho_{X}(R(\smapping(l))) = \denot{\rho_{X}(\smapping(l))}_{R}$.

\item Inductive case: $\pattern = \nodepattern_1, \edgepattern_2, \pattern'$.

Suppose that $\nodepattern = (X_1, l_1)$, $\edgepattern = (X_2, l_2, d_2)$, $\sfHead(\pattern') = \nodepattern_3 = (X_3, l_3)$.
By the inductive hypothesis, we know $\sdt(\denot{\nodepattern_1}_{G}) = \denot{\rho_{X_1}(\smapping(l_1))}_{R}$ and $\sdt(\denot{\nodepattern_3}_{G}) = \denot{\rho_{X_3}(\smapping(l_3))_{R}}$, and $\sdt(\denot{\pattern'}_{G}) = \denot{Q'}_{R}$ where $\sdt, \rschema \vdash \pattern' \transpattern \Xset, Q'$.

Let us first discuss $\mathsf{SubGraphs}(G, [\nodepattern_1, \edgepattern_2, \nodepattern_3])$. By its definition, we are supposed to get a list of subgraphs of $G$ matching the pattern $[\nodepattern_1, \edgepattern_2, \nodepattern_3]$, and $\sdt(\mathsf{SubGraphs}(G, [\nodepattern_1, \\ \edgepattern_2, \nodepattern_3])) = \denot{\rho_{X_1}(\smapping(l_1)) \ijoin_{\pred_1} \rho_{X_2}(\smapping(l_2)) \ijoin_{\pred_2} \rho_{X_3}(\smapping(l_3))}_{R}$ where $\pred_1 = \mathsf{link}(\integrity, \smapping(l_1), \\ \smapping(l_2))$ and $\pred_2 = \mathsf{link}(\integrity, \smapping(l_2), \smapping(l_3))$ by the definition of $\mathsf{link}$ in Figure~\ref{fig:trans-pattern}.

We assume that variables in $\nodepattern$ and $\edgepattern$ are unique. Therefore, for any graphs $g''$ derived from $\nodepattern_1, \edgepattern_2, \nodepattern_3$ and $g'$ derived from $\pattern'$, we are only able to merge them to a graph without isolated nodes and the overlapping node for $g''$ and $g'$ must be the node derived from $\nodepattern_3$; otherwise, we cannot return such a merged graph as an element of $\mathsf{SubGraphs}(G, [\nodepattern_1, \edgepattern_2, \nodepattern_3])$.
Furthermore, for any graph $g' \in \denot{\pattern'}_{G}$, we have 
\[
\begin{array}{l}
\sdt(\sfmap(\lambda g. \sfMerge(g, g'), \sfFilter(\lambda g'', g'' \cap g' \neq \emptyset, \mathsf{Subgraphs}(G, [\nodepattern_1, \edgepattern_2, \nodepattern_3])))) \\
= \sdt(\sfmap(\lambda g. \sfMerge(g, g'), \sfFilter(\lambda g'', g'' \cap g' \neq \emptyset, [g_1'', \ldots, g_n'']))) \\
= \sfmap(\lambda t. \ttMerge(t, t'), \sfFilter(\lambda t''. \denot{\pred_3}_{D, [t''] \times [t']}, [t_1'', \ldots, t_n''])) \\
= \denot{[t_1'', \ldots, t_n''] \ijoin_{\pred_3} [t']}_{R} \\
= \denot{\rho_{X_1}(\smapping(l_1)) \ijoin_{\pred_1} \rho_{X_2}(\smapping(l_2)) \ijoin_{\pred_2} \rho_{X_3}(\smapping(l_3)) \ijoin_{\pred_3} [t']}_{R} \\
= \denot{\rho_{X_1}(\smapping(l_1)) \ijoin_{\pred_1} \rho_{X_2}(\smapping(l_2)) \ijoin_{\pred_2} [t']}_{R} \\
\end{array}
\]
where $\sdt(g') = t'$ and $\pred_3 = \mathsf{link}(\integrity, \smapping(l_3), \smapping(l_3))$.
\[
\begin{array}{rcl}
\sdt(\denot{\nodepattern_1, \edgepattern_2, \pattern'}_{G})
&=& \sdt(\sfFoldl(\lambda gs. \lambda g'. gs \doubleplus \sfmap(\lambda g. \sfMerge(g, g'), \sfFilter(\lambda g''.  \\
&& \hspace{3em} g' \cap g'' \neq \emptyset, \mathsf{Subgraphs}(G, [\nodepattern_1, \edgepattern_2, \nodepattern_3]))), [], \denot{\pattern'}_{G})) \\
&=& \sfFoldl(\lambda xs. \lambda t'. xs \doubleplus \denot{\rho_{X_1}(\smapping(l_1)) \ijoin_{\pred_1} \rho_{X_2}(\smapping(l_2)) \ijoin_{\pred_2} [t']}_{R}, \\
&& \hspace{3em} [], \denot{Q'}_{R}) \\
&=& \denot{\rho_{X_1}(\smapping(l_1)) \ijoin_{\pred_1} \rho_{X_2}(\smapping(l_2)) \ijoin_{\pred_2} Q'}_{R} \\
\end{array}
\]

\end{enumerate}

\end{proof}

\begin{lemma}~\label{lem:expression}
Let $\gschema$ be a graph schema and $E$ be a Cypher expression. 
Let $\rschema$ be a relational schema and $\sdt$ be the standard \trans from $\gschema$ to $\rschema$. 
If $\sdt, \rschema \vdash E \transexpr E'$, it holds that
\[
\forall G, R. \ G \conform \gschema \land R \conform \rschema \land G \sim_{\sdt} R \land (\forall i. \sdt(g_i) = t_i) \ \Rightarrow  \ \denot{E}_{G, gs} = \denot{E'}_{R, \tuples}
\] where $g_i \in gs$ is a graph derived from $G$, and $t_i \in \tuples$ is a tuple derived from $R$.
\end{lemma}

\begin{proof}
Prove by structural induction on $E$. 

\begin{enumerate}
\item Base case: $E = k$.

Lemma~\ref{lem:expression} in this case is proved because $\denot{k}_{G, gs} = \mathsf{lookup}(\sfHead(gs), k) = \mathsf{lookup}\\(\sfHead(\tuples), k) = \denot{k}_{R, \tuples}$ by Figure~\ref{fig:cypher-semantics}.

\item Base case: $E = v$.

Lemma~\ref{lem:expression} in this case is proved because $\denot{v}_{G, gs} = v = \denot{v}_{R, \tuples}$ by Figure~\ref{fig:cypher-semantics}.

\item Inductive case: $E = \sfCast(\pred)$.

By Lemma~\ref{lem:predicate}, we have $\denot{\pred}_{G, gs} = \denot{\pred'}_{R, \tuples}$ where $\sdt, \rschema \vdash \pred \transpred \pred'$.
By Figure~\ref{fig:cypher-semantics}, $\denot{\sfCast(\pred)}_{G, gs} = \sfite(\denot{\pred}_{G, gs} = \sfNull, \sfNull, \sfite(\denot{\pred}_{G, gs} = \top, 1, 0)) = \sfite(\denot{\pred'}_{R, \tuples} = \sfNull, \sfNull, \\ \sfite(\denot{\pred'}_{R, \tuples} = \top, 1, 0)) = \denot{\sfCast(\pred')}_{R, \tuples}$.

\item Inductive case: $E = \sfCount(E_1)$.

By the inductive hypothesis, we have $\denot{E_1}_{G, gs} = \denot{E_1'}_{R, \tuples}$ where $\sdt, \rschema \vdash E_1 \transexpr E_1'$.
By Figure~\ref{fig:cypher-semantics}, 
\[
\begin{array}{rcl}
\denot{\sfCount(E_1)}_{G, gs} &=& \sfite(\land_{g \in gs} \denot{E_1}_{G, [g]} = \sfNull, \sfNull, \\ 
&& \hspace{5em} \sfFoldl(+, 0, \sfMap(\lambda g. \sfite(\denot{E_1}_{G, [g]} = \sfNull, 0, 1), gs))) \\
&=& \sfite(\land_{t \in \tuples} \denot{E_1'}_{R, [t]} = \sfNull, \sfNull, \\ 
&& \hspace{5em} \sfFoldl(+, 0, \sfMap(\lambda t. \sfite(\denot{E_1'}_{R, [t]} = \sfNull, 0, 1), \tuples))) \\
&=& \denot{\ttCount(E_1')}_{R, \tuples} \\
\end{array}
\]

\item Inductive case: $E = \sfSum(E_1)$.

By the inductive hypothesis, we have $\denot{E_1}_{G, gs} = \denot{E_1'}_{R, \tuples}$ where $\sdt, \rschema \vdash E_1 \transexpr E_1'$.
By Figure~\ref{fig:cypher-semantics}, 
\[
\begin{array}{rcl}
\denot{\sfSum(E_1)}_{G, gs} &=& \sfite(\land_{g \in gs} \denot{E_1}_{G, [g]} = \sfNull, \sfNull, \\ 
&& \hspace{5em} \sfFoldl(+, 0, \sfMap(\lambda g. \sfite(\denot{E_1}_{G, [g]} = \sfNull, 0, \denot{E_1}_{G, [g]}), gs))) \\
&=& \sfite(\land_{t \in \tuples} \denot{E_1'}_{R, [t]} = \sfNull, \sfNull, \\ 
&& \hspace{5em} \sfFoldl(+, 0, \sfMap(\lambda t. \sfite(\denot{E_1'}_{R, [t]} = \sfNull, 0, \denot{E_1'}_{R, [t]}), \tuples))) \\
&=& \denot{\ttSum(E_1')}_{R, \tuples} \\
\end{array}
\]

\item Inductive case: $E = \sfAvg(E_1)$.

By the inductive hypothesis, we have $\denot{\sfSum(E_1)}_{G, gs} = \denot{\ttSum(E_1')}_{R, \tuples}$ and $\denot{\sfCount(E_1)}_{G, gs} = \denot{\ttCount(E_1')}_{R, \tuples}$ where $\sdt, \rschema \vdash E \transexpr E'$.
By Figure~\ref{fig:cypher-semantics}, 
\[
\begin{array}{rcl}
\denot{\sfSum(E_1)}_{G, gs} &=& \denot{\sfSum(E_1)}_{G, gs} / \denot{\sfCount(E_1)}_{G, gs} \\
&=& \denot{\ttSum(E_1')}_{R, \tuples} / \denot{\ttCount(E_1')}_{R, \tuples} \\
&=& \denot{\ttSum(E_1')}_{R, \tuples} \\
\end{array}
\]

\item Inductive case: $E = \sfMin(E_1)$.

By the inductive hypothesis, we have $\denot{E_1}_{G, gs} = \denot{E_1'}_{R, \tuples}$ where $\sdt, \rschema \vdash E_1 \transexpr E_1'$.
By Figure~\ref{fig:cypher-semantics}, 
\[
\begin{array}{rcl}
\denot{\sfMin(E_1)}_{G, gs} &=& \sfite(\land_{g \in gs} \denot{E_1}_{G, [g]} = \sfNull, \sfNull, \\ 
&& \hspace{3em} \sfFoldl(\mathsf{min}, + \infty, \sfMap(\lambda g. \sfite(\denot{E_1}_{G, [g]} = \sfNull, + \infty, \denot{E_1}_{G, [g]}), gs))) \\
&=& \sfite(\land_{t \in \tuples} \denot{E_1'}_{R, [t]} = \sfNull, \sfNull, \\ 
&& \hspace{3em} \sfFoldl(\mathsf{min}, + \infty, \sfMap(\lambda t. \sfite(\denot{E_1'}_{R, [t]} = \sfNull, + \infty, \denot{E_1'}_{R, [t]}), \tuples))) \\
&=& \denot{\ttMin(E_1')}_{R, \tuples} \\
\end{array}
\]

\item Inductive case: $E = \sfMax(E_1)$.

By the inductive hypothesis, we have $\denot{E_1}_{G, gs} = \denot{E_1'}_{R, \tuples}$ where $\sdt, \rschema \vdash E_1 \transexpr E_1'$.
By Figure~\ref{fig:cypher-semantics}, 
\[
\begin{array}{rcl}
\denot{\sfMax(E_1)}_{G, gs} &=& \sfite(\land_{g \in gs} \denot{E_1}_{G, [g]} = \sfNull, \sfNull, \\ 
&& \hspace{3em} \sfFoldl(\mathsf{max}, - \infty, \sfMap(\lambda g. \sfite(\denot{E_1}_{G, [g]} = \sfNull, - \infty, \denot{E_1}_{G, [g]}), gs))) \\
&=& \sfite(\land_{t \in \tuples} \denot{E_1'}_{R, [t]} = \sfNull, \sfNull, \\ 
&& \hspace{3em} \sfFoldl(\mathsf{max}, - \infty, \sfMap(\lambda t. \sfite(\denot{E_1'}_{R, [t]} = \sfNull, - \infty, \denot{E_1'}_{R, [t]}), \tuples))) \\
&=& \denot{\ttMax(E_1')}_{R, \tuples} \\
\end{array}
\]

\item Inductive case: $E = E_1 \arithop E_2$.

By the inductive hypothesis, we have $\denot{E_1}_{G, gs} = \denot{E_1'}_{R, \tuples}$ and $\denot{E_2}_{G, gs} = \denot{E_2'}_{R, \tuples}$ where $\sdt, \rschema \vdash E_1 \transexpr E_1'$ and $\sdt, \rschema \vdash E_2 \transexpr E_2'$.
By Figure~\ref{fig:cypher-semantics}, $\denot{E_1 \arithop E_2}_{G, gs} = \denot{E_1}_{G, gs} \arithop \denot{E_2}_{G, gs} = \denot{E_1'}_{R, \tuples} \arithop \denot{E_2'}_{R, \tuples} = \denot{E_1' \arithop E_2'}_{R, \tuples}$.

\end{enumerate}

\end{proof}

\begin{lemma}~\label{lem:predicate}
Let $\gschema$ be a graph schema and $\pred$ be a Cypher predicate. 
Let $\rschema$ be a relational schema and $\sdt$ be the standard \trans from $\gschema$ to $\rschema$. 
If $\sdt, \rschema \vdash \pred \transpred \pred'$, it holds that
\[
\forall G, R. \ G \conform \gschema \land R \conform \rschema \land G \sim_{\sdt} R \land (\forall i. \sdt(g_i) = t_i) \ \Rightarrow  \ \denot{\pred}_{G, gs} = \denot{\pred'}_{R, \tuples}
\]
where $g_i \in gs$ is a graph derived from $G$, and $t_i \in \tuples$ is a tuple derived from $R$.
\end{lemma}

\begin{proof}
Prove by structural induction on $\pred$. 

\begin{enumerate}
\item Base case: $\pred = \top$.

Lemma~\ref{lem:predicate} in this case is proved because $\denot{\top}_{G, gs} = \top = \denot{\top}_{R, \tuples}$ by Figure~\ref{fig:cypher-semantics}.

\item Base case: $\pred = \bot$.

Lemma~\ref{lem:predicate} in this case is proved because $\denot{\bot}_{G, gs} = \bot = \denot{\bot}_{R, \tuples}$ by Figure~\ref{fig:cypher-semantics}.

\item Inductive case: $\pred = E_1 \logicop E_2$.

By Lemma~\ref{lem:expression}, we have $\denot{E_1}_{G, gs} = \denot{E_1'}_{R, \tuples}$ and $\denot{E_2}_{G, gs} = \denot{E_2'}_{R, \tuples}$ where $\sdt, \rschema \vdash E_1 \transexpr E_1'$ and $\sdt, \rschema \vdash E_2 \transexpr E_2'$.
By Figure~\ref{fig:cypher-semantics}, $\denot{E_1 \logicop E_2}_{G, gs} = \denot{E_1}_{G, gs} \logicop \denot{E_2}_{G, gs} = \denot{E_1'}_{R, \tuples} \logicop \denot{E_2'}_{R, \tuples} = \denot{E_1' \logicop E_2'}_{R, \tuples}$.

\item Inductive case: $\pred = \sfIsNull(E)$.

By Lemma~\ref{lem:expression}, we have $\denot{E}_{G, gs} = \denot{E}_{R, \tuples}$ where $\sdt, \rschema \vdash E \transexpr E'$.
By Figure~\ref{fig:cypher-semantics}, $\denot{\sfIsNull(E)}_{G, gs} = \sfite(\denot{E}_{G, gs} = \sfNull, \top, \bot) = \sfite(\denot{E'}_{R, \tuples} = \sfNull, \top, \bot) = \denot{\ttIsNull(E')}_{R, \tuples}$.

\item Inductive case: $\pred = E \in \overline{v}$.

By Lemma~\ref{lem:expression}, we have $\denot{E}_{G, gs} = \denot{E'}_{R, \tuples}$ where $\sdt, \rschema \vdash E \transexpr E'$.
By Figure~\ref{fig:cypher-semantics}, $\denot{E \in \overline{v}}_{G, gs} = \lor_{v \in \overline{v}} \denot{E}_{G, gs} = \lor_{v \in \overline{v}} \denot{E'}_{R, \tuples} = v = \denot{E' \in \overline{v}}_{R, \tuples}$.

\item Inductive case: $\pred = \sfExists(\pattern)$.

By Lemma~\ref{lem:pattern},  we have $\sdt(\denot{\pattern}_{G}) = \denot{Q}_{R}$ where $\sdt, \rschema \vdash \pattern \transpattern \Xset, Q$.
Also, by Lemaa~\ref{lem:expression}, we have $\denot{K}_{G, gs} = \denot{K}_{R, \tuples}$ where $\sdt, \rschema \vdash K \transexpr K$ and $K$ is property key.
Suppose that $\sfHead(\pattern) = (X_1, l_1)$ and $\sfLast(\pattern) = (X_2, l_2)$. 
By the definition of $\mathsf{PKs}$ and $\integrity_{pk}$, $\denot{\overline{K}}_{G, gs} = \denot{[\mathsf{PK}(\sfHead(\pattern)), \mathsf{PK}(\sfLast(\pattern))]}_{G, gs} = \denot{[\integrity_{pk}(\smapping(l_1)), \integrity_{pk}(\smapping(l_2))]}_{R, \tuples} = \denot{\overline{K}}_{R, \tuples}$ iff $\forall i.\sdt(g_i) = t_i$ for any $g_i \in gs$ and $t_i \in \tuples$.
Furthermore, $\sfmap(\lambda g.\denot{\overline{K}}_{G, [g]}, \denot{\pattern}_{G}) = [\denot{\overline{K}}_{G, [g_1]}, \ldots, \denot{\overline{K}}_{G, [g_n]}] = [\denot{\overline{K}}_{R, [t_1]}, \ldots, \denot{\overline{K}}_{R, [t_n]}] = \sfmap(\lambda t. \denot{\overline{K}}_{R, [t]}, \denot{Q}_{R}) = \\ \denot{\proj_{\overline{K}}(Q)}_{R}$ where $n = |\denot{\pattern}_{G}| = |\denot{Q}_{R}|$.

By Figure~\ref{fig:cypher-semantics}, 
\[
\begin{array}{rcl}
\denot{\sfExists(\pattern)}_{G, gs} 
&=& \lor_{g \in \denot{\pattern}_{G}} \land_{K \in \overline{K}}\denot{K}_{G, gs} = \denot{K}_{G, [g]} \\
&=& \lor_{g \in \denot{\pattern}_{G}} \land_{K \in \overline{K}}\denot{K}_{R, \tuples} = \denot{K}_{G, [g]} \\
&=& \lor_{g \in \denot{\pattern}_{G}} \denot{\overline{K} \in \denot{\overline{K}}_{G, [g]}}_{R, \tuples} \\
&=& \lor_{t \in \sfmap(\lambda g.\denot{\overline{K}}_{G, [g]}, \denot{\pattern}_{G})} \denot{\overline{K} \in t}_{R, \tuples} \\
&=& \sfFoldl(\lambda ys. \lambda y. ys \lor \denot{\overline{K} \in y}_{R, \tuples}, \bot, \sfmap(\lambda g. \denot{\overline{K}}_{G, [g]}, \denot{\pattern}_{G})) \\
&=& \sfFoldl(\lambda ys. \lambda y. ys \lor \denot{\overline{K} \in y}_{R, \tuples}, \bot, \denot{\proj_{\overline{K}}(Q)}_{R}) \\
&=& \denot{\overline{K} \in \proj_{\overline{K}}(Q)}_{R, \tuples} \\
\end{array}
\]

\item Inductive case: $\pred = \pred_1 \land \pred_1$.

By the inductive hypothesis, we have $\denot{\pred_1}_{G, gs} = \denot{\pred_1'}_{R, \tuples}$ and $\denot{\pred_2}_{G, gs} = \denot{\pred_2'}_{R, \tuples}$ where $\sdt, \rschema \vdash \pred_1 \transpred \pred_1'$ and $\sdt, \rschema \vdash \pred_2 \transpred \pred_2'$.
By Figure~\ref{fig:cypher-semantics}, $\denot{\pred_1 \land \pred_1}_{G, gs} = \denot{\pred_1}_{G, gs} \land \denot{\pred_2}_{G, gs} = \denot{\pred_1'}_{R, \tuples} \land \denot{\pred_2}_{R, \tuples} = \denot{\pred_1' \land \pred_2'}_{R, \tuples}$.

\item Inductive case: $\pred = \pred_1 \lor \pred_1$.

By the inductive hypothesis, we have $\denot{\pred_1}_{G, gs} = \denot{\pred_1'}_{R, \tuples}$ and $\denot{\pred_2}_{G, gs} = \denot{\pred_2'}_{R, \tuples}$ where $\sdt, \rschema \vdash \pred_1 \transpred \pred_1'$ and $\sdt, \rschema \vdash \pred_2 \transpred \pred_2'$.
By Figure~\ref{fig:cypher-semantics}, $\denot{\pred_1 \lor \pred_1}_{G, gs} = \denot{\pred_1}_{G, gs} \lor \denot{\pred_2}_{G, gs} = \denot{\pred_1'}_{R, \tuples} \lor \denot{\pred_2}_{R, \tuples} = \denot{\pred_1' \lor \pred_2'}_{R, \tuples}$.

\item Inductive case: $\pred = \neg \pred_1$.

By the inductive hypothesis, we have $\denot{\pred_1}_{G, gs} = \denot{\pred_1'}_{R, \tuples}$ where $\sdt, \rschema \vdash \pred_1 \transpred \pred_1'$.
By Figure~\ref{fig:cypher-semantics}, $\denot{\neg \pred_1}_{G, gs} = \neg \denot{\pred_1}_{G, gs} = \neg \denot{\pred_1'}_{R, \tuples} = \denot{\neg \pred_1'}_{R, \tuples}$.

\end{enumerate}

\end{proof}

\begin{theorem}[Completeness of translation (\ref{thm:completeness})]
Let $\gschema$ be a graph schema and $\rschema$ be the induced relational schema of $\gschema$. Given any Cypher query $Q$ over $\gschema$ accepted by the grammar shown in Figure~\ref{fig:cypher-syntax}, there exists a SQL query $Q'$ over $\rschema$ such that $\sdt, \rschema \vdash Q \transquery Q'$.
\end{theorem}

\begin{proof}
Prove by structural induction on $Q$. 

\begin{enumerate}
\item Base case: $Q = \sfReturn(C, \overline{E}, \overline{k})$.

By Lemma~\ref{lem:translation-completeness-clause}, there exist a variable set $\Xset'$ and a SQL query $Q'$ such that $\sdt, \rschema \vdash C \transpred \Xset', Q'$.
Also, by Lemma~\ref{lem:translation-completeness-expr}, for any expression $E \in \overline{E}$ there exists a SQL expression $E'$ such that $\sdt, \rschema \vdash E \transpred E'$.
Let us discuss the predicate $\neg \sfHasAgg(\overline{E})$ in two cases.

\begin{enumerate}
    \item If $\neg \sfHasAgg(\overline{E}) = \top$, then for any expression $E \in \overline{E}$, $\textsf{IsAgg}(E) = \bot$. 
    According to the \textrm{Q-Ret} rule in Figure~\ref{fig:trans-query}, there exists a SQL query $Q'' = \proj_{\rename_{\overline{k}}(\overline{E'})}(Q')$ such that $\sdt, \rschema \vdash \sfReturn(C, \overline{E}, \overline{k}) \transquery Q''$.
    
    \item If $\neg \sfHasAgg(\overline{E}) = \bot$, then there exist some expression $E \in \overline{E}$ s.t. $\textsf{IsAgg}(E) = \top$. 
    Therefore, this $\sfReturn$ statement in Cypher should be translated in $\ttGroupBy$ in SQL. Let $\overline{A} = \sfFilter(\lambda x.\neg \textsf{IsAgg}(x), \overline{E'})$ be the list of SQL expressions containing no aggregation function.
    According to the \textrm{Q-Agg} rule in Figure~\ref{fig:trans-query}, there exists a SQL query $Q'' = \ttGroupBy(Q, \overline{A}, \\ \rename_{\overline{k}}(\overline{E'}), \top)$ such that $\sdt, \rschema \vdash \sfReturn(C, \overline{E}, \overline{k}) \transquery Q''$.
\end{enumerate}

Thus, Theorem~\ref{thm:completeness} in this case is proved.

\item Inductive case: $Q = \sfOrderBy(R_1, k, b)$.

Specifically, $R_1$ denotes the $\sfReturn$ statement of Cypher.
By the inductive hypothesis, we have $\sdt, \rschema \vdash R_1 \transquery Q_1'$. 
Also, by Lemma~\ref{lem:translation-completeness-expr}, there exists a SQL expression $E' = k$ such that $\sdt, \rschema \vdash k \transexpr E'$.
By Lemma~\ref{lem:translation-completeness-pred}, there exists a SQL predicate $b' = b$ such that $\sdt, \rschema \vdash b \transexpr b'$. 
According to the \textrm{Q-OrderBy} rule in Figure~\ref{fig:trans-query}, there exists a SQL query $Q_1'' = \ttOrderBy(Q_1', E', b')$ such that $\sdt, \rschema \vdash \ttOrderBy(R_1, k, b) \transexpr Q_1''$.

\item Inductive case: $Q = \sfUnion(Q_1, Q_2)$.

By the inductive hypothesis, there exist two queries $Q_1'$ and $Q_2'$ such that $\sdt, \rschema \vdash Q_1 \transquery Q_1'$ and $\sdt, \rschema \vdash Q_2 \transquery Q_2'$. 
According to the \textrm{Q-Union} rule in Figure~\ref{fig:trans-query}, there exists a SQL query $Q' = Q_1' \cup Q_2'$ such that $\sdt, \rschema \vdash \sfUnion(Q_1, Q_2) \transexpr Q'$.

\item Inductive case: $Q = \sfUnionAll(Q_1, Q_2)$.

By the inductive hypothesis, there exist two queries $Q_1'$ and $Q_2'$ such that $\sdt, \rschema \vdash Q_1 \transquery Q_1'$ and $\sdt, \rschema \vdash Q_2 \transquery Q_2'$. 
According to the \textrm{Q-Union} rule in Figure~\ref{fig:trans-query}, there exists a SQL query $Q' = Q_1' \uplus Q_2'$ such that $\sdt, \rschema \vdash \sfUnionAll(Q_1, Q_2) \transexpr Q'$.

\end{enumerate}
\end{proof}

\begin{lemma}~\label{lem:translation-completeness-clause}
Let $\gschema$ be a graph schema, $\rschema = (S, \integrity)$ be the induced relational schema of $\gschema$, and $\smapping$ be the corresponding schema mapping. Given any Cypher clause $C$ accepted by the grammar shown in Figure~\ref{fig:cypher-syntax}, there exists a SQL query $Q$ over $\rschema$ such that $\sdt, \rschema \vdash C \transclause \Xset, Q$.
\end{lemma}

\begin{proof}
Prove by structural induction on $C$. 

\begin{enumerate}

\item Base case: $C = \sfMatch(\pattern, \pred)$.

By Lemma~\ref{lem:translation-completeness-pattern}, there exist a variable set $\Xset$ and a SQL query $Q$ such that $\sdt, \rschema \vdash \pattern \transpattern \Xset, Q$.
Also, by Lemma~\ref{lem:translation-completeness-pred}, there exists a SQL predicate $\pred'$ such that $\sdt, \rschema \vdash \pred \transpred \pred'$. 
According to the \textrm{C-Match1} rule in Figure~\ref{fig:trans-clause}, there exists a SQL query $Q' = \filter_{\pred'}(Q)$ such that $\sdt, \rschema \vdash \sfMatch(\pattern, \pred) \transclause \Xset, Q'$.

\item Inductive case: $C = \sfMatch(C_1, \pattern, \pred)$.

By the inductive hypothesis, there exist a variable set $\Xset_1$ and a SQL query $Q_1$ such that $\sdt, \rschema \vdash C_1 \transclause \Xset_1, Q_1$.
By Lemma~\ref{lem:translation-completeness-pattern}, there exist a variable set $\Xset_2$ and a SQL query $Q_2$ such that $\sdt, \rschema \vdash \pattern \transclause \Xset_2, Q_2$.
Also, by Lemma~\ref{lem:translation-completeness-pred}, there exists a SQL predicate $\pred'$ such that $\sdt, \rschema \vdash \pred \transpred \pred'$. 
Let $T_1$ and $T_2$ be two fresh name for the output tables of $Q_1$ and $Q_2$. Then we have a new predicate $\phi'' = \phi' \land \land_{(X: l) \in \Xset_1 \cap \Xset_2}. T_1.\integrity_{pk}(\smapping(l)) = \integrity_{pk}(\smapping(l))$ to join $Q_1$ and $Q_2$.
According to the \textrm{C-Match2} rule in Figure~\ref{fig:trans-clause}, there exists a SQL query $Q' = \filter_{\pred'}(Q)$ such that $\sdt, \rschema \vdash \sfMatch(C_1, \pattern, \pred) \transclause \Xset_1 \cup \Xset_2, \rename_{T_1}(Q_1) \ijoin_{\pred''} \rename_{T_2}(Q_2)$.

\item Inductive case: $C = \sfOptMatch(C_1, \pattern, \pred)$.

By the inductive hypothesis, there exist a variable set $\Xset_1$ and a SQL query $Q_1$ such that $\sdt, \rschema \vdash C_1 \transclause \Xset_1, Q_1$.
By Lemma~\ref{lem:translation-completeness-pattern}, there exist a variable set $\Xset_2$ and a SQL query $Q_2$ such that $\sdt, \rschema \vdash \pattern \transclause \Xset_2, Q_2$.
Also, by Lemma~\ref{lem:translation-completeness-pred}, there exists a SQL predicate $\pred'$ such that $\sdt, \rschema \vdash \pred \transpred \pred'$. 
Let $T_1$ and $T_2$ be two fresh name for the output tables of $Q_1$ and $Q_2$. Then we have a new predicate $\phi'' = \phi' \land \land_{(X: l) \in \Xset_1 \cap \Xset_2}. T_1.\integrity_{pk}(\smapping(l)) = \integrity_{pk}(\smapping(l))$ to join $Q_1$ and $Q_2$.
According to the \textrm{C-Match2} rule in Figure~\ref{fig:trans-clause}, there exists a SQL query $Q' = \filter_{\pred'}(Q)$ such that $\sdt, \rschema \vdash \sfMatch(C_1, \pattern, \pred) \transclause \Xset_1 \cup \Xset_2, \rename_{T_1}(Q_1) \ljoin_{\pred''} \rename_{T_2}(Q_2)$.

\item Inductive case: $C = \sfWith(C_1, \overline{Y}, \overline{Z})$.

By the inductive hypothesis, we have $\sdt, \rschema \vdash C_1 \transclause \Xset_1, \proj_{L}(Q_1)$ where $L$ denotes a list of columns. 
According to the \textrm{C-With} rule in Figure~\ref{fig:trans-clause}, there exists a SQL query $Q_1' = \proj_{\rename_{L[\overline{Z} / \overline{Y}]}}(Q_1)$ such that $\sdt, \rschema \vdash \sfMatch(C_1, \pattern, \pred) \transclause \Xset_1 \setminus \overline{Y} \cup \overline{Z}, Q_1'$.

\end{enumerate}
\end{proof}

\begin{lemma}~\label{lem:translation-completeness-pattern}
Let $\gschema$ be a graph schema, $\rschema = (S, \integrity)$ be the induced relational schema of $\gschema$, and $\smapping$ be the corresponding schema mapping. Given any Cypher pattern $\pattern$ accepted by the grammar shown in Figure~\ref{fig:cypher-syntax}, there exists a SQL query $Q$ over $\rschema$ such that $\sdt, \rschema \vdash \pattern \transpattern \Xset, Q$.
\end{lemma}

\begin{proof}
Prove by structural induction on $\pattern$. 

\begin{enumerate}
\item Base case: $\pattern = \nodepattern = (X, l)$.

According to the \textrm{PT-Node} rule in Figure~\ref{fig:trans-pattern}, there exist a variable set $\Xset = \set{(X, l)}$ and a SQL query $Q = \rename_{X}(\smapping(l))$ such that $\sdt, \rschema \vdash \nodepattern \transpattern \Xset, Q$.

\item Inductive case: $\pattern = \nodepattern_1, \edgepattern_1, \pattern_1$.

Let $\nodepattern_1 = (X_1, l_1)$ be the first node pattern in $\pattern$, $\edgepattern = (X_2, l_2, d_2)$ be the first edge pattern in $\pattern$, and $\sfHead(\pattern_1) = (X_3, l_3)$ be the first node pattern in $\pattern_1$.


By the inductive hypothesis, there exist a variable set $\Xset_1 = \set{(X_1, l_1)}$ and a SQL query $Q_1 = \rename_{X_1}(\smapping(l_1))$ such that $\sdt, \rschema \vdash \nodepattern_1 \transpattern \Xset_1, Q_1$, a variable set $\Xset_2 = \set{(X_2, l_2)}$ and a SQL query $Q_2 = \rename_{X_2}(\smapping(l_2))$ such that $\sdt, \rschema \vdash \nodepattern_2 \transpattern \Xset_2, Q_2$, and a variable set $\Xset_3$ and a SQL query $Q_3$ such that $\sdt, \rschema \vdash \pattern_1 \transpattern \Xset_3, Q_3$.
Let $\pred_1 = \textsf{link}(\integrity, \smapping(l_1), \smapping(l_2))$ be a predicate to join the relations $\smapping(l_1)$ and $\smapping(l_2)$, and $\pred_2 = \textsf{link}(\integrity, \smapping(l_2), \smapping(l_3))$ to join the relations $\smapping(l_2)$ and $\smapping(l_3)$.
According to the \textrm{PT-Path} rule in Figure~\ref{fig:trans-pattern}, there exists a variable set $\Xset = \Xset_1 \cup \Xset_2 \cup \Xset_3$ and a SQL expression $Q = Q_1 \ijoin_{\pred_1} Q_2 \ijoin_{\pred_2} Q_3$ such that $\sdt, \rschema \vdash \pattern \transpattern \Xset,  Q$.


\end{enumerate}
\end{proof}

\begin{lemma}~\label{lem:translation-completeness-expr}
Let $\gschema$ be a graph schema, $\rschema = (S, \integrity)$ be the induced relational schema of $\gschema$, and $\smapping$ be the corresponding schema mapping. Given any Cypher expression $E$ accepted by the grammar shown in Figure~\ref{fig:cypher-syntax}, there exists a SQL expression $E'$ such that $\sdt, \rschema \vdash E \transexpr E'$.
\end{lemma}

\begin{proof}
Prove by structural induction on $E$. 

\begin{enumerate}
\item Base case: $E = k$.

According to the \textrm{E-Prop} rule in Figure~\ref{fig:trans-expr}, there exists a SQL expression $E' = k$ such that $\sdt, \rschema \vdash k \transexpr E'$.

\item Base case: $E = v$.

According to the \textrm{E-Value} rule in Figure~\ref{fig:trans-expr}, there exists a SQL expression $E' = v$ such that $\sdt, \rschema \vdash v \transexpr E'$.

\item Inductive case: $E = \sfCast(\pred)$.

By Lemma~\ref{lem:translation-completeness-pred}, there exists a SQL predicate $\pred'$ such that $\sdt, \rschema \vdash \pred \transpred \pred'$. According to the \textrm{E-Pred} rule in Figure~\ref{fig:trans-expr}, there exists a SQL expression $E' = \texttt{Cast}(\pred')$ such that $\sdt, \rschema \vdash \sfCast(\pred) \transexpr E'$.

\item Inductive case: $E = \textsf{Agg}(E_1)$ where $\textsf{Agg} \in \{\sfCount, \sfMax, \sfMin, \sfAvg, \sfSum\}$.

By the inductive hypothesis, there exists a SQL expression $E_1'$ such that $\sdt, \rschema \vdash E_1 \transexpr E_1'$. Also, for each aggregation function in Cypher there exists an equivalent SQL counterpart.
According to the \textrm{E-Agg} rule in
Figure~\ref{fig:trans-expr}, there exists a SQL expression $E' = \texttt{Agg}(E_1')$ such that $\sdt, \rschema \vdash \textsf{Agg}(E_1) \transexpr E'$.

\item Inductive case: $E = E_1 \arithop E_2$.

By the inductive hypothesis, there exists two SQL expressions $E_1'$ and $E_2'$ such that $\sdt, \rschema \vdash E_1 \transexpr E_1'$ and $\sdt, \rschema \vdash E_2 \transexpr E_2'$.
According to the \textrm{E-Arith} rule in Figure~\ref{fig:trans-expr}, there exists a SQL expression $E' = E_1' \arithop E_2'$ such that $\sdt, \rschema \vdash E_1 \arithop E_2 \transexpr E_1' \arithop E_2'$.

\end{enumerate}

\end{proof}

\begin{lemma}~\label{lem:translation-completeness-pred}
Let $\gschema$ be a graph schema, $\rschema = (S, \integrity)$ be the induced relational schema of $\gschema$, and $\transformer$ be the corresponding \trans. Given any Cypher predicate $\pred$ accepted by the grammar shown in Figure~\ref{fig:cypher-syntax}, there exists a SQL predicate $\pred'$ such that $\sdt, \rschema \vdash \pred \transpred \pred'$.
\end{lemma}

\begin{proof}
Prove by structural induction on $\pred$. 

\begin{enumerate}
\item Base case: $\pred = \top$.

According to the \textrm{P-True} rule in Figure~\ref{fig:trans-pred}, there exists a SQL predicate $\pred' = \top$ such that $\sdt, \rschema \vdash \top \transpred \pred'$.

\item Base case: $\pred = \bot$.

According to the \textrm{P-False} rule in Figure~\ref{fig:trans-pred}, there exists a SQL predicate $\pred' = \bot$ such that $\sdt, \rschema \vdash \bot \transpred \pred'$.

\item Inductive case: $\pred = E_1 \logicop E_2$.

By Lemma~\ref{lem:translation-completeness-expr}, there exists two SQL expressions $E_1'$ and $E_2'$ such that $\sdt, \rschema \vdash E_1 \transexpr E_1'$ and $\sdt, \rschema \vdash E_2 \transexpr E_2'$.
According to the \textrm{P-Logic} rule in Figure~\ref{fig:trans-pred}, there exists a SQL predicate $\pred' = E_1' \logicop E_2'$ such that $\sdt, \rschema \vdash E_1 \logicop E_2 \transpred \pred'$.

\item Inductive case: $\pred = \sfIsNull(E)$.

By Lemma~\ref{lem:translation-completeness-expr}, there exists a SQL expression $E'$ such that $\sdt, \rschema \vdash E \transexpr E'$.
According to the \textrm{P-IsNull} rule in Figure~\ref{fig:trans-pred}, there exists a SQL predicate $\pred' = \ttIsNull(E')$ such that $\sdt, \rschema \vdash \sfIsNull(E) \transpred \pred'$.

\item Inductive case: $\pred = E \in \overline{v}$.

By Lemma~\ref{lem:translation-completeness-expr}, there exists a SQL expression $E'$ such that $\sdt, \rschema \vdash E \transexpr E'$.
According to the \textrm{P-In} rule in Figure~\ref{fig:trans-pred}, there exists a SQL predicate $\pred' = E' \in \overline{v}$ such that $\sdt, \rschema \vdash E \in \overline{v} \transpred \pred'$.

\item Inductive case: $\pred = \sfExists(\pattern)$.


By Lemma~\ref{lem:translation-completeness-pattern}, there exists a SQL query $Q$ such that $\sdt, \rschema \vdash \pattern \transpattern \Xset, Q$.
Let $(X_1, l_1)$ and $(X_2, l_2)$ be the head node and the tail node of the path pattern $\pattern$, and $\overline{a} = [\integrity_{pk}(\smapping(l_1)), $\\ $\integrity_{pk}(\smapping(l_2))]$ be a list of primary keys of those two nodes by the definition of $\smapping$.
According to the \textrm{P-Exist} rule in Figure~\ref{fig:trans-pred}, there exists a SQL predicate $\pred' = \overline{a} \in \proj_{\overline{a}}(Q)$ such that $\sdt, \rschema \vdash \sfExists(\pattern) \transpred \pred'$.

\item Inductive case: $\pred = \pred_1 \circ \pred_2$ where $\circ \in \{\land, \lor\}$.

By the inductive hypothesis, there exist two SQL predicates $\pred_1'$ and $\pred_2'$ such that $\sdt, \rschema \vdash \pred_1 \transpred \pred_1'$ and $\sdt, \rschema \vdash \pred_2 \transpred \pred_2'$.
According to the \textrm{P-AndOr} rule in Figure~\ref{fig:trans-pred}, there exists a SQL predicate $\pred' = \pred_1' \circ \pred_2'$ such that $\sdt, \rschema \vdash \pred_1 \circ \pred_2 \transpred \pred'$.

\item Inductive case: $\pred = \neg \pred_1$.

By the inductive hypothesis, there exists a SQL predicate $\pred_1'$ such that $\sdt, \rschema \vdash \pred_1 \transpred \pred_1'$.
According to the \textrm{P-Not} rule in Figure~\ref{fig:trans-pred}, there exists a SQL predicate $\pred' = \neg \pred_1'$ such that $\sdt, \rschema \vdash \neg \pred_1 \transpred \pred'$.

\end{enumerate}

\end{proof}

\begin{lemma}\label{lem:rst-soundness}
Let $G$ be a graph database instance over graph schema $\gschema$,  $\sdt$ be a standard \trans between $\gschema$ and $\rschema$, $\schema_{R'}$ be the induced relational schema, and $D' \conform \schema_{R'}$ be a relational database such that $D' = \sdt(G)$.
Given a \trans $\transformer$ from $\gschema$ to $\rschema$, a residual \trans $\rdt$ and a relational database instance $D \conform \rschema$ such that $D = \transformer(G)$, it holds that $D' \sim_{\rdt} D$.
\end{lemma}
\begin{proof}

By Algorithm~\ref{algo:infer-rdt}, we know the residual \trans is obtained by a syntactic substitution. Let us discuss the substitution into two cases. 
\begin{enumerate}
\item~\label{lem:rst-soundness:case1} For any node $N(l, a_1, \ldots, a_n)$ with label $l$ and values $a_1, \ldots a_n$ of property keys $K_1, \ldots, K_2$, there exists a corresponding clause $l(a_1, \ldots, a_n) \to l'(a_1, \ldots, a_n)$ in $\sdt$. 
Therefore, $\sigma$ and $\rdt$ are updated by $\sigma \cup \set{l \mapsto l'}$ and $\rdt \cup \transformer[l \mapsto l']$, respectively. 
Let $l(a_1, \ldots, a_n) \to R_{l}(a_1', \ldots a_n')$ be a formula of $\transformer$ that represent the equivalence between the node $N$ of $l$ label and the table of $R_{l}$ relation by the semantics of $\transformer$. Then $l(a_1, \ldots, a_n) \to R_{l}(a_1', \ldots a_n')[l \mapsto l'] = l'(a_1, \ldots, a_n) \to R_{l}(a_1', \ldots a_n')$ is a clause of $\rdt$ that denote the equivalence between the table of $l'$ relation in $\rschema'$ and the table of  of $R_{l}$ relation in  $\rschema$. Formally, for any database instance $D_{l'} \in D'$ that corresponds to nodes $N_{l} \in G$ and database instance $D_{R_{l}} \in D$, $\rdt(D_{l'}) = D_{R_{l}}$

\item For any edge $E(l, s, t, a_1, \ldots, a_n)$ with label $l$ that connects nodes $s$ and $t$ and has property values $a_1, \ldots, a_n$, there exists a corresponding clause $s(\ldots), l(s, t, a_, \ldots, a_n), t(\ldots) \to s'(\ldots), l'(a_1, \ldots, a_n, s, t), t'(\ldots)$ in $\sdt$. Similarly, $\sigma$ and $\rdt$ are updated by $\sigma \cup \set{l \mapsto l'}$ and $\rdt \cup \transformer[l \mapsto l']$, respectively. Let $s(\ldots), l(s, t, a_, \ldots, a_n), t(\ldots) \to R_{s}(\ldots), R_{l}(a_, \ldots, a_n, s, t), \\ R_{t}(\ldots)$ be a formula of $\transformer$ by the semantics of $\transformer$. 
Also, since $s$ and $t$ are two labels for nodes, then by the case~\ref{lem:rst-soundness:case1} we have $s'(\ldots) \to R_{s}(\ldots)$ and $t'(\ldots) \to R_{t}(\ldots)$ from $\rdt$ (namely, $\transformer[s \mapsto s']$ and $\transformer[t \mapsto t']$).
Then $\transformer[s \mapsto s', l \mapsto l', t \mapsto t'] = s'(\ldots), l'(s, t, a_, \ldots, a_n), t'(\ldots) \to R_{s}(\ldots), R_{l}(a_, \ldots, a_n, s, t), R_{t}(\ldots)$. Formally, $\rdt(D_{l'}) = D_{R_l}$ where $D_{l'} \in D'$ and $D_{R_l} \in D$.
\end{enumerate}

Thus, we know $D' \conform \rschema' \land D \conform \rschema \land \rdt(D') = D \Rightarrow D' \sim_{\rdt} D$ where $\rdt$ is derived from Algorithm~\ref{algo:infer-rdt}.

\end{proof}

\begin{lemma}\label{lem:rst-completeness}
Let $\gschema$ be a graph schema, $\schema_{R'}$ be the induced relational schema.
Given a graph query $Q_G$ over $\gschema$ and a SQL query $Q_{R'}$ over $\schema_{R'}$ such that $\sdt, \rschema \vdash Q_G \transquery Q_{R'}$, let $\transformer$ be a \trans from $\gschema$ to $\rschema$, and $\rdt$ be the residual \trans between $\rschema'$ and $\rschema$. If $Q_G \simeq_{\transformer} Q_R$, it holds that $Q_{R'} \simeq_{\rdt} Q_R$.
\end{lemma}

\begin{proof}
By the definition~\ref{def:queryEquivalence}, if $Q_G \simeq_{\transformer} Q_R$ holds, i.e., 
\[
\forall G, D. (G \conform \gschema \land D \conform \rschema \land G \sim_{\transformer} D) \Rightarrow \denot{Q_G}_{G} \equiv \denot{Q_R}_{D}
\]
then $\transformer(G) = D \Leftrightarrow G \sim_{\transformer} D$ holds. Therefore, we can infer that $\denot{Q_G}_{G} \equiv \denot{Q_R}_{D}$ is true.
Furthermore, by Lemma~\ref{lem:rst-soundness}, we know $\rdt(D') = D$ for any $D' \conform \schema_{R'}$ where $D'$ is a relational database instance over $\schema_{R'}$.
Also, by Lemma~\ref{lem:query}, we known $Q_{G} \simeq_{\sdt} Q_{R'} \Rightarrow \denot{Q_{G}}_{G} \equiv \denot{Q_{R'}}_{D'}$ where $\sdt(G) = D'$. Then we can infer that $\denot{Q_R}_{D} \equiv \denot{Q_G}_{G} \equiv \denot{Q_{R'}}_{D'}$.
Therefore, this Lemma is proved by that $Q_{R'} \simeq_{\transformer'} Q_R$ holds, i.e., 
\[
\forall D', D. (D' \conform \schema_{R'} \land D \conform \rschema \land D' \sim_{\transformer'} D) \Rightarrow \denot{Q_{R'}}_{D'} \equiv \denot{Q_R}_{D}
\]
\end{proof}

\begin{theorem}[Soundness(\ref{thm:full-soundness})]
Let $\textsf{CheckSQL}(\rschema, Q, \rschema', Q', \rdt)$ be a sound procedure for equivalence checking of SQL queries $Q, Q'$ over relational schemas $\rschema, \rschema'$ connected by \rdtname $\rdt$.
Given a Cypher query $Q_G$ over graph schema $\gschema$, a SQL query $Q_R$ over relational schema $\rschema$, and their \trans $\transformer$, if $\textsc{CheckEquivalence}(\gschema, Q_G, \rschema, Q_R, \transformer)$ returns $\top$, it holds that $Q_G \simeq_{\transformer} Q_R$.
\end{theorem}

\begin{proof}
By Lemma~\ref{lem:query}, it holds that $Q_G \simeq_{\sdt} Q_{R'}$, i.e., 
\[
\forall G, D'. (G \conform \gschema \land D' \conform \schema_{R'} \land G \sim_{\sdt} D') \Rightarrow \denot{Q_G}_{G} \equiv \denot{Q_{R'}}_{D'}
\]
Since $G \sim_{\sdt} D'$ holds, we have $\denot{Q_G}_{G} \equiv \denot{Q_{R'}}_{D'}$.
Further, if $\textsc{CheckEquivalence}(\gschema, Q_G, \rschema, Q_R, \\ \transformer)$ returns $\top$,  then \textsf{CheckSQL} also returns $\top$, i.e., $\denot{Q_{R'}}_{D'} \equiv \denot{Q_R}_{D}$ where $\rdt(D') = D$ by the definition of $\rdt$. Also, we have $\transformer(G) = D$ by the definition of $\transformer$ and $Q_G \simeq_{\sdt} Q_{R'} \Rightarrow \denot{Q_{G}}_{G} \equiv \denot{Q_{R'}}_{D'}$ by Lemma~\ref{lem:soundness}. Then we have $\denot{Q_{G}}_{G} \equiv \denot{Q_{R'}}_{D'} \equiv \denot{Q_R}_{D}$.

Apparently, $G \sim_{\transformer} D$ holds by our assumption. 
Thus, $Q_G \simeq_{\transformer} Q_R$ is proved, i.e., 
\[
\forall G, D. (G \conform \gschema \land D \conform \rschema \land G \sim_{\transformer} D) \Rightarrow \denot{Q_{G}}_{G} \equiv \denot{Q_{R}}_{D}
\]

\end{proof}

\begin{theorem}[Completeness(\ref{thm:full-completeness})]
Let $\textsf{CheckSQL}(\rschema, Q, \rschema', Q', \rdt)$ be a complete procedure for equivalence checking of SQL queries $Q, Q'$ over schemas $\rschema, \rschema'$ connected by \rdtname $\rdt$.
Given a Cypher query $Q_G$ over graph schema $\gschema$, a SQL query $Q_R$ over relational schema $\rschema$, and their \trans $\transformer$, if $Q_G \simeq_{\transformer} Q_R$, then $\textsc{CheckEquivalence}(\gschema, Q_G, \rschema, Q_R, \transformer)$ returns $\top$.
\end{theorem}

\begin{proof}

If $Q_G \simeq_{\transformer} Q_R$ holds, then $Q_{R'} \simeq_{\rdt} Q_R$ holds by Lemma~\ref{lem:rst-completeness}.
Further, we know that if $Q_{R'} \simeq_{\rdt} Q_R$, then $\textsf{CheckSQL}(\schema_{R'}, Q_{R'}, \rschema, Q, \transformer')$ returns $\top$. Therefore, the procedure \\ $\textsc{CheckSQLEquivalence}(\gschema, Q_G, \rschema, Q_R, \transformer)$ returns $\top$.

\end{proof}

\end{document}